\documentclass[english]{lipics_hacked}
\usepackage[utf8]{inputenc}
\usepackage{amsmath,amsthm, xspace}
\usepackage{algorithm2e}
\usepackage{algorithmic}
\usepackage{float}
\usepackage{verbatim}
\usepackage{tikz}
\usepackage{longtable}
\usetikzlibrary{graphs,quotes,decorations.pathreplacing}

\newcommand{\edge}[1]{%
	\ \begin{tikzpicture}
		\begin{scope}{yshift=0cm}
    \draw[->,line width=0.3pt] (0.0,0) -- (0.6,0);
    \node[ xshift=9, above=-0.04](0,-0.20){  $\escale{$#1$}$ };
    	\end{scope}
    \end{tikzpicture}\
}

\newcommand{\nscale}[1]{\ensuremath{\scalebox{0.8}{#1}}}
\newcommand{\escale}[1]{\ensuremath{\scalebox{0.7}{#1}}}

\begin{document}

\title{The Hardness of Synthesizing Elementary Net Systems from Highly Restricted Inputs }
\author{Christian Rosenke and Ronny Tredup}
\maketitle


\begin{abstract}

Elementary net systems (ENS) are the most fundamental class of Petri nets.
Their synthesis problem has important applications in the design of digital hardware and commercial processes.
Given a labeled transition system (TS) $A$, feasibility is the NP-complete decision problem whether $A$ can be equivalently synthesized into an ENS.
It is well known that $A$ is feasible if and only if it has the event state separation property (ESSP) and the state separation property (SSP).
Recently, these properties have also been studied individually for their practical implications. 
A fast ESSP algorithm, for instance, would allow applications to at least validate the language equivalence of $A$ and a synthesized ENS.
Being able to efficiently decide SSP, on the other hand, could serve as a quick-fail preprocessing mechanism for synthesis.
Although a few tractable subclasses have been found, this paper destroys much of the hope that many practically meaningful input restrictions make feasibility or at least one of ESSP and SSP efficient.
We show that all three problems remain NP-complete even if the input is restricted to linear TSs where every event occurs at most three times or if the input is restricted to TSs where each event occurs at most twice and each state has at most two successor and two predecessor states.
\end{abstract}


\section{Introduction}
\label{sec:Introduction}

Synthesis of elementary net systems (ENS) is a mechanism that transforms a transition system (TS) $A$, which has a global concept of states, into an ENS $N$ having a local concept of states.
More precisely, $A$ is an automaton with states $S$ and arcs labeled by an event set $E$ and $N$ is supposed to be a Petri net, a network with a bipartite node set consisting of the events $E$ and so-called places $P$ which are linked by a flow arc set $F$.
In $N$, every state of $S$, including initial state $s_0$, corresponds to a marking, a specific subset of the places.
For every arc $s \edge{e} s'$ of $A$ the flow $F$ needs to translate the marking $M \subseteq P$ of $s$ into the marking $M' \subseteq P$ of $s'$ in case the event $e$ occurs.
This means that the deallocated places $M \setminus M'$ have to be the inputs $p$ of $e$, that is, where $(p,e) \in F$, and that the occupied places $M' \setminus M$ have to be the outputs $p'$ of $e$, where $(e,p') \in F$.

The act of synthesizing an elementary net systems (ENS) is useful in the description of processes, as for instance in digital hardware and commercial businesses.
ENSs provide a lot of useful properties for the specification, verification, and synthesis of asynchronous or self-timed circuits \cite{C1996,YK1998}.
Their inherent concepts of choice and causality also make ENSs the ideal starting point for process modeling languages like the Business Process Modeling Notation (BPMN) or the Event Driven Process Chains (EPC) \cite{MBD2016}.
Their simpleness is especially useful for the specifications of workflow management systems like \textsc{milano} \cite{AM2000}.

Not every TS $A$ can be transformed into an ENS $N$ as described above.
Hence, the feasibility problem of ENS synthesis is a relevant decision problem.
Traditionally it is approached by the following properties:
The {\em state separation property, SPP}, asserts that different states of $A$ correspond to different markings in $N$.
The {\em event state separation property, ESSP}, makes sure that events $e$ are disabled in markings of $N$ where the corresponding state of $A$ has no outgoing $e$-labeled arc.
By \cite{BBD2015}, $A$ is feasible if and only if it satisfies both properties.
Yet, the SSP and ESSP are worth studying even when considered alone.
TSs having only the ESSP can still be synthesized for ENSs that implement the given behavior but have less states \cite{BBD2015}.
Accordingly, state faithful ENSs that, however, generalize requested behavior, result from synthesizing TSs where only the SSP holds.

Deciding the SSP or ESSP is NP-complete \cite{H1994}.
Naturally, feasibility is NP-complete \cite{BBD1997}, too.
Despite this, practice needs efficient synthesis algorithms for relevant TSs.
Free-Choice Acyclic ENSs, for instance, are applied in workflow modeling and, fortunately, can be synthesized in polynomial time \cite{AM2000}.
There is hope that feasibility or any of the SSP or ESSP could be decided efficiently for many other reasonably large subclasses of TSs with practical significance.

In this paper, we propose two natural fundamental parameters of TSs that, at first glance, should have a positive impact on the synthesizing complexity when being put restrictions on.
We formalize these parameters as follows:
\begin{description}
\item[Event manifoldness]
of a TS $A$ is the maximum amount $k$ of edges in $A$ that can be labeled with the same event.
Accordingly, we shall speak about $k$-fold TSs and the problems $k$-SSP, $k$-ESSP, and $k$-feasibility.  
\item[State degree]
of a TS $A$ is the maximum amount $g$ of in- and, respectively, out-edges at the states of $A$.
Then, the decision problems where input is restricted to so called $g$-grade TSs are refered to as $g$-grade SSP, $g$-grade ESSP, and $g$-grade feasibility.
In case of $g=1$ where, additionally, the initial state has no predecessor we use the term \emph{linear}.
\end{description}
Linear TSs where the cardinality of the event set is bounded by $2$ have already been studied in \cite{BESW2016}.
Their feasibility can be decided by a letter counting algorithm in quadratic time when the sought net is a place/transition Petri net.

Our work shows that even simultaneous and extreme restrictions of event manifoldness and state degree do not help reducing ENS synthesis complexity.
In fact, this paper establishes that all three problems, SSP, ESSP, and feasibility, remain NP-complete for linear $k$-fold TSs if $k\geq 3$ and for $g$-grade $k$-fold TSs if $g\geq 2$ and $k\geq 2$.

On the other hand, $1$-SSP, $1$-ESSP, and $1$-feasibility, that is, when events occur only once, are all trivial and therefore tractable even for unbounded state degree.
In this paper, we show that SSP can also be solved in polynomial time for linear $2$-fold TSs.
Hence, with respect to the two TS parameters, linear $2$-fold TSs are the only non-trivial inputs that are not entirely intractable.
However, it remains for future work to answer if linear $2$-ESSP and linear $2$-feasibility are also efficiently solvable.
We believe that this is the case as the three problems proved to be equivalently challenging for all other examined input restrictions.
Figure \ref{fig:result_overview} shows an overview of our findings.
\begin{figure}[h]
\centering
\newcommand{\tEntry}[3]{
\fill[#3] (#1) +(-0.5,-0.25) rectangle +(0.5,0.25);
\node at (#1) {#2};
}
\begin{tikzpicture}
\tikzstyle{istrivial} = [black!16]
\tikzstyle{inthispaper}=[white]
\tikzstyle{inpreviouspaper}=[white]
\tikzstyle{isopen} = [black!08]
%
\tEntry{-3.125,1}{$g$}{white}
\tEntry{-1.625,1}{problem {\textbackslash} $k$}{white}
\foreach \i in {1,...,4} {\tEntry{\i-1,1}{\i}{white}}
\tEntry{4,1}{\dots}{white}
%
\foreach \j in {1,...,2} {\tEntry{-3.125,-1.5*\j+1.5}{\j}{white}}
\tEntry{-3.125,-2.5}{\vdots}{white}
\tEntry{-1.625,-2.5}{\vdots}{white}
\foreach \j in {1,...,2} {\tEntry{-1.625,-1.5*\j+2}{SSP}{white}}
\foreach \j in {1,...,2} {\tEntry{-1.625,-1.5*\j+1.5}{ESSP}{white}}
\foreach \j in {1,...,2} {\tEntry{-1.625,-1.5*\j+1}{Feasibility}{white}}
%
\foreach \j in {-4,...,1} {\tEntry{0,0.5*\j}{P}{istrivial}}
%
\tEntry{1,0.5}{P}{inthispaper}
\tEntry{1,0}{open}{isopen}
\tEntry{1,-0.5}{open}{isopen}
%
\foreach \j in {-2,...,-4} {
\foreach \i in {1,...,3}{\tEntry{\i,0.5*\j}{NPC}{inthispaper}}
}
%
\foreach \j in {-1,0,1} {
\foreach \i in {2,3}{\tEntry{\i,0.5*\j}{NPC}{inpreviouspaper}}
}
%
\foreach \i in {1,...,4} {\tEntry{\i,-2.5}{\dots}{inthispaper}}
\tEntry{0,-2.5}{\dots}{istrivial}
\foreach \j in {1,...,-1} {\tEntry{4,0.5*\j}{\dots}{inpreviouspaper}}
\foreach \j in {-2,...,-4} {\tEntry{4,0.5*\j}{\dots}{inthispaper}}
%
\tikzstyle{dottedline} = [black!64,dotted]
\tikzstyle{boldline} = [black,line width = 1pt]
\draw[->,boldline] (-3.5,0.75) -- ++(8.75,0);
\foreach \j in {0,-1.5} {
\draw[boldline] (-3.5,\j-0.75) -- ++(8.25,0);
\draw[dottedline] (-2.5,\j-0.25) -- ++(7,0);
\draw[dottedline] (-2.5,\j+0.25) -- ++(7,0);
}
\draw[->,boldline] (-2.75,1.25) -- ++(0,-4.5);
\foreach \i in {0,...,4} {\draw[boldline] (\i-0.5,1.25) -- ++(0,-4.25);}
%
\draw [decorate, decoration={brace,mirror, amplitude=3pt}, yshift=0pt]
(-0.5,-3.25) -- (0.5,-3.25) node [midway, below, yshift=-0.25cm] {trivial cases};
\end{tikzpicture}
\caption{Overview of our results regarding the complexity of SSP, ESSP, and feasibility depending on the parameters event manifoldness $k$ and state degree $g$.}
\label{fig:result_overview}
\end{figure}

This paper is organized as follows:
The following two sections introduce preliminary notions.
Section~\ref{sec:The_Hardness_of_Linear_3Feasibility_and Linear_3ESSP} gives a polynomial time reduction of cubic monotone one-in-three $3$-SAT to linear $3$-ESSP.
Moreover, it shows that, for linear TSs, the ESSP implies the SSP. 
That means, for this class of TSs, feasibility and ESSP are equivalent.
Section~\ref{sec:The_Hardness_of_Linear_3-SSP} provides a polynomial time reduction of linear $3$-ESSP to linear $3$-SSP.
Taking these three results together, we have already shown that feasibility, ESSP, and SSP remain NP-complete on $g$-grade $k$-fold TSs for all $g \geq 1$ and $k \geq 3$.

As $1$-fold TSs are trivial, this just leaves the classes of $2$-fold TSs as unanswered questions. 
Therefore, Section~\ref{sec:The_Hardness_of_2-grade_2-ESSP_2-Feasibility} introduces a polynomial time reduction of cubic monotone one-in-three $3$-SAT to $2$-grade $2$-ESSP.
Our reduction makes sure that if the produced TS instances have the ESSP then they always have the SSP, too.
In this way, ESSP and feasibility become the same problem with respect to the generated instances.
Next, Section~\ref{sec:The_Hardness_of_2-grade_2-SSP} applies the NP-completeness of linear $3$-feasibility from Section~\ref{sec:The_Hardness_of_Linear_3Feasibility_and Linear_3ESSP} and reduces this problem to $2$-grade $2$-SSP in polynomial time.
Consequently, feasibility, ESSP, and SSP are also NP-complete for $g$-grade $2$-fold TSs where $g \geq 2$.

At this point, nearly all cells of the chart in Figure \ref{fig:result_overview} are filled.
Linear $2$-fold TSs are left as the only remaining case of interest.
To attack this last survivor, Section \ref{sec:The_Tractability_of_linear_2-SSP} presents an easy to check property of TSs which is equivalent to linear $2$-SSP, hence, putting this problem into P.
Moreover, this section provides an algorithm that, given a linear 2-fold input TS $A$ and two states $s_i,s_j$ of $A$, computes a separating region for $s_i,s_j$ in quadratic time.

For the sake of readability we transfer some rather technical proofs from Section~\ref{sec:The_Hardness_of_Linear_3Feasibility_and Linear_3ESSP} and Section~\ref{sec:The_Hardness_of_2-grade_2-ESSP_2-Feasibility} to Section~\ref{sec:Secondary_Proofs}.


\section{Preliminaries}
\label{sec:Preliminaries}

In this paper, we deal with (deterministic) \emph{transition systems} (TS) $A = (S,E,\delta,s_0)$ which are determined by finite disjoint sets $S$ of states and $E$ of events, a partial transition function $\delta: S\times E\rightarrow S$, and an initial state $s_0 \in S$.
Usually, we think of $A$ as an edge-labelled directed graph with node set $S$ and where every triple $\delta(s,e)=s'$ is interpreted as an $e$-labelled edge $s\edge{e}s'$ from $s$ to $s'$.
To improve readability, we say that an event $e$ \emph{occurs} at a state $s$ if $\delta(s,e)=s'$ for some state $s'$ and we formally abbreviate this with $s\edge{e}$.
Aside from determinism, TSs are required to be \emph{simple}, that is, there are no multi-edges $s\edge{e}s'$ and $s\edge{e'}s'$, \emph{loop-free}, which rules out instant state recurrence like $s\edge{e}s$, \emph{reachable}, where every state can be reached from $s_0$ by a directed path, and \emph{reduced}, which means free of unused events in $E$.
To describe subclasses of TSs we formally introduce the parameters \emph{event manifoldness} and \emph{state degree}.
We say that $A$ is a \emph{$k$-fold} TS, if every event labels at most $k$ edges, that is, the cardinality of the set $\{(s,s') \mid \delta(s,e)=s'\}$ is at most $k$ for all $e \in E$.
Moreover, $A$ is a \emph{$g$-grade} TS, if the states of $A$ have at most $g$ successors, respectively predecessors, that is, the cardinalities of the sets $\{s' \mid \exists e \in E: \delta(s',e) = s\}$ and $\{s' \mid \exists e \in E: \delta(s,e) = s'\}$ are at most $g$ for all $s \in S$.

We use the term \emph{linear} for $1$-grade TS where the initial state has no predecessor.
A linear TSs can be defined by $A = s_0 \edge{e_1} s_1 \edge{e_2} \dots \edge{e_t}  s_t$, the sequence of states and events starting with $s_0$ and ending at a terminal $s_t$.
The only present arcs $s_{i-1} \edge{e_i} s_i$ link consecutive states for $i \in \{1, \dots, t\}$.
The events are $E = \bigcup_{i = 1}^t \{e_i\}$.
Of  course, defining linear TSs like this assures determinism, simplicity, loop-freeness, reachability and reducedness.

The key concept for following notions are \emph{regions} of TSs.
A set of states $R \subseteq S$ is called a region of a TS $A$ if it permitts a so-called \emph{signature} $sig: E\rightarrow \{-1,0,1\}$.
This means, all edges $s\edge{e}s'$ have to satisfy the equation $R(s')=sig(e)+R(s)$, where, by a little abuse of notation, $R(s)=1$ if $s\in R$ and otherwise $R(s)=0$ for all $s\in S$.
It is easy to see that every region $R$ has a unique signature which is therefor called \emph{the} signature $sig_{R}$ of $R$.
We use $enter_{R}=\{e\mid sig_R(e) = 1\}$, $exit_{R}=\{e\mid sig_R(e) = -1\}$, and $obey_{R}=\{e\mid sig_R(e) = 0\}$ to cumulate events according to their orientation with respect to $R$'s border.
Analogously, we say an event $e$ exits (enters) $R$ when $e\in exit_R$ ($e\in enter_R$) or, otherwise, obeys.
By $\mathcal{R}(A)$ we refer to the set of all regions of $A$.

Based on the previous definition, we say that two states $s,s' \in S$ are \emph{separable} in $A$ if there is a region $R \in \mathcal{R}(A)$ with $R(s) \not= R(s')$.
Moreover, an event $e\in E$ is called \emph{inhibitable} at state $s \in S$ if there is a region $R \in \mathcal{R}(A)$ with either $R(s)=0$ and $sig_R(e)=-1$ or $R(s)=1$ and $sig_R(e)=1$.

Using these notions, we are able to define the conditions of TSs studied in this paper.
Formally, a TS $A$ has the \emph{state separation property} (SSP), if all states of $A$ are pairwise separable.
For the \emph{event state separation property} (ESSP), all events $e$ of $A$ are required to be inhibitable at all states $s$ that have no occurence of $e$, that is, where $s\edge{e}$ is not fulfilled.
We also say that $A$ is \emph{feasible} if and only if it has the SSP and the ESSP.

For convenience, we reuse the names SSP and ESSP for the computational problems of deciding the respective property while we use the new term \emph{feasibility} for the decision if a given TS is feasible.
To study the problems for restricted event manifoldness, we define $k$-SSP, $k$-ESSP, and $k$-feasibility for all naturals $k$ where input is restricted to $k$-fold TSs.
Analoguously, for all naturals $g$, we let $g$-grade SSP, $g$-grade ESSP, and $g$-grade feasibility be the decision problems with input restricted to $g$-grade TSs.
Needless to say that for all, SSP, ESSP, and feasibility, the $g$-grade $k$-problem restricts the input to TSs where the event manifoldness is bounded by $k$ and, at the same time, the state degree is bounded by $g$.

In this context it is noteworthy that the set of $g$-grade $k$-fold TSs is a subclass of $g'$-grade $k'$-fold TSs in case $k \leq k'$ and $g \leq g'$.
Hence, hardness results for a specific class propagate to higher classes of the hierarchy and an efficient algorithm that solves a particular case is also a legitimate solution for lower classes.

\emph{Elementary net systems} (ENS), as the computational objective of synthesis, play a surprisingly secondary role in the argumentation of our paper.
The reason is that we approach the problem rather by the SSP and ESSP which are defined on the basis of TSs.
As a consequence, we postpone their definition to Section \ref{sec:The_Hardness_of_Linear_3Feasibility_and Linear_3ESSP}, the only place in this paper where it is really needed.
For the interested reader, we recommend the monograph of Badouel, Bernardinello, and Darondeau \cite{BBD2015} that gives an excellent introduction to the topic.


\section{Unions, Transition System Containers}
\label{sec:UnionsTSContainers}

NP-completeness proofs are often built on gadgets to modularize arguments.
We introduce the concept of \emph{unions} to allow the independent creation of gadget TSs with specific properties.
In a union they can be grouped together and treated as if they were parts of the same big TS.
A final \emph{joining} operation puts together the independent parts of a union and, in the process, preserves the SSP and feasibility in the joined TS; however not necessarily the ESSP.

Formally, if $A_0=(S_0,E_0,\delta_0,s_0^0), \dots ,A_m=(S_m,E_m,\delta_m,s_0^m)$ are TSs with pairwise disjoint states then we say that $U(A_0, \dots, A_m)$ is their \emph{union}.
By $S(U)$ we denote the entirety of all states in $A_0, \dots, A_n$ and $E(U)$ is the aggregation of all events.
The joint transition function $\Delta^U = \bigcup_{i=0}^m \delta_i$ of $U$ is defined as 
\[
\Delta^U(s,e) = \begin{cases}
\delta_i(s,e), & \text{if } s \in S_i \text{ and } e \in E_i,\\
\text{undefined}, & \text{else}
\end{cases}
\]
for all $s\in S(U)$ and all $e \in E(U)$.
If every event in $E(U)$ occurs at most $k$ times in $U$, not necessarily as part of the same TS, we say that $U$ is $k$-fold.

For their versatility, we allow to build unions recursively as follows:
First of all, we identify any TS $A$ with the union consisting only of $A$, that is, we let $U(A) = A$.   
Next, assume we have a collection of unions $U_1= U(A^1_0,\dots,A^1_{m_1})$ to $U_n=(A^n_0,\dots,A^n_{m_n})$.
Notice that, by the equalization of monadic unions and the respective TSs, this collection possibly includes plain TSs, namely for all $i \in \{1, \dots, n\}$ where $m_i = 0$.
Finally, the complex union $U(U_1,\dots,U_n)$ is simply defined as the flattened union $U(A^1_0,\dots,A^1_{m_1},\dots, A^n_0,\dots,A^n_{m_n})$.
This allows us to nest unions for the sake of a higher degree of modularity.

The merit of unions is to combine independent TSs with specific functions and treat them as if they were part of the same TS.
To this end, we lift the concept of regions as well as the SSP and ESSP to unions as follows:
Let $U = U(A_0, \dots, A_{m})$ be a union of the TSs $A_0, \dots, A_{m}$.
We say that  $R \subseteq S(U)$ is a region of $U$ if and only if it permits a signature $sig_R: E(U) \rightarrow \{-1, 0, 1\}$.
Hence, for all $i \in \{0, \dots, m\}$ the subset $R_i = R \cap S_i$ of $R$, coming from the states $S_i$ of $A_i$, has to be a region of $A_i$ with a signature $sig_{R_i}$ that resembles $sig_R$ on the events $E_i$ of $A_i$.
This means $sig_{R_i}(e) = sig_R(e)$ for all $e \in E_i$.
Then, $U$ has the SSP if and only if for all distinct $s, s' \in S(U)$ we have one of the following conditions:
Either $s$ and $s'$ are part of different TSs or, if both are in $A_i$, there is a region $R$ of $U$ such that $R \cap S_i$, the subregion of $A_i$, separates $s$ and $s'$.
Moreover, $U$ has the ESSP if and only if for all events $e \in E(U)$ and all states $s \in S(U)$ we are in one of the following cases:
Either $A_i$, the TS that contains $s$, fulfills $s\edge{e}$ or there is a region $R$ of $U$ such that $R_i = R \cap S_i$, the subregion of $A_i$, inhibits $e$ at $s$, that is, $R_i(s)=0$ and $sig_{R_i}(e)=-1$ or $R_i(s)=1$ and $sig_{R_i}(e)=1$.
Naturally, $U$ is called feasible if it has both, the SSP and ESSP.

Finally, we need to merge the elements of a given union $U$ into one TS $A$ that preserves the behavior of $U$.
To this end, let $U = U(A_0, \dots, A_{m})$ be a union of the TSs $A_0, \dots, A_{m}$ and let $s^0_0, \dots, s^m_0$ be the respective initial states.
Moreover, for all $i \in \{0, \dots, m-1\}$, let $t^i$ be an arbitrary but fixed discretionary state of $A_i$ that has been chosen to support the respective construction.
Then we define the joining $A(U) = (S(U) \cup Z, E(U) \cup Y, \delta, s^0_0)$ to be the TS with additional connector events and states
\[Y = \bigcup_{i=1}^{m} \{y_1^i,y_2^i\} \text{ and } Z= \{z^1, \dots, z^{m}\}\]
which link together the loose elements of $U$ by
\[\delta(s,e) = \begin{cases}
\Delta^U(s,e), & \text{if } s \in S(U) \text{ and } e \in E(U),\\
z^i, & \text{if } s = t^{i-1} \text{ and } e=y^i_1,\\
s^{i}_0, & \text{if } s = z^i \text{ and } e=y^i_2,\\
\text{undefined,} & \text{else.}
\end{cases}\]
Notice that the construction of $A(U)$ preserves determinism, simplicity, loop-freeness, reachablity, and reducedness.

In this paper, we refer to $t^i$ as the \emph{terminal} state of $A_i$.
In particular, for a linear TS $A_i$ the state $t^i$ is the actual terminal state $s_{t_i}$.
Therefor, joining a union of linear TSs $A_0, \dots, A_{m}$ yields a linear TS
\[
A(U) = A_0 \edge{y_1^1} z^1 \edge{y_2^1} A_1 \edge{y_1^2} z^2 \edge{y_2^2} \dots \edge{y_1^{m}} z^{m} \edge{y_2^{m}} A_{m}.
\]

To justify the replacement of bulky composite TSs with modular unions, we use the following lemma:
\begin{lemma}\label{lem:union_validity}
If $U = U(A_0,\ldots,A_m)$ is a union of TSs $A_0,\ldots,A_m$ then $U$ has the SSP, respectively is feasible, if and only if the joining $A(U)$ has the SSP, respectively is feasible.
\end{lemma}
\begin{proof}
\emph{If}:
Projecting a region separating $s$ and $s'$, respectively inhibiting $e$ at $s$, in $A(U)$ to the component TSs yields a region separating $s$ and $s'$, respectively inhibiting $e$ at $s$ in $U$.
Hence, the (E)SSP of $A(U)$ trivially implies the (E)SSP of $U$.

\emph{Only if}:
A region $R$ of $U$ separating $s$ and $s'$, respectively inhibiting $e$ at $s$, can be completed to become an equivalent region of $A(U)$ by setting
\[R(z^i) = 0, sig_{R}(y^i_1) = -R(t^{i-1}), \text{ and } sig_{R}(y^i_2) = R(s^{i}_0)\]
for all $i \in \{1, \dots, m\}$, where $t^i$ is the transfer state of $A_i$.
Notice that $R$ also inhibits $e$ at all connector states.
Hence, constructing one region for every event as demonstrated inhibits all events at all connector states.

For the (E)SSP of $A(U)$ it is subsequently sufficient to analyze  (event) state separation concerning the connector states (events).
In fact, to find a region $R^i$ separating $z^i$ and another state $s \in S \cup Z$ we select $R^i(z^i) = 1$ and $R^i(s') = 0$ for all other states $s' \in S \cup Z$, and all events become obeying except for $sig_{R^i}(y^i_1) = 1$ and $sig_{R^i}(y^i_2) = -1$.
This means, if $U$ has the SSP then $A(U)$ has the SSP, too.
Moreover, notice that this region $R^i$ also inhibits $y^i_2$ at all required states of $A(U)$.
Hence, constructing a region $R^i$ for every $i \in \{1, \dots, m\}$ inhibits all secondary connector events in $A(U)$.

Let $U$ be feasible, which implies that $A(U)$ has the SSP.
To get that $A(U)$ is feasible, $A(U)$ needs the ESSP, where only the inhibition of $y^i_1$ at all appropriate states is left to show for all $i \in \{1, \dots, m\}$.
Firstly, the set $S_i$ of states from component $A_i$ is a region of $A(U)$ that makes sure that $y^i_1$ is inhibited at all required states in $S \setminus S_i$.
Secondly, the event $y^i_1$ can be inhibited at any state $s \in S_i$ as follows:
As $A(U)$ has the SSP, there is a region $R$ of $A(U)$ with $R(s) = 0$ and $R(t^i) = 1$.
If $R(z^i) = 0$ then $R$ already inhibits $y^i_1$ at $s$.
Otherwise, as each connector event is unique in $A(U)$, we simply get what we  need by removing $z^i$ from $R$ which yields a region $R'=R\setminus \{z^i\}$ which behaves like $R$ except for $R'(z^i) = 0$, $sig_{R'}(y^i_1)=-1$ and $sig_{R'}(y^i_2) = R'(s^{i}_0)$.
\end{proof}
Unfortunately, it is not easy to merge a union $U$ and preserve the ESSP of $U$ for $A(U)$.
Generally, this works only in case $U$ has the SSP.
Therefore, reductions to the unions ESSP problem introduced in this paper need to make sure that generated instances always have the SSP.


\section{The Hardness of Linear 3-Feasibility and Linear 3-ESSP}
\label{sec:The_Hardness_of_Linear_3Feasibility_and Linear_3ESSP}

This section starts with feasibility and the ESSP and shows that both problems remains hard even if the input is restricted to linear $3$-fold TSs:
\begin{theorem}
\label{thm:The_Hardness_of_Linear_3Feasibility_and Linear_3ESSP}
The feasibility problem and the ESSP problem are NP-complete on $g$-grade $k$-fold transition systems for all $g \geq 1$ and all $k \geq 3$.
\end{theorem}
To prove Theorem~\ref{thm:The_Hardness_of_Linear_3Feasibility_and Linear_3ESSP}, we present a polynomial time reduction of a cubic monotonic set $\varphi$ of boolean $3$-clauses to a $3$-fold union $U^\varphi$ of linear TSs such that $\varphi$ has a one-in-three model $M$ if and only if $U^\varphi$ has the ESSP.
As deciding the existence of $M$ is NP-complete \cite{MR2001} we thereby deduce that it is NP-complete to decide if a $3$-fold union of linear TSs has the ESSP.
However, as discussed in Section \ref{sec:UnionsTSContainers}, it is difficult to preserve the ESSP when going from union $U^\varphi$ to joining $A(U^\varphi)$.
To solve this issue, this section shows that unions of linear TSs inherit their SSP from the ESSP.
This means, a union of linear TSs that has the ESSP, automatically gets the SSP and, thus, turns feasible.
As, reversely, a feasible union of TS has the ESSP by definition, the problems feasibility and ESSP become equivalent on unions of linear TSs.
Feasibility, however, is preserved by $A(U^\varphi)$ according to Lemma \ref{lem:union_validity}.
Hence, we are allowed to conclude that NP-completeness of deciding the ESSP, respectively feasibility, passes from $U^\varphi$ to $A(U^\varphi)$.

To start the reduction, we let, in compliance to \cite{MR2001}, $\varphi$ be a set $\{K_0, \dots, K_{m-1}\}$ of $m$ clauses.
Each clause $K_i$ is a subset of exactly three elements from $V(\varphi)$, the set of all boolean variables in $\varphi$.
Moreover, every variable occurs in exactly three clauses of $\varphi$ which implies $|V(\varphi)| = m$.
Hence, we assume a numbering of $V(\varphi)$ such that every clause $K_i = \{X_a, X_b, X_c\}$ can be specified by three distinct indices $a, b, c \in \{0, \dots, m-1\}$.
A one-in-three model $M$ of $\varphi$ is a subset of the variables $V(\varphi)$ such that $\vert M \cap K_i \vert =1$ for all $i\in \{0, \dots, m-1\}$.

The development of union $U^\varphi= U(B, T)$ is divided into the subunions $B$ and $T$.
Basically, $B$ provides all \emph{basic} components for the translation of one-in-three satisfiability to ESSP.
It implements a single key ESSP instance, that is, a key event $k$ which is inhibited at a certain key state of $B$ by a unique region $R^B$.
By design, $R^B$ fixes a negative signature for an event series called the key copies.

In $T$, the fixed signature of key copies is used for the actual \emph{translation} of one-in-three satisfiability to ESSP.
In fact, a region $R$ of $U^\varphi$ inhibiting $k$ at the key state has to extend the unique region $R^B$ of $B$ by a region $R^T$ of $T$ that has a consistent signature for all events shared by $B$ and $T$.
As the only shared events are exactly the key copies which in $T$ occur, $R^T$ inherits their negative signature from $R^B$.

Next, $T$ applies the exiting key copies to make sure that $R^T$ exists if and only if $\varphi$ has a one-in-three model.
To this end, $T$ encodes all variables $X_j \in V(\varphi)$ as an event $X_j \in E(T)$ and every clause $K_i = \{X_a,X_b,X_c\}$ is implemented by a \emph{translator} union $T_i$.
In $T_i$, the three events of $K_i$ are arranged in such a way that, if the key copies exit then exactly one of $X_a,X_b,X_c$ has a positive signature while the other two obey.
As this happens simultaneously for all $i\in \{0, \dots, m-1\}$, there is a region $R^B\cup R^T$ that inhibits $k$ at the key state if and only if exactly one event $X_a,X_b$ or $X_c$ enters in every translator $T_i$ if and only if a variable subset $M \subseteq V(\varphi)$ intersects every clause $K_i = \{X_a,X_b,X_c\}$ in exactly one element if and only if there is a one-in-three model $M$ of $\varphi$.

The behavior of $B$ and $T$ is created by several gadget TSs.
A single \emph{master} $M$ provides the key event and the key state.
Next, there are $6m$ \emph{refreshers} $F_j$ and $6m$ \emph{duplicators} $D_j$ that generate the negative signature of all key copies.
Hence, the union $B = U(M, F_0, \dots, F_{6m-1}, D_0, \dots, D_{6m-1})$ consists of altogether $12m+1$ TSs.
The union $T = U(T_0, \dots, T_{m-1})$ comprises $m$ \emph{translators}, each a union $T_i = U(T_{i,0}, T_{i,1}, T_{i,2})$ of three linear TSs.
To create a complete picture of our reduction, we subsequently introduce the details of all these gadget TSs:
\begin{description}
\item[Master] 
$M$ is a linear TS providing the key event $k$ and the key state $m_6$.
Figure~\ref{fig:Linear3-foldESSP_B}(a) defines $M$ and shows $R^M$ the part of region $R^B$ belonging to $M$.
Notice that the region $R^M$ that inhibits $k$ at $m_6$ is unique.
The so-called \emph{opposites} $o_0,o_1$ enter and the \emph{zero} $z_0$ obeys.
These events initialize the setup of the key copies' negative signature.
The subsequent refreshes and duplicators implement this synchronization in an assembly line fashion. 

More precisely, for all $j \in \{0, \dots, 6m-1\}$ refresher $F_j$ takes the two previously prepared opposites $o_{2j},o_{2j+1}$ and forces a negative signature onto the key copies $k_{3j},k_{3j+1},k_{3j+2}$.
This consumes both remaining applications of $o_{2j},o_{2j+1}$ such that no further opposites are available at this point.
Every refresher, however, is accompanied by a duplicator $D_j$ that eats up the remaining two applications of $k_{3j},k_{3j+1}$ and consumes one application of $k_{3j+2}, z_{2j}$.
As a result, $D_j$ generates the next opposites $o_{2j+2},o_{2j+3}$ and the next zero $z_{2j+2}$, which prepares the work of $F_{j+1}$ and $D_{j+1}$.

Main result of step $j$, however, is the event $k_{3j+2}$ which has a negative signature and one free application.
Therefore, the whole process chain produces $6m$ key copies $k_{2}, k_{5}, \dots, k_{18m-1}$, each of them free to be applied one more time.
\item[Refreshers]
$F_j$ are linear TSs that consume the opposites  $o_{2j},o_{2j+1}$ to generate key copies $k_{3j},k_{3j+1},k_{3j+2}$.
The definition of $F_j$ together with its fraction $R^{F_j}$ of region $R^B$ is given in Figure~\ref{fig:Linear3-foldESSP_B}(b).
Lemma~\ref{lemma:KeyRegionB} shows that, if the input opposites enter, there is no other way to choose $R^{F_j}$ and that this leads to exiting $k_{3j},k_{3j+1},k_{3j+2}$.

\item[Duplicators]
$D_j$ are linear TSs that take $k_{3j},k_{3j+1},k_{3j+2},z_{2j}$ to provide the next opposites $o_{2j+2},o_{2j+3}$ and zero $z_{2j+2}$.
Figure~\ref{fig:Linear3-foldESSP_B}(c) introduces $D_j$ as well as the part $R^{D_j}$ of region $R^B$.
Lemma~\ref{lemma:KeyRegionB} verifies that, if the input key copies exit and $z_{2j}$ obey, $R^{D_j}$ is the only possible region of $D_j$. 
\end{description}
\begin{figure}[t]
\centering
\begin{tikzpicture}[scale = 0.86]
\begin{scope}[yshift =-1cm]
\node at (-1,0) {a)};
\foreach \i in {0,...,8} {\coordinate (\i) at (\i*1.5cm,0);}
\foreach \i in {0, 3, 7} {\fill[black!15, rounded corners] (\i) +(-0.5,-0.3) rectangle +(0.5,0.3);}
\foreach \i in {0,...,8} {\node (m\i) at (\i) {\nscale{$m_{\i}$}};}
\graph { (m0) ->["\escale{$k$}"] (m1) ->["\escale{$z_0$}"] (m2) ->["\escale{$o_0$}"] (m3) ->["\escale{$k$}"] (m4) ->["\escale{$h$}"] (m5)->["\escale{$z_0$}"](m6)->["\escale{$v_1$}"](m7)->["\escale{$k$}"](m8);};
\end{scope}
\begin{scope}[yshift = -2.5cm]
\node at (-1,0) {b)};
\foreach \i in {0,...,7} {\coordinate (\i) at (\i*1.7143cm,0);}
\foreach \i in {1, 3, 5, 7} {\fill[black!15, rounded corners] (\i) +(-0.5,-0.3) rectangle +(0.5,0.3);}
\foreach \i in {0,...,7} {\node (p\i) at (\i) {\nscale{$f_{j,\i}$}};}
\graph { (p0) ->["\escale{$o_{2j}$}"] (p1) ->["\escale{$k_{3j}$}"] (p2) ->["\escale{$o_{2j+1}$}"] (p3) ->["\escale{$k_{3j+1}$}"] (p4) ->["\escale{$o_{2j}$}"] (p5) ->["\escale{$k_{3j+2}$}"] (p6) ->["\escale{$o_{2j+1}$}"] (p7);};
\end{scope}
\begin{scope}[yshift = -5cm]
\node at (-1,0) {c)};
\foreach \i in {0,...,8} {\coordinate (\i) at (\i*1.5cm,0.75);}
\foreach \i in {9,...,14} {\coordinate (\i) at (25.5cm-\i*1.5cm,-0.75);}
\foreach \i in {0, 3, 10, 13} {\fill[black!15, rounded corners] (\i) +(-0.5,-0.3) rectangle +(0.5,0.3);}
\fill[black!15, rounded corners] (6) +(-0.5,-0.3) rectangle +(2,0.3);
\foreach \i in {0,...,14} {\node (d\i) at (\i) {\nscale{$d_{j,\i}$}};}
\graph { (d0) ->["\escale{$k_{3j}$}"] (d1) ->["\escale{$z_{2j}$}"] (d2) ->["\escale{$h_{j}$}"] (d3) ->["\escale{$k_{3j}$}"] (d4) ->["\escale{$z_{2j+1}$}"] (d5) ->["\escale{$h_{j}$}"] (d6) ->["\escale{$z_{2j+2}$}"] (d7) ->["\escale{$k_{3j+1}$}"] (d8) ->["\escale{$z_{2j+1}$}"] (d9) ->["\escale{$o_{2j+2}$}"] (d10) ->["\escale{$k_{3j+1}$}"] (d11) ->["\escale{$z_{2j+2}$}"] (d12) ->["\escale{$o_{2j+3}$}"] (d13)->["\escale{$k_{3j+2}$}"] (d14);};
\end{scope}
\end{tikzpicture}
\caption{
Gadgets of union $B$ together with their parts of the region $R^B$.
States in the respective region are shown with a gray background.
(a) $M$ with $R^{M}$, (b) $F_j$ mit $R^{F_j}$, (c) $D_j$ with $R^{D_j}$.
}
\label{fig:Linear3-foldESSP_B}
\end{figure}
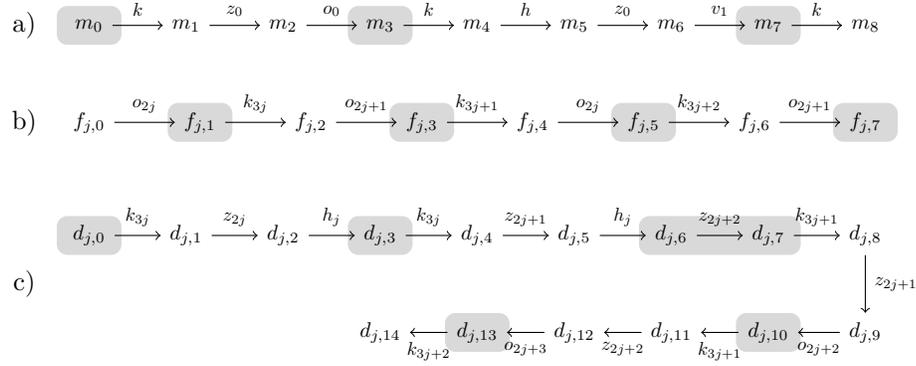
The following lemma summarizes the functionality of union $B$:
\begin{lemma}
\label{lemma:KeyRegionB}
Let $R^{M}=\{m_0,m_3,m_7\}$ and let $R^{F_j}=\{f_{j,1},f_{j,3},f_{j,5},f_{j,7}\}$ and $R^{D_j}=\{d_{j,0},d_{j,3},d_{j,6},d_{j,7},d_{j,10},d_{j,13}\}$ for all $j \in \{0, \dots, 6m-1\}$.
Except for the complement, the set of states 
\[R^B = R^M\cup R^{D_0}\cup \dots R^{D_{6m-1}}\cup R^{F_0}\cup \dots \cup R^{F_{6m-1}}\] 
is the only region of $B$ that inhibits $k$ at $m_6$.
For all $i \in \{0, \dots, 18m-1\}$ the signature of the key copy $k_i$ is exiting, that is, $sig_{R^B}(k_i) = -1$.
\end{lemma}

\begin{proof}
As Figure~\ref{fig:Linear3-foldESSP_B}(a) shows, $R^M$ is a region of $M$ that inhibits $k$ at $m_6$.
Moreover, Figure~\ref{fig:Linear3-foldESSP_B}(b,c) demonstrates that the set $R^{F_j}\cup R^{D_j}$ is a region of $U(F_j, D_j)$ with exiting key copies, obeying zeros obey and entering opposites.
The master $M$ shares events only with $U(F_0,D_0)$, namely $o_0,o_1$ and $z_0$.
The signature of these three events coincides with respect to $R^{M}$ and $R^{F_0}\cup R^{D_0}$.
Furthermore, for all $i \not= j$, the unions $U(F_i, D_i)$ and $U(F_j, D_j)$ do not share events except for the zeros $z_{2j}, z_{2j+1}$ and the opposites $o_{2j}, o_{2j+1}$ in case $i+1 = j$.
By definition, the signature of these events agree with respect to $R^{F_i}\cup R^{D_i}$ and $R^{F_j}\cup R^{D_j}$.
Hence, $R^B$ is a region of $B$ that inhibits $k$ at $m_6$ where all key copies exit.

Next, assume that $R$ is a region of $B$ that inhibits $k$ at $m_6$ where, without lost of generality, $sig_R(k) = -1$ and $R(m_6)=0$.
That $R \cap \{m_0, \dots, m_8\} = R^{M}$ follows simply from $sig_R(k)=-1$ forcing $R(m_0)=R(m_3)=R(m_7)=1$ and $R(m_1)=R(m_4)=R(m_8)=0$, which, in turn, yields $sig_R(z_0)=0$ and by that $R(m_2)=R(m_5)=0$.
By inductively iterating $j$ over the sequence $0, \dots, 6m-1$, we show that $R\cap \{f_{j,0}, \dots, f_{j,7}\}=R^{F_j}$ and $R\cap \{d_{j,0}, \dots, d_{j,14}\}=R^{D_j}$:
For a start, we get $sig_R(o_0) = sig_R(o_1) = 1$ and $sig_R(z_0) = 0$ from the master.
As $o_{2j}, o_{2j+1}$ enter it follows immediately that $R(f_{j,1})=R(f_{j,3})=R(f_{j,5})=R(f_{j,7})=1$ and $R(f_{j,0})=R(f_{j,2})=R(f_{j,4})=R(f_{j,6})=0$ which provides exiting $k_{3j}, k_{3j+1}, k_{3j+2}$.
Together with obeying $z_{2j}$ the just fixed exiting key copies imply $R(d_{j,0})=R(d_{j,3})=R(d_{j,7})=R(d_{j,10})=R(d_{j,13})=1$ and $R(d_{j,1})=R(d_{j,2})=R(d_{j,4})=R(d_{j,8})=R(d_{j,11})=R(d_{j,14})=0$. This also gives $sig_{R}(h_j)=1$, which, in turn, means $R(d_{j,5})=0$ and $R(d_{j,6})=1$ and by that $sig_{R}(z_{2j+1})=sig_{R}(z_{2j+2})=0$.
Using the obeying zeros $R(d_{j,8})=R(d_{j,11})=0$ follows.
\end{proof}

After $B$ has finished the job, we have $6m$ key copies with one free application, namely the events $k_{3j+2}$ for all $j \in \{0, \dots, 6m-1\}$.
In the construction of $T$, we need a sequence of six free key copies for every translator $T_i$. 
Consequently, we assign to $T_i$ the events $k_{18i+2+3\ell}$ for all $\ell \in \{0, \dots, 5\}$.
We continue with the description of our gadgets:
\begin{description}
\item[Translators] 
$T_i = U(T_{i,0},T_{i,1},T_{i,2})$ are unions that represent the clauses $K_i = \{X_a,X_b,X_c\}$ of $\varphi$.
Figure~\ref{fig:Linear3-foldESSP_T} defines the three linear TSs $T_{i,0},T_{i,1},T_{i,2}$ of $T_i$ and presents $R^{T_i}_b$, a possible part of a region of $T$.
In the remainder of this section, we argue that a negative signature of the key copies makes sure that exactly one of the variable events $X_a,X_b,X_c$ gets a positive signature while the other two obey.
In fact, $R^{T_i}_b$ selects the event $X_b$ and there are two other regions $R^{T_i}_a$ and $R^{T_i}_c$ of $T_i$ for the selection of $X_a$, respectively $X_c$.

The TSs $T_{i,1}$ and $T_{i,2}$ create a copy of $X_b$, namely the event $\tilde{X}_b$, and guarantee that the signature of both events cannot be negative.
To achieve this, both TSs surround a sequence, $X_b, p_i$ or, respectively, $\tilde{X}_b, p_i$, with key copies.
As the proxy event $p_i$ and the key copies behave equally in both TSs, $X_b$ and $\tilde{X}_b$ have to be equal, too.
Moreover, the negative signature of the key copies makes sure that their neighboring events $X_b, \tilde{X}_b, p_i$ cannot exit.

The TS $T_{i,0}$ is simply the event sequence $X_a,\tilde{X}_b,X_c$ surrounded by key copies.
Again, the negative signature of key copies prevents a negative signature for their neighboring events, $X_a$ and $X_c$.
The exiting key copies also imply that the signature of the event sequence $X_a,\tilde{X}_b,X_c$ has to add up to one.
Hence, by the equality of $\tilde{X}_b$ and $X_b$, exactly one of $X_a,X_b,x_c$ enters.
\end{description}
\begin{lemma}
\label{lemma:KeyRegionT}
Let $R^T$ be a region of $T$ where all contained key copies exit, that is, where $sig_{R^T}(k_{3j+2}) = -1$ for all $j \in \{0, \dots, 6m-1\}$.
Then $R^T = \bigcup_{i=0}^{m-1} R^{T_i}$ with $R^{T_i}$ being one of  
\begin{align*}
R^{T_i}_c &= \{t_{i,0,0}, t_{i,0,4}, t_{i,1,0}, t_{i,1,3}, t_{i,2,0}, t_{i,2,3}\},\\
R^{T_i}_b &= R^{T_i}_c \cup \{t_{i,0,3}, t_{i,1,2}, t_{i,2,2}\},\\
R^{T_i}_a &= R^{T_i}_c \cup \{t_{i,0,2}, t_{i,0,3}\}
\end{align*}
for all $i \in \{0, \dots, m-1\}$.
Hence, for all clauses $K_i = \{X_a, X_b, X_c\}$ there is exactly one variable event $X \in K_i$ such that $sig_{R^T}(X) = 1$.
\end{lemma}
\begin{proof}
In this proof, we use that a region $R$ of a linear TS $A=s_0\edge{e_1}\dots\edge{e_t}s_t$ fulfills $\sum_{k=i+1}^{j}(e_k)=R(s_j)-R(s_i)$ for any subsequence determined by a starting state $s_i$ and a final state $s_j$.
This trivial fact does not need a proof and is called signature aggregation.

If $R$ is a region of $T_i$ such that the key copies of $T_i$ exit then $R(t_{i,0,1}) = R(t_{i,1,1}) = R(t_{i,2,1}) = R(t_{i,0,5}) = R(t_{i,1,4}) = R(t_{i,2,4}) = 0$ and $R(t_{i,0,0}) = R(t_{i,1,0}) = R(t_{i,2,0}) = R(t_{i,0,4}) = R(t_{i,1,3}) = R(t_{i,2,3}) = 1$.
Hence, all of $sig_R(X_a)$, $sig_R(X_b)$, $sig_R(X_c)$, $sig_R(\tilde{X}_b)$, $sig_R(p_i)$ are not negative.
Moreover, by signature aggregation, $sig_R(X_b)+sig_R(p_i)= 1 =sig_R(\tilde{X}_b)+sig_R(p_i)$ implying $sig_R(X_b) = sig_R(\tilde{X}_b)$.
Again by signature aggregation, we obtain $sig_R(X_a)+sig_R(\tilde{X}_b)+sig_R(X_c) = 1$ which, by the equality of $X_b$ and $\tilde{X}_b$ and the non-negativity of $X_a, X_b, X_c$ means that exactly one of $X_a, X_b, X_c$ enters while the other two obey.
If $sig_R(X_c)=1$ and $sig_R(X_a)=sig_R(X_b)=sig_R(\tilde{X}_b)=0$, we immediately obtain that $R=R^{T_i}_c$.
If $sig_R(X_b)=sig_R(\tilde{X}_b)=1$ and $sig_R(X_a)=sig_R(X_c)=0$ we have $R(t_{i,0,3})=R(t_{i,1,2})=R(t_{i,2,2})=1$, that is, $R=R^{T_i}_b$.
Otherwise, $sig_R(X_a)=1$, $sig_R(X_c)=sig_R(X_b)=sig_R(\tilde{X}_b)=0$ implies $R(t_{i,0,2})=R(t_{i,0,3})=1$ and $R(t_{i,2,2})=0$ meaning $R=R^{T_i}_a$.

Finally, if $R^T$ is a region of $T$ where all key copies exit, by definition and the above argumentation, it has to be a unification of regions $R^{T_i}$ for all subunions $T_i$.
\end{proof}
\begin{figure}
\centering
\begin{tikzpicture}[scale = 0.86]
\begin{scope}[yshift = 0cm]
\node at (-1,0) {a)};
\foreach \i in {0,...,5} {\coordinate (\i) at (\i*1.9cm,0);}
\foreach \i in {0} {\fill[black!15, rounded corners] (\i) +(-0.5,-0.3) rectangle +(0.5,0.3);}
\fill[black!15, rounded corners] (3) +(-0.5,-0.3) rectangle +(2.4,0.4);
\foreach \i in {0,...,5} {\node (p\i) at (\i) {\nscale{$t_{i,0,\i}$}};}
\graph { (p0) ->["\escale{$k_{18i+2}$}"] (p1) ->["\escale{$X_a$}"] (p2) ->["\escale{$\tilde{X}_b$}"] (p3) ->["\escale{$X_c$}"] (p4) ->["\escale{$k_{18i+11}$}"] (p5);};
\end{scope}
\begin{scope}[yshift = -1.5cm]
\node at (-1,0) {b)};
\foreach \i in {0,...,4} {\coordinate (\i) at (\i*1.9cm,0);}
\foreach \i in {0} {\fill[black!15, rounded corners] (\i) +(-0.6,-0.3) rectangle +(0.5,0.3);}
\fill[black!15, rounded corners] (2) +(-0.5,-0.3) rectangle +(2.4,0.4);
\foreach \i in {0,...,4} {\node (p\i) at (\i) {\nscale{$t_{i,1,\i}$}};}
\graph { (p0) ->["\escale{$k_{18i+5}$}"] (p1) ->["\escale{$X_b$}"] (p2) ->["\escale{$p_{i}$}"] (p3) ->["\escale{$k_{18i+14}$}"] (p4);};
\end{scope}
\begin{scope}[yshift = -3cm]
\node at (-1,0) {c)};
\foreach \i in {0,...,4} {\coordinate (\i) at (\i*1.9cm,0);}
\foreach \i in {0} {\fill[black!15, rounded corners] (\i) +(-0.5,-0.3) rectangle +(0.5,0.3);}
\fill[black!15, rounded corners] (2) +(-0.5,-0.3) rectangle +(2.4,0.4);
\foreach \i in {0,...,4} {\node (p\i) at (\i) {\nscale{$t_{i,2,\i}$}};}
\graph { (p0) ->["\escale{$k_{18i+8}$}"] (p1) ->["\escale{$\tilde{X}_b$}"] (p2) ->["\escale{$p_{i}$}"] (p3) ->["\escale{$k_{18i+17}$}"] (p4);};
\end{scope}
\end{tikzpicture}
\caption{
The translators (a) $T_{i,0}$, (b) $T_{i,1}$, (c) $T_{i,2}$ with region $R^{T_i}_b$, one of three possibilities in case of exiting key copies.
Here just $X_b$ has a positive signature while $X_a, X_c$ obey.
}
\label{fig:Linear3-foldESSP_T}
\end{figure}
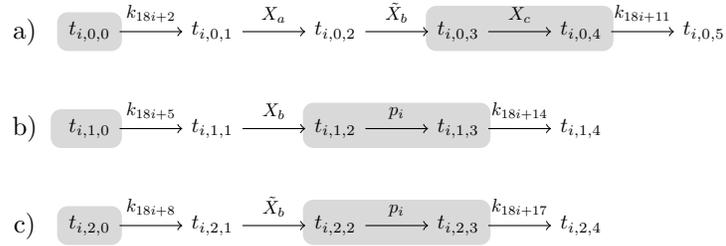

The following lemma establishes the foundation for the correctness of our reduction:
\begin{lemma}
\label{lemma:KeyRegionTruthAssignement}
In the union $U^\varphi$, the key event $k$ is inhibitable at the key state $m_6$ if and only if $\varphi$ has a one-in-three model.
\end{lemma}
\begin{proof}
\noindent\emph{If:}
If $M$ is a one-in-three model of $\varphi$, we firstly define a region $R^{T} = \bigcup_{i = 0}^{m-1} R^{T_i}_j$ of $T$ by adding $R^{T_i}_j$ for every $T_i$ where $X_j = M \cap K_i$ is the one variable $X_a, X_b$ or $X_c$ of $K_i$ covered by $M$.
By definition, this lets all contained key copies exit.
By Lemma~\ref{lemma:KeyRegionB} the region $R^B$ of $B$ inhibits $k$ at $m_6$ and lets all key copies exit, too.
As key copies are the only events shared by $B$ and $T$, the two regions are compatible and $R^B \cup R^T$ is a region of $U^\varphi$ inhibiting $k$ at $m_6$.

\noindent\emph{Only if:}
If $R$ is a region of $U^\varphi$ that inhibits $k$ at $m_6$ then Lemma~\ref{lemma:KeyRegionB} states that, without loss of generality, $R$ contains $R^B$ as subregion for $B$.
This implies that all key copies exit.
By Lemma~\ref{lemma:KeyRegionT}, {\em every} $i \in \{0, \dots, m-1\}$ {\em exactly} defines one variable event $X \in K_i = \{X_a,X_b,X_c\}$ that has $sig_R(X) = 1$, while $sig_R(Y) = 0$ for the other two $Y \in K_i \setminus X$.
This yields a one-in-three model $M = \{X \mid X \in V(\varphi), sig_R(X) = 1\}$.
\end{proof}
Having constituted the connection between the key region and the original satisfiability problem, it remains to show for all other combinations of event $e$ and state $s$ that, independent of the existence of a one-in-three-model, $e$ can be inhibited at $s$.
That this is possible, is stated in the following Lemma:
\begin{lemma}
\label{lem:Secondary_Proofs_for_The_Hardness_of_Linear_3Feasibility_and Linear_3ESSP}
For all events $e \in E(U^\varphi)$ and all states $s \in S(U^\varphi)$ that fulfill $(e,s) \not= (k,m_6)$ and $\neg(s\edge{e})$ there is a region inhibiting $e$ at $s$.
\end{lemma}
The proof of Lemma \ref{lem:Secondary_Proofs_for_The_Hardness_of_Linear_3Feasibility_and Linear_3ESSP} is pretty long and technical.
Therefore it has been split into multiple lemmas and moved to the auxiliary proofs in Section~\ref{sec:Secondary_Proofs}.
Since, furthermore, the polynomial running time of our reduction is obvious, we have established the NP-completeness of deciding the ESSP for $3$-fold unions of linear TSs.

It remains to show that the NP-completeness of deciding the ESSP reaches the joining $A(U^\varphi)$.
The principal element for this is the following Lemma~\ref {lem:StrengtheningLinESSPimpliesLinSSP} showing that, in case of linear TS, the ESSP implies the SSP.
More precisely, if there is a subset $\mathcal{R}\subset \mathcal{R}(A)$ of regions that suffices for the ESSP of $A$ then $\mathcal{R}$ also contains enough regions to separate all states of $A$.
For the following formalism, we call a set $\mathcal{R} \subset \mathcal{R}(A)$ a witness for the ESSP of TS $A$, if for all states $s$ and all events $e$ of $A$ that fail $s\edge{e}$ there is a region in $\mathcal{R}$ inhibiting $e$ at $s$.
Analogously, $\mathcal{R} \subset \mathcal{R}(A)$ a witness for the SSP of $A$, if for all states $s,s'$ there is a region in $\mathcal{R}$ separating $s$ and $s'$.
\begin{lemma}
\label{lem:StrengtheningLinESSPimpliesLinSSP}
If $A$ is a linear TS and $\mathcal{R} \subset \mathcal{R}(A)$ is a witness for the ESSP of $A$ then $\mathcal{R}$ is a witness for the SSP of $A$.
\end{lemma}
Before we prove the lemma, we include some additional definitions and facts from \cite{BBD2015} that are only relevant for this section:
Let $A=(S,E,\delta,s_0)$ be a TS.
First of all, the language of $A$ is defined as the set  
\[
L(A)=\{e_0 \dots e_m \in E^* \mid s_0 \edge{e_1}\dots \edge{e_m} s_m \text{ is a (maybe not simple) directed path of } A\}.
\]
Notice, that the language is trivially finite if $A$ is linear.
If $\mathcal{R} \subseteq \mathcal{R}(A)$ is a subset of $A$'s regions then the $\mathcal{R}$-restricted synthesized net $SN_{\mathcal{R}}(A)=(P,T,F,M_0)$ of $A$ is defined as follows:
The set $P$ of places equals $\mathcal{R}$ and the set $T$ of transitions is just $E$.
The set of flow arcs is defined by $F = \{(R,e) \mid R \in \mathcal{R}, e \in exit_R\} \cup \{(e,R) \mid R \in \mathcal{R}, e \in enter_R\}$.
A marking $M \subseteq P$ is a subset of places and $M_0$ is the initial marking given by $M_0 = \{R \in \mathcal{R} \mid R(s_0)=1\}$.
A marking $M$ reaches marking $M'$ by transition $e \in T$, denoted by $M[e>M'$, if and only if the input $\{R \mid (R,e) \in F\}$ of $e$ is exactly $M\setminus M'$ and the output $\{R \mid (e,R) \in F\}$ of $e$ is exactly $M'\setminus M$.

The reachability graph of $SN_{\mathcal{R}}(A)$ is the TS $RG(SN_{\mathcal{R}}(A))=(S',E,\delta', s'_0)$ with states $S'=\{M \subseteq P \mid \exists e_1, \dots,  e_m \in T: M_0[e_1> \dots [e_m> M\}$, the set of all reachable markings of $SN_{\mathcal{R}}(A)$, including the initial state $s'_0$ which equals $M_0$.
For all $M \in S'$ and all $e \in E$, the transition function is defined as $\delta'(M,e) = M'$ if and only if $M[e>M'$.
The states of TS $A$, the reachable markings of a $\mathcal{R}$-restricted synthesized net $SN_{\mathcal{R}}(A)$  and the states of the reachability graph $RG(SN_{\mathcal{R}}(A))$ are related as follows:
For every $s \in S$ we have $R_s=\{R \in \mathcal{R} \mid R(s)=1\}$, the set of all regions in $\mathcal{R}$ containing $s$, is a reachable marking of $SN_{\mathcal{R}}(A)$. 
Moreover, for each reachable marking $M$ of $SN_{\mathcal{R}}(A)$ there is a state $s\in S$ such that $M=R_s$. 
As markings become states in $RG(SN_{\mathcal{R}}(A))$, we get that $S' =\{R_s \mid s\in S\}$.

A set $\mathcal{R} \subset \mathcal{R}(A)$ is a witness for the ESSP of $A$ if and only if the function $\psi: S \longrightarrow S'$ with $\psi(s)=R_{s}$ for all $s \in S$ is a surjective morphism between $A$ and $RG(SN_{\mathcal{R}}(A))$.
That means, for every arc $s\edge{e} s'$ in $A$ there is an arc $\psi(s)\edge{e} \psi(s')$ in $RG(SN_{\mathcal{R}}(A))$.
This is equivalent to $L(A)=L(RG(SN_{\mathcal{R}}(A)))$~\cite{BBD2015}.
\begin{proof}[Proof of Lemma \ref{lem:StrengtheningLinESSPimpliesLinSSP}]
Let $A=s_0\edge{e_1}\dots\edge{e_n}s_n$ and $\mathcal{R} \subset \mathcal{R}(A)$ be witness for the ESSP of $A$.
Then we have a surjective morphism $\psi: S \longrightarrow \{R_s\vert s\in S\}$ implying $L(A)=L(RG(SN_{\mathcal{R}}(A)))$.
Assume we have states $s_i,s_j\in S$ that are not separable by any region of $\mathcal{R}$.
Without loss of generality, let $i < j$.
Then, for every region $R \in \mathcal{R}$ we get $R(s_i)=R(s_j)$.
From $s_0\edge{e_1}\dots \edge{e_i} s_i \edge{e_{i+1}} \dots \edge{e_j} s_j$ in $A$ we can deduce that $\psi(s_0)\edge{e_1} \dots \edge{e_{i}}\psi(s_i)\edge{e_{i+1}}\dots\edge{e_j}\psi(s_j)$ is in $RG(SN_{\mathcal{R}}(A))$, where $\psi(s_i)=R_{s_i}=R_{s_j}=\psi(s_j)$. 
Because of $s_i\not=s_j$ we have that $\vert e_{i+1}\dots e_j\vert \geq 1$.
In particular, we get $\{e_1 \dots e_{i} (e_{i+1} \dots e_{j})^m \mid m\in \mathbb{N}\}$ as a subset of $L(RG(SN_{\mathcal{R}}(A)))$. 
This contradicts $L(A)=L(RG(SN_{\mathcal{R}}(A)))$ because $L(A)$ is finite.
\end{proof}

Assume that $U^\varphi$ has the ESSP.
To complete the proof of this section's theorem, we can now use Lemma \ref{lem:StrengtheningLinESSPimpliesLinSSP} to deduce the SSP of $U^\varphi$.
We define for all TSs $A = (S, E, \delta, s_0)$ in $U^\varphi$ the region set $R_A = \{R \cap S \mid R \in \mathcal{R}(U^\varphi)\}$.
By definition, $R_A$ is a witness for the ESSP of $A$.
According to Lemma \ref{lem:StrengtheningLinESSPimpliesLinSSP}, $R_A$ is a witness for the SSP of $A$, too.
Consequently, for any states $s,s' \in S(U^\varphi)$ coming from the same TS $A$, there is a region $R \in \mathcal{R}(U^\varphi)$ such that $R \cap S \in R_A$ is a region of $A$ that separates $s, s'$. 
Hence, $U^\varphi$ has the SSP.

By Lemma~\ref{lem:union_validity}, we obtain that $U^\varphi$ is feasible, respectively has the ESSP, if and only if $A(U^\varphi)$ is feasible, respectively has the ESSP.
This proves Theorem~\ref{thm:The_Hardness_of_Linear_3Feasibility_and Linear_3ESSP}.


\section{The Hardness of Linear 3-SSP}
\label{sec:The_Hardness_of_Linear_3-SSP}

If SSP was of a lesser complexity than ESSP, it could serve as a fast fail preprocess for feasibility.
That means, if a linear TS fails the efficient SSP test in question one could save the costly ESSP test.
However, this possibility is ruled out by a polynomial time reduction of linear $3$-ESSP to linear $3$-SSP provided in this section:
\begin{theorem}\label{thm:The_Hardness_of_SSP}
To decide the SSP is NP-complete on $g$-grade $k$-fold transition systems for all $g \geq 1$ and all $k \geq 3$.
\end{theorem}
Given a linear $3$-fold TS $A=(S, E, \delta, s_0)$, our reduction creates a separate union $U_s^e$ of linear $3$-fold TSs for every pair of event $e \in E$ and state $s \in S$ that do not fulfill $s \edge{e}$.
By construction, $U_s^e$ has two key states that are separable if and only if $e$ can be inhibited at $s$ in $A$.
Then, as the unions get mutually disjoint state and event sets by design, they can be smoothly merged into an aggregate union $U^A = U(U_s^e \mid e \in E, s \in S, \neg (s \edge{e}))$.
If $U^A$ has the SSP then the key states of $U_s^e$ are separable for all relevant $(e,s)$, which implies the ESSP for $A$.
Reversely, if $A$ has the ESSP, we use Lemma \ref{lem:StrengtheningLinESSPimpliesLinSSP} to get the SSP for $U^A$.
Finally, we join the TSs of $U^A$ simply by concatenating the individual linear components, which results in a linear $3$-fold TS $A'$.
Applying Lemma \ref{lem:union_validity}, we get that $A'$ has the SSP if and only if $U^A$ has the SSP if and only if $A$ has the ESSP.

Having an outline of the primal reduction approach, the remainder of this section focuses on the introduction of $U_s^e$ for event $e \in E$ and state $s \in S$ failing $s \edge{e}$.
Basically, $U_s^e$ installs a TS $M$ with two key states $m_0, m_1$ and a TS $C$ representing a copy of $A$ such that, effectively, $m_0$ and $m_1$ can be separated only by a key region $R$ that has $e$ as leaving event and excludes $s$.
Hence, the separability of $m_0$ and $m_1$ implies that $e$ is inhibitable at $s$ and vice versa.
One difficulty with this idea is to get along with just three assignments of event $e$.
The union solves this by including additional TSs that exploit properties of a key region to copy the signature of $e$ to other events subsequently serving as replacements for $e$.
To become more specific, we go into the details of the applied gadget TSs:
\begin{description}
\item[Mapper]
$M$ is a linear TS on states $\{m_0, \dots, m_5\}$ where the separation of $m_0$ and $m_1$ by a suitable region $R^M$ makes sure that $e \in exit_{R^M}$.
Figure \ref{fig:LinearThreeSSPGadgets}(a) introduces $M$ together with $M$'s part of a key region, $R^M$.
The event $e$ is already used three times by $M$ to force the first two so-called vice events $v_0$ and $v_1$ into $enter_{R^M}$.
\item[Duplicators]
$D_j$ are linear TSs on states $\{d_{j,0}, \dots, d_{j,13}\}$ that, for a duplicator part $R^{D_j}$ of a key region, are supposed to synchronize the signature of $e$ with events $\{e_{2j}, e_{2j+1}\}$, so-called $e$-copies.
In this way we work around the limitation of using $e$ only three times.
There are five duplicators $D_0, \dots, D_4$ and each of them generates one $e$-copy $e_{2j+1}, j \in \{0, \dots, 4\}$ that has a free assignment to be used somewhere else.
As Figure \ref{fig:LinearThreeSSPGadgets}(b) demonstrates, every duplicator $D_j$ exploits the previous vice events $v_{2j}, v_{2j+1}$, which are appointed to $enter_{R^{D_j}}$, to force a negative signature onto the $e$-copies and to synchronize two further vice events $v_{2j+2}, v_{2j+3}$ to be used in the next duplicator.
\item[Provider]
$P$ is a linear TS on states $\{p_0, \dots, p_7\}$ which, for the provider part $R^P$ of a key region, applies the negative signature of the $e$-copies $e_7$ and $e_9$ and the positive signature of the last two vice events $v_{10}$ and $v_{11}$ to provide helper events $h_1$ and $h_2$ with $sig_{R^P}(h_1) = 1$ and $sig_{R^P}(h_2) = -1$.
The purpose of $h_1$ and $h_2$ is to enhance the following copy of $A$ in such a way that $s$ is guaranteed to be outside the key region.
See Figure \ref{fig:LinearThreeSSPGadgets}(c) for a definition of $P$ and $R^P$.
\item[Copy]
$C$ is a linear TS that basically copies $A$ into the union.
However, $C$ replaces every of the at most three occurrences of $e$ by a free $e$-copy $e_1, e_3, e_5$.
Also, $s$ is enhanced by two edges $s \edge{h_1} p \edge {h_2} s'$.
If there was an edge $s \edge{e'} q$ in $A$ we let $C$ continue with this edge after $s'$, that is, $s' \edge{e'} q$.
In $C$'s part $R^C$ of a key region, the $e$-copies inherit the negative signature of $e$.
Moreover, we get $sig_{R^C}(h_1) = 1$ and $sig_{R^C}(h_2) = -1$ which means that neither $s$ nor $s'$ are in $R^C$.
Combining these two facts, we get a slightly modified version of $R^C$ that can be used as a region of $A$ to inhibit $e$ at $s$.
\end{description}
Altogether, the construction of $U_s^e$ results in $U(M, D_0, \dots, D_4, P, C)$.
Notice that the same construction scheme generates unions $U_s^e$ for multiple instances $(e,s)$.
Although not explained in detail at this point, we later enhance the construction by a renaming mechanism that prevents possible state or event clash by enhancing the used state and event names with a unique identifier of the union $U_s^e$ they occur in.
\begin{figure}
\centering
\begin{tikzpicture}[scale = 0.87]
\begin{scope}[yshift = 2.5cm]
\node at (-1,0) {a)};
\foreach \i in {0,...,5} {\coordinate (\i) at (\i*2cm,0);}
\foreach \i in {0, 2, 4} {\fill[black!15, rounded corners] (\i) +(-0.5,-0.3) rectangle +(0.5,0.3);}
\foreach \i in {0,...,5} {\node (m\i) at (\i) {\nscale{$m_{\i}$}};}
\graph { (m0) ->["\escale{$e$}"] (m1) ->["\escale{$v_0$}"] (m2) ->["\escale{$e$}"] (m3) ->["\escale{$v_1$}"] (m4) ->["\escale{$e$}"] (m5);};
\end{scope}
\begin{scope}
\node at (-1,0) {b)};
\foreach \i in {0,...,6} {\coordinate (\i) at (\i*2cm,0.75);}
\foreach \i in {7,...,13} {\coordinate (\i) at (26cm-\i*2cm,-0.75);}
\foreach \i in {1, 3, 5, 10, 12} {\fill[black!15, rounded corners] (\i) +(-0.5,-0.3) rectangle +(0.5,0.3);}
\fill[black!15, rounded corners] (8) +(-0.5,-0.3) rectangle +(2.5,0.3);
\foreach \i in {0,...,13} {\node (d\i) at (\i) {\nscale{$d_{j,\i}$}};}
\graph { (d0) ->["\escale{$v_{2j}$}"] (d1) ->["\escale{$e_{2j}$}"] (d2) ->["\escale{$v_{2j+1}$}"] (d3) ->["\escale{$b_{2j}$}"] (d4) ->["\escale{$v_{2j}$}"] (d5) ->["\escale{$e_{2j+1}$}"] (d6) ->["\escale{$v_{2j+1}$}"] (d7) ->["\escale{$b_{2j+1}$}"] (d8) ->["\escale{$e_{2j}$}"] (d9) ->["\escale{$v_{2j+2}$}"] (d10) ->["\escale{$e_{2j+1}$}"] (d11) ->["\escale{$v_{2j+3}$}"] (d12) ->["\escale{$e_{2j}$}"] (d13);};
\end{scope}
\begin{scope}[yshift = -2.5cm]
\node at (-1,0) {c)};
\foreach \i in {0,...,7} {\coordinate (\i) at (\i*1.7143cm,0);}
\foreach \i in {0, 2, 5, 7} {\fill[black!15, rounded corners] (\i) +(-0.5,-0.3) rectangle +(0.5,0.3);}
\foreach \i in {0,...,7} {\node (p\i) at (\i) {\nscale{$p_{\i}$}};}
\graph { (p0) ->["\escale{$e_7$}"] (p1) ->["\escale{$h_1$}"] (p2) ->["\escale{$e_9$}"] (p3) ->["\escale{$b$}"] (p4) ->["\escale{$v_{10}$}"] (p5) ->["\escale{$h_2$}"] (p6) ->["\escale{$v_{11}$}"] (p7);};
\end{scope}
\end{tikzpicture}
\caption{
Gadgets of union $U_s^e$ together with their parts of a key region.
(a) $M$ with $R^{M}$, (b) $D_j$ with $R^{D_j}$, (c) $P$ with $R^{P}$.
}
\label{fig:LinearThreeSSPGadgets}
\end{figure}
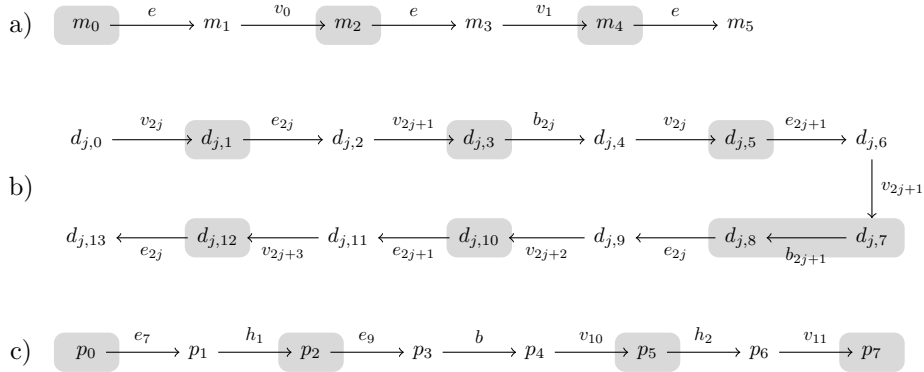

The correctness of the given reduction is based on the following argumentation.
Firstly, the following lemma formalizes that separation of key states implies a unique key region that inhibits $e$ at $s$:
\begin{lemma}
\label{lem:union_separability}
If $R$ is a region of $U^e_s$ with $m_0 \in R$ and $m_1 \not\in R$ then
\begin{enumerate}
\item $R \cap \{m_0, \dots, m_5\} = R^M = \{m_0,m_2,m_4\}$,
\item $R \cap \{d_{j,0}, \dots, d_{j,13}\} = R^{D_j} = \{d_{j,1},d_{j,3},d_{j,5},d_{j,7},d_{j,8},d_{j,10},d_{j,12}\}$ for all $j \in \{0, \dots, 4\}$,
\item $R \cap \{p_0, \dots, p_5\} = R^P = \{p_0,p_2,p_5,p_7\}$,
\item $sig_R(e) = sig_R(e_0) = \dots = sig_R(e_9) = -1$,
\item $sig_R(v_0) = \dots = sig_R(v_{11}) = 1$,
\item $sig_R(h_1) = 1$, and $sig_R(h_2) = -1$, and
\item the set $R' = R \cap S$ is a region of $A$ inhibiting $e$ at $s$.
\end{enumerate}
\end{lemma}
\begin{proof}
By $m_0 \in R$ and $m_1 \not\in R$ it is easy to see that $sig_R(e) = -1$, which implies $m_2,m_4 \in R$ and $m_1, m_3, m_5 \not\in R$ and, thus, $sig_R(v_0) = sig_R(v_{1}) = 1$.
Iterating through the duplicators for $j \in \{0, \dots, 4\}$ , we get from $sig_R(v_{2j}) = sig_R(v_{2j+1}) = 1$ that $d_{j,0}, d_{j,2}, d_{j,4}, d_{j,6} \not\in R$ and $d_{j,1}, d_{j,3}, d_{j,5}, d_{j,7} \in R$ and $sig_R(e_{2j}) = sig_R(e_{2j+1}) = -1$, which in turn means that $d_{j,8}, d_{j,10}, d_{j,12} \in R$ and $d_{j,9}, d_{j,11}, d_{j,13} \not\in R$ and that $sig_R(v_{2j+2}) = sig_R(v_{2j+3}) = 1$.
Finally, using $sig_R(e_7) = sig_R(e_9) = -1$ and $sig_R(v_{10}) = sig_R(v_{11}) = 1$ we get $p_0,p_2,p_5,p_7 \in R$ and $p_1,p_3,p_4,p_6 \not\in R$ which means that $sig_R(h_1) = 1$ and $sig_R(h_2) = -1$.

To see the correctness of statement 7, recall that, by definition, $R$ provides a (sub) region $R^C$ for $C$.
If $C$ had not been modified with $e$-copies and $s \edge{h_1} p \edge {h_2} s'$ we basically could use $R^C$ as a region of $A$.
However, as $sig_R(e) = sig_R(e_1) = sig_R(e_3) = sig_R(e_5) = -1$, we get that the signature of $e$ is negative and appropriately translated into $R'$.
Moreover, by $sig_R(h_1) = 1$ and $sig_R(h_2) = -1$, we correctly have $s,s' \not\in R$ which means that $s \not\in R'$ and that, in case the edge $s \edge{e'} q$ exists, $sig_R(x) = sig_{R'}(x)$.
\end{proof}
Hence, as the part $R^C$ basically is a region of $A$ that inhibits $e$ at $s$, the SSP for $U^A$, which includes the separability of the key states in $U_s^e$ for all relevant $(e,s)$, implies the ESSP of $A$.

Next, the following lemma makes sure that every region $R$ of $A$ can be translated into a meaningful region $R'$ of $U^e_s$.
Meaningful is to say that $R'$ is a region that adopts the signature of $A$ for the according events of $C$.
We take care that the translation works in a way that forces as many events of $U^A$ as possible to be included in $obey_{R'}$.
This will make it easier to conclude the SSP of $U^A$ from the ESSP of $A$.
\begin{lemma}
\label{lem:region_translation}
If $R$ is region of $A$, then the set $R' = R^C \cup R^U$ with
\[R^U= \begin{cases}
\bigcup_{j = 0}^2 \{d_{j,6}, d_{j,7}, d_{j,11}, d_{j,12}, d_{j,13}\}, & \text{if } sig_R(e) = -1,\\
\emptyset,  & \text{if } sig_R(e) = 0,\\
\bigcup_{j = 0}^2 \{d_{j,0}, \dots,  d_{j,5}, d_{j,8}, d_{j,9}, d_{j,10}\}, & \text{if } sig_R(e) = 1, \text{ and}\\
\end{cases}\]
\[R^C= \begin{cases}
R, & \text{if } s \not\in R,\\
R \cup \{p,s'\}, & \text{otherwise,}\\
\end{cases}\]
is a region of $U^e_s$ where for all $x \in E \setminus \{e\}$ it is true $sig_R(x) = sig_{R'}(x)$ and $sig_R(e) = sig_{R'}(e_1) = sig_{R'}(e_3) = sig_{R'}(e_5)$ and $R(y) = R'(y)$ for all $y \in S$ and $R'(s) = R'(p) = R'(s')$.
\end{lemma}
\begin{proof}
We start by showing that $R^U = R' \cap S(U)$ is a region of the union $U = U(M, D_0, \dots, D_4, P)$ and that $R_C = R' \cap S^C$ is a region of the TS $C$ having state set $S^C$.
Subsequently, we argue that the signatures $sig_{R^U}$ and $sig_{R^C}$ merge into an aggregate signature $sig_{R'}$ such that $R'$ becomes a region of $U^e_s$.

For $R^U$ it is easy to check that it admits the signature
\[sig_{R^U}(x) = \begin{cases}
0, & \text{if } x \in \{e, e_0, e_2, e_4, e_6, \dots, e_9, h_1, h_2, v_0, \dots, v_{11}, b, b_0,\\  &\hspace{1cm} b_2, b_4, b_6, \dots, b_9\},\\
sig_R(e), & \text{if } x \in \{e_1, e_3, e_5\},\\
-sig_R(e), & \text{if } x \in \{b_1, b_3, b_5\}
\end{cases}\]
and therefore is a region of $U$.

Next, to show that $R^C$ is a region of $C$ is pretty straightforward as $C$ is an enhanced copy of $A$ and $R^C$ is a superset of $R$, a region of $A$.
It is sufficient to argue that, (1) with respect to events $E \setminus \{e\}$, the signature $sig_R$ can be kept for $R^C$, (2) the $e$-copies $e_1, e_3, e_5$ inherit their signature from $e$, and (3) the signature of $h_1$ and $h_2$ is simply zero.
Observe, $s, p$ and $s'$ are either all in $R^C$ or all not in $R ^C$ such that a possible edge $s' \edge{e'} q$ in $C$ behaves just like the original edge $s \edge{e'} q$ and therefore fulfills $sig_{R ^C}(e') = sig_R(e')$.
Hence, $R^C$ admits the signature
\[sig_{R ^C}(x) = \begin{cases}
sig_R(x), & \text{if } x \in E \setminus \{e\},\\
sig_R(e), & \text{if } x \in \{e_1, e_3, e_5\},\\
0, & \text{if } x \in \{h_1, h_2\}
\end{cases}\]
and, consequently, is a region of $C$.
Having $R^U$ as a region of $U$ and $R^C$ as a region of $C$, it remains to show that $R' = R^U \cup R^C$ is a region of $U^e_s = U(U, C)$.
This requires to show that the signatures $sig_{R^U}$ and $sig_{R^C}$ coincide on shared events and, thus, can be merged to a signature $sig_{R^e_s}$.
However, by construction, the only events shared by $U$ and $C$ are $e_1, e_3, e_5, h_1, h_2$ which, by definition, have equal signatures in both regions.
\end{proof}
In the following argumentation, Lemma \ref{lem:region_translation} is used to show that the ESSP of $A$, which by Lemma~\ref {lem:StrengtheningLinESSPimpliesLinSSP} also provides the SSP of $A$, implies the SSP of $U^A$, too.
But before we aim for $U^A$, we show the property for $U^e_s$:
\begin{lemma}
\label{lem:ESSP_to_SSP}
If $A$ has the ESSP, then for every event $e \in E$ and state $s \in S$ with $\neg (s \edge{e})$ the union $U^e_s$ has the SSP. 
\end{lemma}
\begin{proof}
In this proof we need the following mechanism which follows immediately from the definition of regions:
If $R$ is a region of union $U=U(A_1,\dots,A_m)$ and $A_{m+1},\dots,A_{n}$ are linear TSs such that, for all $i \in \{m+1,\dots, n\}$, there is at most one arc $s^i \edge{e^i} z^i$ in $A_i$ where the signature $sig_R(e^i)$ is defined and not zero, then $R$ can be enhanced to a region $R'$ of union $U(A_1,\dots,A_m,A_{m+1},\dots,A_{n})$ with $sig_{R'} = sig_{R}$.
In fact, we get $R'= R \cup R^{m+1} \cup \dots \cup R^{n}$ by defining for all $i \in \{m+1, \dots, n\}$ the set $R^i$ relative to $P_i$, the set of $s^i$ and all its predecessor states or, in case $s^i$ does not exist, $P_i = \emptyset$.
Then $R^i =S_i \setminus P_i$ if $sig_R(e^i) = 1$, where $S_i$ is the state set of $A_i$, and otherwise, $R^i = P_i$.
This mechanism is summarized in Lemma~\ref{lemma:easyLiftedRegion} in the auxiliary proof section.

Using region enhancement, we show that all pairs of states of $U^e_s$ are separable.
If the states originate from different TSs of $U^e_s$ then they are separable by definition.
Hence, we have to consider four cases, namely that both come from the same TS, either $M, D_j, P,$ or $C$.

We start with $M$ and in particular the key states $m_0$ and $ m_1$.
To separate them, we use $R^U = R^M \cup R^{D^0} \cup \dots \cup R^{D^4} \cup R^P$, which is a region of union $U = U(M, D_0, \dots, D_4, P)$ as easily verified by Figure \ref{fig:LinearThreeSSPGadgets}.
Notice that $sig_{R^U}(e_1) = sig_{R^U}(e_3) = sig_{R^U}(e_5) = sig_{R^U}(h_2) = -1$ and $sig_{R^U}(h_1) = 1$.
Moreover, as $A$ has the ESSP, we can find a region $R$ inhibiting $e$ at $s$ such that $sig_R(e)=-1$  and $R(s)=0$ and enhance it to $R^C = R \cup \{p\}$.
It is easy to see that $R^C$ is a region of $C$ with $sig_{R^C}(e_1) = sig_{R^C}(e_3) = sig_{R^C}(e_5) = sig_{R^C}(h_2) = -1$ and $sig_{R^C}(h_1) = 1$.
Clearly, the regions are compatible on the events $e_1, e_3, e_5, h_1, h_2$ shared by $U$ and $C$ and we can combine them to the key region $R_{key} = R^U \cup R^C$ of $U^e_s$ that separates the key states.

Moreover, in $M$ the region $R_{key}$ separates every state in $\{m_0,m_2,m_4\}$ from every state in $\{m_1,m_3,m_5\}$.
Here, it remains to show that $ m_i$ and $ m_j$ are separable if they originate from the same of the both sets, that is, if $i$ and $j$ are of the same parity.
Consider the region $R_0 = \{m_0,m_1,d_{0,0},d_{0,4}\}$ of $U$ having just obeying events except for $sig_{R_0}(v_0) = -1$ and $sig_{R_0}(b_0) = 1$.
As $v_0$ and $b_0$ do not occur in $C$, region $R_0$ can be enhanced to a region $R'_0$ of $U^e_s = U(U, C)$.
Obviously, $R'_0$ separates all states in $\{m_0, m_1\}$ from all states in $\{m_2, \dots, m_5\}$.
Consider the region $R_1 = \{m_0, \dots, m_3, d_{0,0}, d_{0,1}, d_{0,2}, d_{0,4}, d_{0,5}, d_{0,6}\}$ of $U$ where again all events are obeying except for $sig_{R_1}(v_0) = -1$ and $sig_{R_1}(v_1) = 1$.
The TS $C$ does not contain $v_0$ and $v_1$ and, consequently, region $R_1$ can be enhanced to a region $R'_1$ for $U^e_s$ that separates $m_2$ from $m_4$ and $m_3$ from $m_5$.

Next, we consider two states $s_1,s_2$ of $C$.
We firstly assume both states $s_1$ and $s_2$ to be in the state set $S$ of $A$.
Then, we use the fact that, by Lemma \ref{lem:StrengtheningLinESSPimpliesLinSSP}, $A$ inherits the SSP from the ESSP.
Let $R$ be a region of $A$ separating $s_1$ and $s_2$.
By Lemma \ref{lem:region_translation}, we can use the according region $R'$ that separates $s_1$ and $s_2$ in $C$.
If exactly one of the states, without loss of generality $s_1$, is in $S\setminus \{s\}$ and $s_2\in \{p,s'\}$ then let $R$ be a region of $A$ separating $s_1$ and $s$.
Again, we can enhance region $R$ for a region $R'$ that separates $s_1$ and $s$ in $C$.
As $s, p$ and $s'$ are not separated by this region, $R'$ separates $s_1$ and $s_2$, too.
The states $s$ and $p$ respectively $s'$ and $p$ are separated by the key region $R_{key}$ introduced above.
Finally, for $s$ and $s'$ we can generate a separating region of $U^e_s$ as follows:
We let $R^U = \{p_0, p_1\}$ and $R^C$ is the set of states from $S$ that are predecessors of or equal to $s$ in $A$.
It is easy to see that $sig_{R^U}(h_1) = sig_{R^C}(h_1) = -1$ and for all other events, both signatures are zero.
Hence, the set $R^U \cup R^C$ is a region of $U^e_s$ that contains $s$ but not $s'$.

In the following, we investigate $D_j$ for $j \in \{0, \dots, 4\}$ as well as $P$.
The key region $R_{key}$ already separates every state of $\{d_{j,0}, d_{j,2}, d_{j,4}, d_{j,6}, d_{j,9}, d_{j,11}, d_{j,13}\}$ from every state of $\{d_{j,1}, d_{j,3}, d_{j,5}, d_{j,7}, d_{j,8}, d_{j,10}, d_{j,12}\}$ in $D_j$ and in $P$ every state of $\{p_0, p_2, p_5, p_7\}$ from every state of $\{p_1, p_3, p_4, p_6\}$.
The argumentation for the remaining state pairs always works like already seen for $M$:
We define a region $R$ of $U$ by specifying the signature of all non-obeying events, we observe that $C$ has at most one edge labelled with such an event, we enhance region $R$ to a region of $U^e_s = U(U, C)$, and we list two sets of states $X$ and $Y$ such that every state in $X$ is separated from every state of $Y$ by $R$:

\vspace*{0.5cm}
\noindent\begin{tabular}{l | l | l | p{2.5cm}}
\small
 & Signature & $X$ & $Y$\\ \hline
$R_2$ & $sig_{R_2}(b_{2j}) = -1$ & $\{d_{j,0},\dots,d_{j,3}\}$ & $\{d_{j,4},\dots,d_{j,13}\}$\\
$R_3$ & $sig_{R_3}(b_{2j+1}) = -1$ & $\{d_{j,4},\dots,d_{j,7}\}$ & $\{d_{j,8},\dots,d_{j,13}\}$\\
$R_4$ & $sig_{R_5}(e_{2j+1}) = -sig_{R_5}(e_{2j}) = 1$ &$\{d_{j,0},d_{j,1},d_{j,8}, d_{j,11}\}$ & $\{d_{j,3},\dots,d_{j,5},d_{j,9},d_{j,10}\}$\\
$R_5$ & $sig_{R_6}(e_{2j+1}) = -sig_{R_6}(b_{2j}) = 1$ & $\{d_{j,8},d_{j,9}\}$ & $\{d_{j,12},d_{j,13}\}$\\
$R_6$ & $sig_{R_7}(b) = -1$ & $\{p_0, \dots, p_3\}$ & $\{p_4,\dots,p_7\}$\\
$R_7$ & $sig_{R_8}(h_1) = -1$ & $\{p_0, p_1\}$ & $\{p_2,p_3\}$\\
$R_8$ & $sig_{R_9}(b) = -sig_{R_9}(h_2) = 1$ & $\{p_4,p_5\}$ & $\{p_6,p_7\}$
\end{tabular}

\vspace*{0.5cm}
This completes the proof.
\end{proof}
Before we can put together the proof of Theorem \ref{thm:The_Hardness_of_SSP} there is one last thing that we have to consider: 
Our idea is to create a union $U^A = (U^e_s \mid e \in E, s\in S, \neg s \edge{e}\})$ to finish the reduction for input $A$.
However, in the given form, the unions $U^e_s$ would not have mutually disjoint state and event sets.
To resolve the name clash, we do the following renaming:
\begin{lemma}
\label{lem:rectified_union}
Let $U^e_s=U(M, D_0, \dots, D_4, P, C)$ be a union of TSs as defined above.
Then, we let $\tilde{U}^e_s=U(\tilde{M}, \tilde{D}_0, \dots, \tilde{D}_4, \tilde{P}, \tilde{C})$ be the rectified union where, for all TSs $T = (S, E, s_0, \delta)$ in $\{M, D_0, \dots, D_4, P, C\}$, we define $\tilde{T} = (\tilde{S}, \tilde{E}, \tilde{s}_0, \tilde{\delta})$ by
$\tilde{S} = \{(e,s,x) \mid x \in S\}$ and $\tilde{E} = \{(e,s,x) \mid x \in E\}$ and $\tilde{s}_0 = (e,s,s_0)$ and
\[\tilde{\delta}((e,s,x),(e,s,y)) = \begin{cases}
(e,s,\delta(x,y)), & \text{if $\delta(x,y)$ is defined},\\
\text{undefined}, & \text{otherwise}
\end{cases}\]
for all $(e,s,x) \in \tilde{S}$ and all $(e,s,y) \in \tilde{E}$.
Then $R$ is a region of $U^e_s$ if and only if $\tilde{R}=\{(e,s,x) \mid x \in R\}$ is a region of the rectified union $\tilde{U}^e_s$.
\end{lemma}
The proof of this lemma is straightforward and therefore omitted.
However, as the rectified unions have mutually disjoint state and event sets, the union $U^A$ can now be defined as
\[
U^A = (\tilde{U}^e_s \mid e \in E, s\in S, \neg s \edge{e}).
\]
\begin{proof}[Proof of Theorem \ref{thm:The_Hardness_of_SSP}]
We start by showing that $A$ has the ESSP if and only if $U^A$ has the SSP.
If $U^A$ has the SSP then, for every $e \in E$ and every $s \in S$ with $\neg (s\edge{e})$, there is a region $\tilde{R}$ of $U^A$ separating the key states $(e,s,m_0), (e,s,m_1)$ in $\tilde{U}^e_s$.
This means, by definition and by Lemma \ref{lem:rectified_union}, that
\[R^e_s = \{x \mid (e,s,x) \in \tilde{R} \cap S(\tilde{U}^e_s)\}\]
is a region of $U^e_s$ separating $m_0$ and $m_1$.
Lemma \ref{lem:union_separability} provides that $e$ is inhibtable at $s$.
Hence, $A$ has the ESSP.

Reversely, if $A$ has the ESSP then, by Lemma \ref{lem:ESSP_to_SSP}, $U^e_s$ has the SSP for all $e \in E$ and all $s \in S$ where $\neg s\edge{e}$.
Using Lemma \ref{lem:rectified_union}, we get for all $e \in E$ and all $s \in S$ where $\neg (s\edge{e})$ that $\tilde{U}^e_s$ has the SSP.
As these union are pairwise disjoint with respect to states and events, this implies, by definition, that $U^A$ has the SSP, too.

Given $A$, we can now create a linear $3$-fold TS $A' = A(U^A)$ by joining the individual TSs of $U^A$.
To make $A'$ linear, we basically concatenate the linear components of $U^A$ by selecting the terminal states as discretionary joining points.
With the help of Lemma \ref{lem:union_validity}, we subsequently get that $A'$ has the SSP if and only if $U^A$ has the SSP if and only if $A$ has the ESSP.
As creating $A'$ out of $A$ is obviously doable in polynomial time, we get that  linear $3$-ESSP $\leq_p$ linear $3$-SSP.
In Section~\ref{sec:The_Hardness_of_Linear_3Feasibility_and Linear_3ESSP} we show that ESSP of linear $3$-fold TSs is NP-complete and \cite{BBD1997} states that SSP is in NP, which, together with our argumentation, imply that linear $3$-SSP is NP-complete.
\end{proof}


\section{The Hardness of 2-Grade 2-Feasibility and 2-Grade 2-ESSP}
\label{sec:The_Hardness_of_2-grade_2-ESSP_2-Feasibility}

In this section we present our result, that feasibility and the ESSP remain NP-complete for $2$-grade $2$-fold TSs.
\begin{theorem}
\label{thm:twograde_twofold_NPcompleteness}
The feasibility problem and the ESSP problem are NP-complete on $g$-grade $k$-fold transition systems for all $g \geq 2$ and all $k \geq 2$.
\end{theorem}
Here, we again use the cubic monotone one-in-three $3$-SAT problem as a reduction source to argue that $2$-grade $2$-ESSP is NP-complete.
The definitions of a cubic monotone set $\varphi$ of boolean $3$-clauses as well as according one-in-three models is given in Section~\ref{sec:The_Hardness_of_Linear_3Feasibility_and Linear_3ESSP}.

Given $\varphi$, we again construct a union $U^{\varphi}$ in polynomial time.
However, this time $U^{\varphi}$ is a $2$-fold union of $2$-grade TSs.
Hence, every event is used at most twice and, even after joining the components in $A^{\varphi} = A(U^\varphi)$, every state has at most two predecessors and two successors.
We will show that if $U^{\varphi}$ has the ESSP then it has the SSP, too.
This makes ESSP and feasibility the same problem, even for $A^{\varphi}$ which is affirmed by Lemma~\ref{lem:union_validity}.
By that, we simultaneously get the NP-hardness of both problems for the input class of $2$-grade $2$-fold TSs.

With the rough proof outline in mind, we can take a closer look on the construction of $U^{\varphi}$:
First of all, the union $U^{\varphi}$ installs a TS $H$, called the \emph{headmaster}.
The headmaster initializes the connection between the satisfiability problem and the ESSP by introducing a key event $k$ that is supposed to be inhibitable at a certain key state if and only if $\varphi$ has a one-in-three model.
In order to achieve this behavior, $U^{\varphi}$ basically adds, for every clause $K_i$, a similiar translator $T_i$ as in Section~\ref{sec:The_Hardness_of_Linear_3Feasibility_and Linear_3ESSP}.
Again, for a key region, one that inhibits $k$ at the key state, $T_i$ translates the one-in-three condition of $K_i$ into the ESSP vocabulary.
This means that $T_i$ applies events that represent the three variables of $K_i$ and exactly one of them has a positive signature while the other two obey.
Consequently, there is a key region where every gadget $T_0, \dots, T_{m-1}$ has exactly one entering variable event if and only if there is a one-in-three model of $\varphi$.

Like before, the main problem is to get along with, in this case, just two occurrences of every event.
The problem is again solved by adding TSs that, for a key region, generate helper and replacement events with predefined signatures, leaving, entering, or obeying.
See Figure \ref{fig:gadgets} to visualize the technical elaborations of the following detailed description of our gadget TSs. 
\begin{description}
\item[Headmaster]
$H$ is a TS that introduces the key event $k$ and the key state $h_{0,8}$.
Moreover, $H$ cooperates with the subsequent duplicator gadgets to prepare the key copies, the event series with negative signature needed by the translators.
To this end, $H$ works with a production line of $14m$ submodules $H_j$, each cooperating with a duplicator to initialize one key copy.
More precisely, for a key region, $H_j$ prepares three events for the duplicator $D_j$.
Two so-called \emph{vice} events $v_{2j}, v_{2j+1}$ have a positive signature, that is, vice with respect to the signature of $k$.
Then there is one obeying \emph{wire} event $w_{2j}$ used to communicate the condition of one duplicator state to another.
In return, the duplicator provides two key copies $k_{3j}, k_{3j+1}$ and one obeying \emph{accordance} event $a_j$.
In $H_{j+1}$ the three resulting events are used for the synchronization of the next vice and wire events.
The main result of $D_j$, however, is $k_{3j+2}$, one of the $14m$ key copies that are free to be in translators and other gadgets.

See Figure~\ref{fig:gadgets}~a) for a definition of $H$ together with an illustration of $H$'s part of a key region, $R^H$.
Observe that there are \emph{reachability} events $r_0, \dots, r_{14m-2}$ which make every state of $H$ reachable from the initial state $h_{0,0}$.
Moreover, for $R^H$, every module $H_j$ receives the key copies $k_{2j-2}, k_{2j-1}$ and the obeying accordance event $a_{j-1}$ from duplicator $D_{j-1}$.
Thanks to $a_{j-1}$, the state $h_{j,8}$ behaves according to the key event $h_{0,8}$ and is excluded from $R^H$.
Because of the exiting key copy $k_{2j-2}$ the state $h_{j,1}$ is not in $R^H$.
The last two facts together imprint a zero signature on the \emph{zero} events $z_{2j}, z_{2j+1}$.
The exiting key copy $k_{2j-1}$ puts $h_{j,4}$ into $R^H$ and excludes $h_{j,5}$ which, together with the obedience of $z_{2j}, z_{2j+1}$, takes care that $v_{2j},v_{2j+1}$ are entering and that $w_{2j}, w_{2j+1}$ are obeying.
\item[Duplicators]
$D_j$ are TSs that, for a key region, generate three key copies $k_{3j}, k_{3j+1}, k_{3j+2}$ and one obeying accordance event $a_j$ using the vice events $v_{2j}, v_{2j+1}$ with positive signature and the obeying wire event $w_{2j}$.

Figure \ref{fig:gadgets}~b) defines $D_j$ and demonstrates the duplicator fraction $R^{D_j}$ of a key region.
It is obvious, that the predefined positive signature of the two vice events puts $d_{j,2}, d_{j,4}$ into $R^{D_j}$ and exclude $d_{j,1}, d_{j,3}$.
The obedience of the wire event is used to signal the condition of $d_{j,4}$ to $d_{j,0}$ forcing it into $R^{D_j}$, too.
By design, $k_{3j}, k_{3j+1}, k_{3j+2}$ turn into entering events and $a_j$ becomes obeying.

As $H$ consumes only $k_{3j}, k_{3j+1}$ and $a_j$, we keep the remaining duplicate $k_{3j+2}$ of $k$.
As we create $14m$ duplicators, we get $14m$ free key copies.
\item[Barters]
$B_q$ are TSs that, for a key region, barter two key copies $k_{q_1}, k_{q_2}$ for one obeying so-called \emph{consistency} event $c_q$.
Here, the indexes $q_1 = 6q + 18m + 1$ and $q_2 = q_1 + 3$ are chosen to select from the last $4 \cdot 2m$ free items in the list of generated key copies.

The use of consistency events is to make the signature of three events consistent for every variable $X$ of $\varphi$.
The reason is that we cannot represent three occurrences of $X$ by an event that can only be used twice.
This is a crucial difference with respect to Section \ref{sec:The_Hardness_of_Linear_3Feasibility_and Linear_3ESSP}, where every variable can be represented by the three occurrences of one event.
Here, we require multiple synchronized events to represent $X$, which is enabled by the consistency events.

Figure \ref{fig:gadgets}~c) introduces $B_q$ and shows a respective key region part $R^{B_q}$.
Clearly, as both key copies are leaving, we get $b_{q,0}, b_{q,2} \in R^{B_i}$ which makes $c_q$ obeying.
Altogether, we add $4m$ barters that consume $8m$ key copies to generate $4m$ consistency events. 
\item[Variable manifolders]
$X_i$ are TSs that essentially represent one variable $X_i \in V(\varphi)$.
They are used to synchronize three variable events for $X_i$ to be used in the translators.
In particular, if $\alpha, \beta, \gamma \in \{0, \dots, m-1\}$ specify the clauses $K_\alpha, K_\beta, K_\gamma$ that contain variable $X_i$ then the TS $X_i$ provides three events $X_i^\alpha, X_i^\beta, X_i^\gamma$.
The construction of $U^\varphi$ assures that, for a key region, these events have equal signatures such that they can be treated like manifestations of the same event representing $X_i$.

The definition of TS $X_i$ as well as an illustration of a possible key region fragment $R^{X_i}$ are given in Figure \ref{fig:gadgets}~d).
To create the event equivalence, $X_i$ applies four consistency events, $c_{4i}$, $c_{4i+1}$, $c_{4i+2}$, $c_{4i+3}$.
The obedience of $c_{4i}$ and $c_{4i+1}$ condemns the states $x_{i,0}, x_{i,1}, x_{i,2}$ to a consistent behavior with respect to $R^{X_i}$, that is, either all or none of them are part of the region.
The same follows for $x_{i,3}, x_{i,4}, x_{i,5}$ from the obedience of $c_{4i+2}$ and $c_{4i+3}$.
It is easy to see that this makes the signature of all three events $X_i^\alpha, X_i^\beta, X_i^\gamma$ equal.
\item[Translators]
$T_i$ are unions $T_i = U(T_{i,0}, T_{i,1}, T_{i,2})$ of three TSs, that, in accordance to the translators of Section \ref{sec:The_Hardness_of_Linear_3Feasibility_and Linear_3ESSP}, implement the clauses $K_i = \{X_a, X_b, X_c\}$.
This means, $T_i$ represents the three variables by the representation events $X^i_a, X^i_b, X^i_c$ generated for clause $K_i$ and makes sure that exactly one of these events has a positive signature.
Figure \ref{fig:gadgets}~e)-g) defines the three TSs and introduces a possible key region fragment $R^{T_i}_b$ that assigns a positive signature to event $X^i_b$ representing $X_b$.
The methods of $T_i$ are the same as in Section~\ref{sec:The_Hardness_of_Linear_3Feasibility_and Linear_3ESSP} except that every variable event $X_j$, beeing unique in $T_i$, is replaced by the representer event $X^i_j$ of clause $K_i$.
\end{description}
Altogether, the construction results in the union
\[
U^\varphi = U(H, D_0, \dots, D_{14m-1}, B_0, \dots, B_{4m-1}, X_0, \dots, X_{m-1}, T_0, \dots, T_{m-1}).
\]
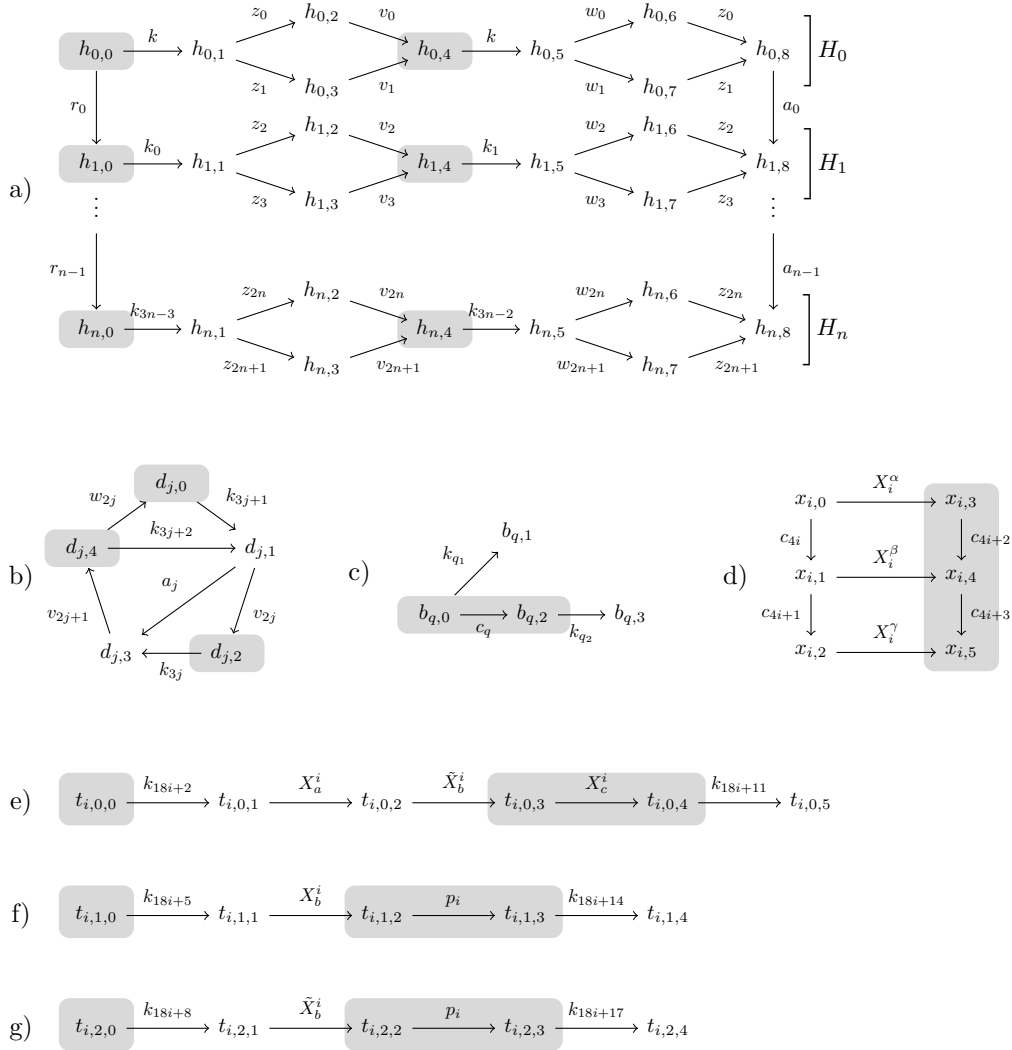
\begin{figure}[h!]
\begin{tikzpicture}
	\begin{scope}
 	
		\node (H) at (-1,-1.85) {a)};
		
		\filldraw[fill=black!15,rounded corners, dotted, line width=0.001] (0,0) ++(-0.5,-0.25) rectangle ++(1,0.5);
		\filldraw[fill=black!15,rounded corners, dotted, line width=0.001] (4.5,0) ++(-0.5,-0.25) rectangle ++(1,0.5);
		\node (0) at (0,0) {\nscale{$h_{0,0}$}};
		\node (1) at (1.5,0) {\nscale{$h_{0,1}$}};
		
		\node (2) at (3,0.5) {\nscale{$h_{0,2}$}};
		\node (3) at (3,-0.5) {\nscale{$h_{0,3}$}};

		\node (4) at (4.5,0) {\nscale{$h_{0,4}$}};
		\node (5) at (6,0) {\nscale{$h_{0,5}$}};
		
		\node (6) at (7.5,0.5) {\nscale{$h_{0,6}$}};
		\node (7) at (7.5,-0.5) {\nscale{$h_{0,7}$}};
		
		\node (8) at (9,0) {\nscale{$h_{0,8}$}};
		\node () at (9.7,0) {\scalebox{0.9}{$\Biggl\rbrack H_0$}};
		
		\filldraw[fill=black!15,rounded corners, dotted, line width=0.001] (0,-1.5) ++(-0.5,-0.25) rectangle ++(1,0.5);
		\filldraw[fill=black!15,rounded corners, dotted, line width=0.001] (4.5,-1.5) ++(-0.5,-0.25) rectangle ++(1,0.5);
		\node (9) at (0,-1.5) {\nscale{$h_{1,0}$}};
		\node (10) at (1.5,-1.5) {\nscale{$h_{1,1}$}};
		
		\node (11) at (3,-1) {\nscale{$h_{1,2}$}};
		\node (12) at (3,-2) {\nscale{$h_{1,3}$}};

		\node (13) at (4.5,-1.5) {\nscale{$h_{1,4}$}};
		\node (14) at (6,-1.5) {\nscale{$h_{1,5}$}};
		
		\node (15) at (7.5,-1) {\nscale{$h_{1,6}$}};
		\node (16) at (7.5,-2) {\nscale{$h_{1,7}$}};
		
		\node (17) at (9,-1.5) {\nscale{$h_{1,8}$}};
		\node () at (9.7,-1.5) {\scalebox{0.9}{$ \Biggl\rbrack H_1$}};
		
		\node () at (0,-2) {\nscale{$\vdots$}};
		\node () at (9,-2) {\nscale{$\vdots$}};
		\node (m) at (0,-2.3) {\nscale{$$}};
		\node (n) at (9,-2.3) {\nscale{$$}};
		
		\filldraw[fill=black!15,rounded corners, dotted, line width=0.001] (0,-3.7) ++(-0.5,-0.25) rectangle ++(1,0.5);
		\filldraw[fill=black!15,rounded corners, dotted, line width=0.001] (4.5,-3.7) ++(-0.5,-0.25) rectangle ++(1,0.5);
		\node (18) at (0,-3.7) {\nscale{$h_{n,0}$}};
		\node (19) at (1.5,-3.7) {\nscale{$h_{n,1}$}};
		
		\node (20) at (3,-3.2) {\nscale{$h_{n,2}$}};
		\node (21) at (3,-4.2) {\nscale{$h_{n,3}$}};

		\node (22) at (4.5,-3.7) {\nscale{$h_{n,4}$}};
		\node (23) at (6,-3.7) {\nscale{$h_{n,5}$}};
		
		\node (24) at (7.5,-3.2) {\nscale{$h_{n,6}$}};
		\node (25) at (7.5,-4.2) {\nscale{$h_{n,7}$}};
		
		\node (26) at (9,-3.7) {\nscale{$h_{n,8}$}};
		\node () at (9.7,-3.7) {\scalebox{0.9}{$ \Biggl\rbrack H_{n}$}};

	\graph {
			(0) ->["$\escale{$k$}$"] (1);
			
			(1) ->[pos=0.7,"\escale{$z_0$}"](2) ->[pos=0.35,"\escale{$v_0$}"](4);
			(1) ->[pos=0.7,swap, "\escale{$z_1$}"](3) ->[pos=0.35,swap, "\escale{$v_1$}"](4);
			
			(4) ->["$\escale{$k$}$"](5);
		
			(5) ->[pos=0.7,"\escale{$w_0$}"](6) ->[pos=0.35,"\escale{$z_0$}"](8);
			(5) ->[pos=0.7,swap, "\escale{$w_1$}"](7) ->[pos=0.35,swap,"\escale{$z_1$}"](8);
			
			(9) ->["$\escale{$k_0$}$"](10);
			
			(10) ->[pos=0.7,"\escale{$z_2$}"](11) ->[pos=0.35,"\escale{$v_2$}"](13);
			(10) ->[pos=0.7,swap, "\escale{$z_3$}"](12) ->[pos=0.35,swap, "\escale{$v_3$}"](13);
			
			(13) ->["$\escale{$k_1$}$"](14);
		
			(14) ->[pos=0.7,"\escale{$w_2$}"](15) ->[pos=0.35,"\escale{$z_2$}"](17);
			(14) ->[pos=0.7,swap, "\escale{$w_3$}"](16) ->[pos=0.35,swap,"\escale{$z_3$}"](17);
			
			(0) ->[swap,"$\escale{$r_0$}$"](9);
			(8) ->["$\escale{$a_0$}$"](17);
			
			(m) ->[swap, "$\escale{$r_{n-1}$}$"](18);
			(n) ->["$\escale{$a_{n-1}$}$"](26);

			(18) ->["$\escale{$k_{3n-3}$}$"](19);
			
			(19) ->[pos=0.7,"\escale{$z_{2n}$}"](20) ->[pos=0.35,"\escale{$v_{2n}$}"](22);
			(19) ->[pos=0.7,swap, "\escale{$z_{2n+1}$}"](21) ->[pos=0.3,swap, "\escale{$v_{2n+1}$}"](22);
			
			(22) ->["$\escale{$k_{3n-2}$}$"](23);
		
			(23) ->[pos=0.7,"\escale{$w_{2n}$}"](24) ->[pos=0.35,"\escale{$z_{2n}$}"](26);
			(23) ->[pos=0.7,swap, "\escale{$w_{2n+1}$}"](25) ->[pos=0.3,swap,"\escale{$z_{2n+1}$}"](26);
			
	};
	\end{scope}
	
	\begin{scope}[xshift = 1cm, yshift = -7cm]
		\node (D) at (-2,0) {b)};
		
		\filldraw[fill=black!15,rounded corners, dotted, line width=0.001] (90:1.25) ++(-0.5,-0.25) rectangle ++(1,0.5);
		\filldraw[fill=black!15,rounded corners, dotted, line width=0.001] (306:1.25) ++(-0.5,-0.25) rectangle ++(1,0.5);
		\filldraw[fill=black!15,rounded corners, dotted, line width=0.001] (162:1.25) ++(-0.5,-0.25) rectangle ++(1,0.5);
		\node (d0) at (90:1.25) {\nscale{$d_{j,0}$}};
		\node (d1) at (18:1.25) {\nscale{$d_{j,1}$}};	
		\node (d2) at (306:1.25) {\nscale{$d_{j,2}$}};
		\node (d3) at (234:1.25) {\nscale{$d_{j,3}$}};
		\node (d4) at (162:1.25) {\nscale{$d_{j,4}$}};

	\graph {
			(d0) ->["$\escale{$k_{3j+1}$}$"](d1);
			(d1) ->["\escale{$v_{2j}$}"](d2);
			(d2)->["\escale{$k_{3j}$}"](d3);
			(d3) ->["\escale{$v_{2j+1}$}"](d4);
			(d4) ->["\escale{$w_{2j}$}"](d0);
			(d4) ->["\escale{$k_{3j+2}$}"](d1);
			(d1) ->[swap,"\escale{$a_j$}"](d3);
	};
	\end{scope}
	
	\begin{scope}[xshift = 4.5cm, yshift = -7.5cm]
		\node (X) at (-1,0.55) {c)};

		\filldraw[fill=black!15,rounded corners, dotted, line width=0.001] (0.65,0) ++(-1.15,-0.25) rectangle ++(2.3,0.5);
		\node (b0) at (0,0) {\nscale{$b_{q,0}$}};
		\node (b1) at (1.1,1.1) {\nscale{$b_{q,1}$}};
		\node (b2) at (1.3,0) {\nscale{$b_{q,2}$}};
		\node (b3) at (2.6,0) {\nscale{$b_{q,3}$}};	
		
	\graph {
			(b0)->["$\escale{$k_{q_1}$}$"](b1);
			(b0)->[ swap,"\escale{$c_q$}"](b2);
			(b2) ->[swap, "$\escale{$k_{q_2}$}$"](b3);
	};
	\end{scope}

	\begin{scope}[xshift = 9.5cm, yshift = -7cm]
		\node (X) at (-1,0) {d)};

		\filldraw[fill=black!15,rounded corners, dotted, line width=0.001] (2,0) ++(-0.5,-1.25) rectangle ++(1,2.5);
		\node (x1) at (0,0) {\nscale{$x_{i,1}$}};
		\node (x0) at (0,1) {\nscale{$x_{i,0}$}};
		\node (x2) at (0,-1) {\nscale{$x_{i,2}$}};	
		
		\node (x4) at (2,0) {\nscale{$x_{i,4}$}};
		\node (x3) at (2,1) {\nscale{$x_{i,3}$}};
		\node (x5) at (2,-1) {\nscale{$x_{i,5}$}};

	\graph {
			(x0) ->[swap,"$\escale{$c_{4i}$}$"](x1);
			(x1)->[ swap,"\escale{$c_{4i+1}$}"](x2);
			
			(x3) ->["$\escale{$c_{4i+2}$}$"](x4);
			(x4)->[ "\escale{$c_{4i+3}$}"](x5);
			
			(x0) ->["\escale{$X^\alpha_i$}"](x3);
			(x1) ->["\escale{$X^\beta_i$}"](x4);
			(x2) ->["\escale{$X^\gamma_i$}"](x5);
	};
	\end{scope}

\begin{scope}[yshift = -10cm]
\node at (-1,0) {e)};
\foreach \i in {0,...,5} {\coordinate (\i) at (\i*1.9cm,0);}
\foreach \i in {0} {\fill[black!15, rounded corners] (\i) +(-0.5,-0.3) rectangle +(0.5,0.4);}
\fill[black!15, rounded corners] (3) +(-0.5,-0.3) rectangle +(2.4,0.4);
\foreach \i in {0,...,5} {\node (p\i) at (\i) {\nscale{$t_{i,0,\i}$}};}
\graph { (p0) ->["\escale{$k_{18i+2}$}"] (p1) ->["\escale{$X_a^i$}"] (p2) ->["\escale{$\tilde{X}_b^i$}"] (p3) ->["\escale{$X_c^i$}"] (p4) ->["\escale{$k_{18i+11}$}"] (p5);};
\end{scope}
\begin{scope}[yshift = -11.5cm]
\node at (-1,0) {f)};
\foreach \i in {0,...,4} {\coordinate (\i) at (\i*1.9cm,0);}
\foreach \i in {0} {\fill[black!15, rounded corners] (\i) +(-0.5,-0.3) rectangle +(0.5,0.4);}
\fill[black!15, rounded corners] (2) +(-0.5,-0.3) rectangle +(2.4,0.4);
\foreach \i in {0,...,4} {\node (p\i) at (\i) {\nscale{$t_{i,1,\i}$}};}
\graph { (p0) ->["\escale{$k_{18i+5}$}"] (p1) ->["\escale{$X_b^i$}"] (p2) ->["\escale{$p_{i}$}"] (p3) ->["\escale{$k_{18i+14}$}"] (p4);};
\end{scope}
\begin{scope}[yshift = -13cm]
\node at (-1,0) {g)};
\foreach \i in {0,...,4} {\coordinate (\i) at (\i*1.9cm,0);}
\foreach \i in {0} {\fill[black!15, rounded corners] (\i) +(-0.5,-0.3) rectangle +(0.5,0.4);}
\fill[black!15, rounded corners] (2) +(-0.5,-0.3) rectangle +(2.4,0.4);
\foreach \i in {0,...,4} {\node (p\i) at (\i) {\nscale{$t_{i,2,\i}$}};}
\graph { (p0) ->["\escale{$k_{18i+8}$}"] (p1) ->["\escale{$\tilde{X}_b^i$}"] (p2) ->["\escale{$p_{i}$}"] (p3) ->["\escale{$k_{18i+17}$}"] (p4);};
\end{scope}
\end{tikzpicture}

\caption{
The gadgets of $U^\varphi$ with their respective fractions of a key region.
In every example, the gray marked states are part of the key region while the unmarked are outside.
a) The headmaster $H$ with submodules $H_0, \dots, H_n$ where $n = 14m-1$.
b) $D_j$, one of the $14m$ duplicators that provide the $14m$ key copies.
c) $B_q$, one out of $4m$ barters trading $8m$ key copies for $4m$ consistency events.
Here, $q_1 = 6q + 18m +1$ and $q_2 = q_1 + 3$.
d) The variable manifolder for $X_i$ synchronizing three variable events and using four consistency events.
Together, the $m$ variable manifolders consume all $4m$ available consistency events.
e-g) The translator $T_i$ consisting of $T_{i,0}$ (e), $T_{i,1}$ (f), and $T_{i,2}$ (g).
Using six key copies, $T_i$ implements the clause $K_i$.
All $m$ translators together consume the remaining $6m$ key copies.
}
\label{fig:gadgets}
\end{figure}

In the following, we establish the proof of the announced equivalence between the existence of a one-in-three model for $\varphi$ and the ESSP of $U^\varphi$, respectively, the feasibility of $U^\varphi$.
We start with a formalization of the properties of key regions:
\begin{lemma}
\label{lem:key_region_2-grade_2-fold_ESSP_Feasibility}
If $R$ is a key region of $U^\varphi$ inhibiting $k$ at $h_{0,8}$, that is, where, without loss of generality, $sig_{R}(k) = -1$ and $h_{0,8} \not\in R$ then
\begin{enumerate}
\item\label{lem:key_region_2-grade_2-fold_ESSP_Feasibility:item1}
for all $j \in \{0, \dots, 14m-1\}$ the region contains $h_{j,0}, h_{j,4}$ and excludes $h_{j,1}$, $h_{j,2}$, $h_{j,3}$, $h_{j,5}$, $h_{j,6}$, $h_{j,7}$, $h_{j,8}$,
\item\label{lem:key_region_2-grade_2-fold_ESSP_Feasibility:item2}
for all $j \in \{0, \dots, 14m-1\}$ the region contains $d_{j,0}, d_{j,2}, d_{j,4}$ and excludes $d_{j,1}, d_{j,3}$,
\item\label{lem:key_region_2-grade_2-fold_ESSP_Feasibility:item3}
all key copies exit, that is, for all $j \in \{0, \dots, 14m-1)$ the events $k_{3j}$, $k_{3j+1}$, $k_{3j+2}$ have a negative signature,
\item\label{lem:key_region_2-grade_2-fold_ESSP_Feasibility:item4}
for all $q \in \{0, \dots, 4m-1\}$ the region contains $b_{q,0}, b_{q,2}$, excludes $b_{q,1}, b_{q,3}$, and defines $sig_{R}(c_q) = 0$,
\item\label{lem:key_region_2-grade_2-fold_ESSP_Feasibility:item5}
for all $i \in \{0, \dots, m-1\}$ the variable $X_i$, which occurs in clauses $K_\alpha$, $K_\beta$, $K_\gamma$, is represented by events $X^\alpha_i$, $X^\beta_i$, $X^\gamma_i$, that is, they have the same signature
\[sig_{R}(X^\alpha_i) = sig_{R}(X^\beta_i) = sig_{R}(X^\gamma_i), \text{ and}\]
\item\label{lem:key_region_2-grade_2-fold_ESSP_Feasibility:item6}
for all $i \in \{0, \dots, m-1\}$ the clause $K_i = \{X_a, X_b, X_c\}$ is realized in translator $T_i$ by making exactly one of the events $X^i_a, X^i_b, X^i_c$ enter while the other two obey. 
\end{enumerate}
\end{lemma}
\begin{proof}
Observe that the demonstrated regions of Figure~\ref{fig:gadgets} introduce a key region.

We show (\ref{lem:key_region_2-grade_2-fold_ESSP_Feasibility:item1}-\ref{lem:key_region_2-grade_2-fold_ESSP_Feasibility:item3}) by induction over $j$: 
For a start, let $j = 0$:
From $sig_R(k)=-1$ we have $R(h_{0,0})=R(h_{0,4})=1$ and $R(h_{0,1})=0, R(h_{0,5})=0$.
By $R(h_{0,1})=0$ we get $sig_R(z_0), sig_R(z_1)\in \{ 0,1\}$ and by $R(h_{0,8})=0$ we get $sig_R(z_0), sig_R(z_1)\in \{ -1,0\}$ resulting in $sig_R(z_0)=sig_R(z_1)=0$.
This implies $R(h_{0,2})=R(h_{0,3})=0$ and $R(h_{0,6})=R(h_{0,7})=0$ making $v_0,v_1$ entering and $w_0,w_1$ obeying.
From the entering signature of $v_0,v_1$ we have that $R$ includes $d_{0,2},d_{0,4}$ and excludes $d_{0,1},d_{0,3}$. 
As $R(d_{0,4})=1$ and $w_0$ obeys, we get $R(d_{0,0})=1$.
By $R(d_{j,0})=R(d_{j,4})=1$ and $R(d_{j,1})=0$ we obtain that $k_1,k_2$ exit.
By $R(d_{j,1})=R(d_{j,1})=0$ and $R(d_{j,2})=1$ we have that $k_0$ exits and $a_0$ obeys.

Now assume that (\ref{lem:key_region_2-grade_2-fold_ESSP_Feasibility:item1}-\ref{lem:key_region_2-grade_2-fold_ESSP_Feasibility:item3}) hold for all indices less $j$:
As $R(h_{j-1,8})=0$ and $R(a_{j-1})=0$ we have that $R(h_{j,8})=0$.
Furthermore, we have the exiting $k_{3(j-1)}, k_{3(j-1)+1}$.
Hence, we basically have the same situation as in $H_0$ and $D_0$.
Consequently, a similar argumentation as for the induction start yields that $R$ contains exactly the states $h_{j,0}, h_{j,4}$ of $H_j$ and $d_{j,0}, d_{j,2}, d_{j,4}$ of $D_j$.
This makes the vice events $v_{2j}, v_{2j+1}$ enter and the wire events $w_{2j}, w_{2j+1}$ obey.
Moreover, $D_j$ lets the key copies $k_{3j+1}$, $k_{3j+2}$ have a negative signature and $a_j$ obeys.
This concludes the induction step.

The statement (\ref{lem:key_region_2-grade_2-fold_ESSP_Feasibility:item4}) is trivially implied by (\ref{lem:key_region_2-grade_2-fold_ESSP_Feasibility:item3}).

For (\ref{lem:key_region_2-grade_2-fold_ESSP_Feasibility:item5}), it follows immediately from the obedience of the consistency events, given in (\ref{lem:key_region_2-grade_2-fold_ESSP_Feasibility:item4}), that $R(x_{i,0})=R(x_{i,1})=R(x_{i,2})$ and $R(x_{i,3})=R(x_{i,4})=R(x_{i,5})$.
This implicates $R(x_{i,3})-R(x_{i,0})=R(x_{i,4})-R(x_{i,1})=R(x_{i,5})-R(x_{i,2})$.
Consequently, the variable events $X^\alpha_i,X^\beta_i,X^\gamma_i$ have the same signature.

Finally, (\ref{lem:key_region_2-grade_2-fold_ESSP_Feasibility:item6}) is shown analogously to Lemma \ref{lemma:KeyRegionT} using signature aggregation.
\end{proof}
Lemma \ref{lem:key_region_2-grade_2-fold_ESSP_Feasibility} states that the structure of a key region defines a model of $\varphi$.
That is why we can say that the existence of key region for $U^\varphi$ implies the one-in-three satisfiablity of $\varphi$:
\begin{lemma}
\label{lem:Key_region_then_Model_2-grade_2-fold_ESSP_Feasibility}
If there is a key region of $U^\varphi$ then $\varphi$ has a one-in-three model.
\end{lemma}

\begin{proof}
If $R_{key}$ is a key region of $U^\varphi$ then, by Lemma \ref{lem:key_region_2-grade_2-fold_ESSP_Feasibility}, we obtain for all variables $X_i$ and their three occurrences in clauses $K_\alpha, K_\beta, K_\gamma$ that the respective events have
\[sig_{R_{key}}(X^\alpha_i)=sig_{R_{key}}(X^\beta_i)=sig_{R_{key}}(X^\gamma_i).\]
Moreover, for every clause $K_i = \{X_a, X_b, X_c\}$ it is true that exactly one of the events $X^i_a, X^i_b, X^i_c$ has a positive signature while the other two obey.
Then, a variable $X_i$ is put into set $M$ if and only if their synchronized events $X^\alpha_i, X^\beta_i, X^\gamma_i$ have positive signature.
It clearly follows for all clauses $K_i$ that $|M \cap K_i| = 1$ which makes $M$ a one-in-three model of $\varphi$.
\end{proof}
The other way around, the proof of the equivalence requires to derive a key region from a one-in-three model.
We have to argue that working our way backwards through the construction ends up in a region that inhibits $k$ at the key state.
\begin{lemma}
\label{lem:Model_the_Key_Region_2-grade_2-fold_ESSP_Feasibility}
If $\varphi$ has a one-in-three model then there is a key region of $U^\varphi$.
\end{lemma}

\begin{proof}
Let $M \subseteq V(\varphi)$ be a one-in-three model of $\varphi$.
We progressively build a region $R$ by following the requirements of every individual gadget:
First of all, for every variable $X_i$ with occurrences in clauses $K_\alpha, K_\beta, K_\gamma$ we take care that $sig_R(X^i_\alpha) = sig_R(X^i_\beta) = sig_R(X^i_\gamma) = M(X_i)$ where we define $M(X_i) = 1$ if $X_i \in M$ and $M(X_i) = 0$, otherwise.
To this end, we create a sub region $R^{X_i}$ for TS $X_i$ and let $x_{i,3}, x_{i,4}, x_{i,5} \in R^{X_i}$.
Moreover, we let $x_{i,0}, x_{i,1}, x_{i,2} \not\in R^{X_i}$ if and only if $X_i \in M$.
Notice that this makes the consistency events $c_{4i}, c_{4i+1}, c_{4i+2}, c_{4i+3}$ obey.
As different variable manifolders do not share states or events, the regions $R^{X_0}\dots,R^{X_{m-1}}$ are pairwise compatible.

For every clause $K_i = \{X_a, X_b, X_c\}$ the model $M$ selects exactly one variable.
By the $M$-conform construction of $R^{X_a}, R^{X_b}, R^{X_c}$ we get that exactly one of the events $X^i_a, X^i_b, X^i_c$ enters and the others obey.
We let the key copies of $T_i$ be exiting and easily generate a sub region $R^{T_i}$ as in Section~\ref{sec:The_Hardness_of_Linear_3Feasibility_and Linear_3ESSP}.
Observe that unions $T_i$ and $T_j$ share no event if $i\not=j$, which makes $R^{T_0},\ldots,R^{T_{m-1}}$ pairwise compatible.
As the variable events are selected in compliance with the variable manifolders and as translators and manifolders do not share further events, their sub regions are also compatible.

Next, as all consistency events are obeying, we can define a sub region $R^{B_q}$ of $B_q$ for every $q\in \{0,\dots,4m-1\}$ in compliance to Figure~\ref{fig:gadgets}~c).
This makes the used key copies exiting.
As different barters have no event in common, share only the obeying consistency events with variable manifolders and share no events at all with translators, their regions are pairwise compatible.

Headmaster and duplicators only share key copies with translators and barters.
As we have defined all key copies exiting and as the provided sub regions meet the conditions given in Lemma \ref{lem:key_region_2-grade_2-fold_ESSP_Feasibility}, we can use the lemma as a description on how to build compatible sub regions $R^{H}, R^{D_0}, \dots, R^{D_{14m-1}}$ for headmaster and duplicators.
Altogether, we have that 
\[
R=R^{H} \cup R^{D_0} \cup \dots \cup R^{D_{14m-1}} \cup R^{B_0} \cup \dots \cup R^{B_{4m-1}} \cup R^{X_0} \cup \dots \cup R^{X_{m-1}} \cup R^{T_0} \cup\dots \cup R^{T_{m-1}}
\]
is a region of $U^\varphi$ inhibiting $k$ at $h_{0,8}$.
\end{proof}
The proof of Theorem \ref{thm:twograde_twofold_NPcompleteness} is based on the previous lemmas, that is, that $\varphi$ has a one-in-three model $M$ if and only if there is a key region for $U^\varphi$.
In particular, the \emph{if} direction of this statement follows directly from Lemma \ref{lem:Key_region_then_Model_2-grade_2-fold_ESSP_Feasibility}.
Reversely, having $M$, Lemma \ref{lem:Model_the_Key_Region_2-grade_2-fold_ESSP_Feasibility} only allows to inhibit $k$ at the key state.
For the remaining events $e$ and states $s$ of $U^{\varphi}$ we have to show that $e$ is inhibitable at $s$:
\begin{lemma}
\label{lem:Secondary_Proofs_for_the_Hardness_of_2grade_2-ESSP}
If $\varphi$ has a one-in-three model, then there is region inhibiting $e$ at $s$ for all events $e \in E(U^\varphi)$ and all states $s \in S(U^\varphi)$ that fulfill $(e,s) \not= (k,h_{0,8})$ and $\neg(s\edge{e})$.
\end{lemma}
The proof of Lemma \ref{lem:Secondary_Proofs_for_the_Hardness_of_2grade_2-ESSP} is very technical which makes it tedious and not much edifying.
Therefore, we split it into multiple lemmas and move it to Section~\ref{sec:Secondary_Proofs}.
 
Finally, we want to join the TSs of $U^\varphi$ to obtain a combined TS $A(U^\varphi)$.
See Figure~\ref{fig:gadgets} to be convinced that the initial state of every gadget has at most one predecessor state and that we can always find a state with at most one successor state.
More precisely, for $H$ we select $t^H = h_{0,8}$, every duplicator $D_j$ designates $t^{D_j} = d_{j,0}$, the barters $B_q$ use $t^{B_q} = b_{q,1}$, variable manifolders $X_i$ have $t^{X_i} = x_{i,5}$, and the translator TSs apply $t^{T_{i,0}} = t_{i,0,5}$, $t^{T_{i,1}} = t_{i,1,4}$, and $t^{T_{i,2}} = t_{i,2,4}$.
Using these terminal states, the definition of a joining makes sure that $A(U^\varphi)$ does not exceed the state degree of two.
Moreover, as $U^\varphi$ is $2$-fold and $A(U^\varphi)$ just introduces additional unique connector events, $A(U^\varphi)$ is a $2$-grade $2$-fold TS.

To transfer the hardness of ESSP for $U^\varphi$ to the joining $A(U^\varphi)$, feasibility and the ESSP should again be the same problem for generated unions $U^\varphi$.
This would allow to use Lemma~\ref{lem:union_validity} in order to transfer feasibility, and by that also the ESSP, from $U^\varphi$ to $A(U^\varphi)$.
This time, however, we cannot use the linearity of the used gadget TSs.
However, the design of the union $U^\varphi$ of this section guarantees the following lemma:
\begin{lemma}
\label{lem:alwaysSSP}
If the union $U^\varphi$ has the ESSP then it has the SSP.
\end{lemma}
The proof of Lemma~\ref{lem:alwaysSSP} bases crucially on the fact that, besides of the zeros in $H$, all events of a given TS $A$ of $U^\varphi$ are unique in $A$.
Additionally, there is at most one state without successor in $A$.
Therefor, given two states $s,s'$ at least one of them, say $s$, is the source of an event $e$, that is, $s\edge{e}$.
As $e$ is inhibitable at $s'$ by a region $R$ of $U^\varphi$, that is, $sig_R(e)=-1$ and $R(s')=0$, as $s$ is a source of $e$, we have that $R$ separates $s$ and $s'$.
This argument doesn't work for the sources of the zeros in $H$.
By proving Lemma~\ref{lem:Secondary_Proofs_for_the_Hardness_of_2grade_2-ESSP} in Section~\ref{sec:Secondary_Proofs} we gain appropriate regions for these remaining states.
As we need to refer to these regions, we move the detailed proof of Lemma~\ref{lem:alwaysSSP} to Section~\ref{sec:Secondary_Proofs}, too.

Hence, if $U^\varphi$ has the ESSP then it is automatically feasible.
Reversely, if it is feasible, it has the ESSP by definition.
This makes both problem equivalent, again.
If $\varphi$ has a one-in-three model then $U^\varphi$ has the ESSP by Lemmas \ref{lem:Model_the_Key_Region_2-grade_2-fold_ESSP_Feasibility} and \ref{lem:Secondary_Proofs_for_the_Hardness_of_2grade_2-ESSP}, which, by the additional SSP from Lemma \ref{lem:alwaysSSP}, implies that $U^\varphi$ is feasible.
By Lemma \ref{lem:union_validity}, $A(U^\varphi)$ is feasible and, therefore, has the ESSP.

Reversely, let $A(U^\varphi)$ have the ESSP.
As $U^\varphi$ has the SSP, Lemma \ref{lem:union_validity} provides $A(U^\varphi)$ with the SSP making it feasible.
The use of Lemma~\ref{lem:union_validity} again makes $U^\varphi$ feasible and, thus, provides the union with the ESSP.
By Lemma \ref{lem:Key_region_then_Model_2-grade_2-fold_ESSP_Feasibility} we can then derive a one-in-three model for $\varphi$.

As our construction can easily be done in polynomial time we have shown Theorem \ref{thm:twograde_twofold_NPcompleteness}. 


\section{The Hardness of 2-Grade 2-SSP}
\label{sec:The_Hardness_of_2-grade_2-SSP}

Here we prove the hardness of $2$-grade $2$-SSP by a reduction using linear $3$-SSP, shown NP-complete in Section~\ref{sec:The_Hardness_of_Linear_3-SSP}, as a source:
\begin{theorem}
The SSP is NP-complete on $g$-grade $k$-fold transition systems for all $g \geq 2$ and all $k \geq 2$.
\end{theorem}
\begin{proof}
For a linear $3$-fold TS $A=(S,E,\delta,s_0)$, we construct a union $U^A$ of $2$-grade $2$-fold TSs such that the SSP holds for $A$ if and only if it holds for $U^A$.
By $E_3$ we denote the subset of $E$ containing all events of $A$ that actually occur three times in $A$.
For each of these events $e \in E_3$, we construct a so-called $e$-duplicator $D_e$ which provides three copies $e^0, e^1, e^2$ of $e$ that, in the following construction, replace the occurrences of $e$ in $A$:

\begin{tikzpicture}
	\begin{scope}
		\node (X) at (-0.75,0) {$D_e=$};
		\node (2) at (0,0) {\nscale{$d_{e,2}$}};
		\node (0) at (0,1) {\nscale{$d_{e,0}$}};
		\node (4) at (0,-1) {\nscale{$d_{e,4}$}};	
		\node (3) at (2,0) {\nscale{$d_{e,3}$}};
		\node (1) at (2,1) {\nscale{$d_{e,1}$}};
		\node (5) at (2,-1) {\nscale{$d_{e,5}$}};
	\end{scope}
	\graph {
			(0) ->["\escale{$e^0$}"] (1);
			(0)->[swap, "\escale{$a^e_0$}"] (2);
			
			(2) ->["$\escale{$e^1$}$"](3);
			(2)->[swap, "\escale{$a^e_1$}"](4);
			
			(1) ->["\escale{$a^e_0$}"](3);
			(3) ->["\escale{$a^e_1$}"](5);
			(4) ->["\escale{$e^2$}"](5);
	};
\end{tikzpicture}

The idea of $D_e$ is that, for every region, the accordance events $a^e_0$, $a^e_1$ enforce the same signature for all three $e$-copies.
Furthermore, we construct the $2$-fold modification $A^{\text{2fold}}$ of $A$ as follows:
For every $e\in E_3$ occurring at the transitions $u_0\edge{e}v_0,u_1\edge{e}v_1$ and $u_2\edge{e}v_2$ in $A$, we replace these transitions by $u_0\edge{e^0}v_0,u_1\edge{e^1}v_1$ and $u_2\edge{e^2}v_2$.
All other transitions of $A^{\text{2fold}}$ remain unchanged.

If $E_3=\{e_1,\ldots,e_n\}$, our construction results in $U^A=U(A^{\text{2fold}},D_{e_1},\dots,D_{e_n})$.
It is easy to see, that a region $R$ of $A$ can basically be regarded as a region of $A^{\text{2fold}}$ where the copies of $3$-fold events get the same signature.
Reversely, a region $R'$ of $A^{\text{2fold}}$ where the copies of $3$-fold event $e$ have the same signature is effectively a region of $A$ where $e$ gets the signature from its copies.
It remains to argue that no region of $U^A$ assigns different signatures to copies of the same $3$-fold event.
This is assured by the accordance events, that, for every region of $U^A$, enforce the same signature to the copies of $3$-fold events.
At this point it is clear that every region of $U^A$ can be reduced to a region of $A$.
Hence, if $U^A$ has the SSP then, even more, $A$ does.

For the other direction, where $A$ has the SSP, we argue (1) that every region $R'$ of $A$ can be lifted to a region of $U^A$ such that all states of $A^{\text{2fold}}$ become separable and (2) that the states of a duplicator can be mutually separated.
Firstly, we can lift $R'$ simply by copying the signature of every $3$-fold event to its copies and by letting accordance events obey.
Secondly, all duplicator state pairs, excluding $(d_{e,0},d_{e,1})$, $(d_{e,2},d_{e,3}),(d_{e,4},d_{e,5})$, are obviously separable by the accordance events.
As $A$ has the SSP, for every $3$-fold event $e$ occurring at a transition $u_0\edge{e}v_0$ the states $u_0,v_0$ are separable.
Consequently, $(d_{e,0},d_{e,1})$, $(d_{e,2},d_{e,3}),(d_{e,4},d_{e,5})$ become separable, too. 

As $A^{\text{2fold}}$ and duplicators share only event copies and as different duplicators have no events in common we can easily join $U^A$ in a $2$-grade $2$-fold TS that has the SSP if and only if $A$ does.
This and the observation that our construction clearly can be done in polynomial time completes the proof.
\end{proof}


\section{The Tractability of Linear 2-SSP}
\label{sec:The_Tractability_of_linear_2-SSP}

As shown in Section \ref{sec:The_Hardness_of_Linear_3-SSP}, deciding if a linear 3-fold TS has the SSP is NP-complete.
In this section, we show that linear SSP becomes tractable if input is $2$-fold.
To this end, we provide an easy to check property that is equivalent to SSP for linear $2$-fold TSs and based on the following notion of \textit{exact 2-fold} subsequences:
Let $A=s_0\edge{e_1}\dots\edge{e_t} s_t$ be a linear transition system.
For $0\leq i< j\leq t$ the subsequence of $A$ starting at $i$ and ending at $j$ is defined as $A^i_j=s_i\edge{e_{i+1}}\dots\edge{e_j}s_j$.
We call $A^i_j$ an exact 2-fold subsequence if each event of $A^i_j$ occurs exactly twice in $A^i_j$, that is, if for each $i\leq u < j$ with $u\edge{e}$ there is exactly one $v\not= u$ having $i \leq v < j$ and $v\edge{e}$.

\begin{theorem}
\label{theorem:Tractability_of_linear_2-SSP}
A linear 2-fold TS $A$ has the SSP if and only if $A$ has no exact 2-fold subsequence.
\end{theorem}

\begin{proof}

\textit{Only-if:} 
Let $A=s_0\edge{e_1}\dots\edge{e_t} s_t$ be a linear 2-fold TS that has the SSP.
Assume $A^i_j=s_i\edge{e_{i+1}}\dots\edge{e_j}s_j$ is an exact 2-fold subsequence of $A$ where $\{e_{i_1},\ldots,e_{i_{q}}\}$ are the $q=\frac{j-i}{2}$ events occurring in $A^i_j$.
As $A$ has the SSP, the states $s_i,s_j$ are separable. 
Without lost of generality let $R$ be a separating region of $A$ with signature $sig_R$, such that $R(s_i)=1$ and $R(s_j) =0$.
Using signature aggregation, we get that $\sum_{k=i+1}^j sig_R (e_k)=R(s_j)-R(s_i)=-1$ contradicts to $2(sig_R(e_{i_1})+\dots +sig_R(e_{i_q}))\equiv\ \text{$0$ mod $2$}$.
Hence, $A^i_j$ cannot exist.\newline
\textit{If:} Let $A=s_0\edge{e_1}\dots\edge{e_t} s_t$ be a linear 2-fold TS that has no exact 2-fold subsequence and let $s_i,s_j$ be two arbitrary states of $A$ such that $i< j$.
We proof the states to be separable by an exhausted case analyses:
\begin{description}
	\item[Case 1]  There is an unique event $ e\in E$ and a state $u$ in $A$ such that $ i \leq u < j$ and $u\edge{e}$. 
	Then a separating region $R$ is defined by the signature $exit_R=\{e\}$.

	\item [Case 2] 
	Otherwise, as $A$ has no exact 2-fold subsequence, there has to be an event $e\in E$ and states $s_u,s_v$ with $s_u \edge{e}, s_v\edge{e}$ such that $u < i \leq v < j$ or $ i \leq u < j\leq v$.
	Then we have an event between $s_i,s_j$ whose second occurrence is the \glqq leftmost\grqq. 
	More exactly, we have an event $e_{min}$ and states $s_a,s_b$ with $s_a\edge{e_{min}},s_b\edge{e_{min}}$ such that $a < i\leq b < j$ and for all events $e$ and states $s_u,s_v$ with $s_u \edge{e}, s_v\edge{e}$ such that $u < i \leq v < j$ we have that $a\leq u$.
	If there is an unique event $c$ and a state $s_u$ with $s_u \edge{c}$ such that $a < u < i$ or if there is an event $c$ and  states $s_u,s_v$ with $s_u \edge{c}, s_v\edge{c}$ such that $a < u < i \leq b < j < v $ or $u < a < v < i \leq b < j$.
	Then it is easy to check that a separating region $R$ is defined by the signature given by $enter_R=\{c\}$ and $exit_R=\{e_{min}\}$.
	Otherwise, all events occurring between $s_a$ and $s_i$ occur twice between $s_a$ and $s_j$, that is, for all events $c$ and states $s_u,s_v$ with $s_u \edge{e}, s_v\edge{e}$ we have if $a < u < i$ then $a < v < j$.
	As there is no exact 2-fold subsequence in $A$ and no unique event between $s_i$ and $s_j$ there has to be an event between $s_i$ and $s_j$ whose second occurrence is on the right side of $s_j$. 
	Consequently, we have an event $e_{max}$ and states $s_{a'},s_{b'}$ with $s_{a'}\edge{e_{max}},s_{b'}\edge{e_{max}}$ such that $  i\leq a' < j < b'$ and for all events $e$ and states $s_u,s_v$ with $s_u \edge{e}, s_v\edge{e}$ such that $ i \leq u < j < v$ we have that $v\leq b'$.
	As there is no exact 2-fold subsequence in $A$ there has to be an unique event $c$ and a state $s_u$ with $s_u \edge{c}$ such that $j\leq u < b'$ or there has to be an event $c$ and states $s_u,s_v$ with $s_u \edge{c}, s_v\edge{c}$ such that $ u < a < i \leq a' < j \leq v < b' $ or $ i \leq a' < j \leq u < b' < v$.	
	It is easy to check that a separating region $R$ is defined by the signature given by $enter_R=\{c\}$ and $exit_R=\{e_{max}\}$.
	\item[Case 3] 
	Finally, we have to deal with the case, when neither an event between $s_i,s_j$ occur once nor has a second occurrence on the left side of $s_i$.
	We have no $e,s_u,s_v$ such that  $u < i \leq v < j$.
	As argued above, we must have $e,s_u,s_v$ such that $s_u \edge{e}, s_v\edge{e}$ and $ i \leq u < j \leq v$.
	Again, we have an event $e_{max}$ as defined in the previous case and analogously we can argue for the separability of $s_i$ and $s_j$.
\end{description}
This analysis is comprehensive and proves the lemma.
\end{proof}
\begin{figure}
\begin{algorithmic}[1]
\renewcommand{\algorithmicprint}{\textbf{break}}
\renewcommand{\algorithmicrequire}{\textbf{Input:}}
\renewcommand{\algorithmicensure}{\textbf{Output:}}
\REQUIRE  $2$-fold TS $A=s_0\edge{e_1}\dots \edge{e_n}s_n$, $i,j: 0\leq i < j \leq n$
\ENSURE pair $(exit_R,enter_R)$ of sets of events defining a region $R$ of $A$ that separates $s_i,s_j$
\STATE preprocess $I_A$;  
\STATE $a:=-1,b:=-1$;  
\FOR{$k=i$ to $j-1$}
    \IF{$I_A(k)==-1$}
        \RETURN $(\{e_{k+1}\},\emptyset)$
    \ENDIF
\ENDFOR
\FOR{$k=0$ to $i-1$}     
	\IF{$i \leq I_A(k)\leq j-1$}
		\STATE $a:=k;\ b:=I_A(k);$
		\PRINT
	\ENDIF
\ENDFOR
\IF{$a\not= -1$}
	\FOR{$k=a+1$ to $i-1$}
		\IF{$I_A(k)==-1$ \OR $I_A(k) < a$ \OR $I_A(k) \geq j$ }
			\RETURN $(\{e_{a+1}\},\{e_{k+1}\})$
		\ENDIF
	\ENDFOR
\ENDIF
\FOR{$k=n-1$ to $j$}     
	\IF{$i \leq I_A(k)\leq j-1$}
		\STATE $a:=I_A(k);\ b:=k;$
		\PRINT
	\ENDIF
\ENDFOR
\FOR{$k=j$ to $b-1$}
	\IF{$I_A(k)==-1$ \OR $I_A(k) < i$ \OR $I_A(k) > b$ }
			\RETURN $(\{e_{b+1}\},\{e_{k+1}\})$
		\ENDIF
\ENDFOR
\RETURN $(\emptyset,\emptyset)$
\end{algorithmic}
\caption{ 
Algorithm \textsc{separator} finding, for a linear $2$-fold TS $A$ having the SSP, a separating region for given states $s_i,s_j$.
}
\label{alg:separator}
\end{figure}

It is easy to see that this property can be checked by a straight forward algorithm in time $O(\vert S\vert ^3)$.
Moreover, the proof of Theorem~\ref{theorem:Tractability_of_linear_2-SSP} implicitly motivates the algorithm \textsc{separator} in Figure \ref{alg:separator} that, given a linear $2$-fold TS $A$ that has the SSP and two states $s_i,s_j$, returns a separating region $R$ for $s_i,s_j$ with at most two non-obeying events.
For $A=s_0\edge{e_1}\dots \edge{e_n}s_n$, this algorithm uses a function $I_A(k)$ that, given an index $k \in \{0,\dots,n-1\}$, returns the index of the second occurrence of event $e_{k+1}$ if it exists and $-1$ otherwise.
This function can be preprocessed as an array in at most $O(\vert S\vert \log \vert S\vert)$ time and queried in constant time.
If the input TS has a exact $2$-fold subsequence, \textsc{separator} fails at least for non-separable states and return $(\emptyset, \emptyset)$

The correctness and termination of Algorithm~\ref{alg:separator} follows from the proof of Theorem~\ref{theorem:Tractability_of_linear_2-SSP}.
For its time complexity, we observe, that the for-loops starting at the lines $3,8,15,21,27$ query $I_A$ at most $\vert S\vert$ times.
The costs of subsequent comparisons are constant.
Consequently, after preprocessing, \textsc{separator} returns a separating region for two input states $s_i,s_j$ in $O(\vert S\vert \log \vert S\vert)$ time.

To check if $A$ has the SSP, \textsc{separator} can be called for all state pairs.
As we have to investigate $\frac{\vert S\vert(\vert S\vert -1)}{2}$ input pairs this takes $O(\vert S\vert^3)$ time as the preprocessing is done only once.


\section{Conclusion}

With the present work on feasibility, ESSP, and SSP of ENS synthesis we attempt to demonstrate the surprisingly high difficulty of this kind of Petri net synthesis, which is convincingly illustrated by Figure \ref{fig:result_overview}.
While general intractability has been known before, we show that even simultaneously fixing two obvious input parameters, the restrictions of event occurrence and state degree, has no positive effect on the complexity of synthesis or even one of the subproblems SSP and ESSP.
In fact, as we can narrow down the hardness barrier of ENS synthesis close to trivial inputs, it actually becomes tough to think of other promising parameters.
In a narrow sense, bringing intractability even to linear TSs actually makes many obvious considerations of restricting the TS graph structure futile, too.
Consequently, our results rule out many straight forward approaches from \emph{fixed parameter tractability}.


\section{Auxiliary Proofs}
\label{sec:Secondary_Proofs}

In the following two subsections we proof Lemma~\ref{lem:Secondary_Proofs_for_The_Hardness_of_Linear_3Feasibility_and Linear_3ESSP} (Subsection~\ref{sec:Secondary_Proofs_for_The_Hardness_of_Linear_3Feasibility_and Linear_3ESSP}) and Lemma~\ref{lem:Secondary_Proofs_for_the_Hardness_of_2grade_2-ESSP} (Subsection~\ref{sec:Secondary_Proofs_for_the_Hardness_of_2grade_2-ESSP}).
In both cases we present a series of lemmata that investigate successively every type of event that occur in the appropriate reduction.
In particular, for each type of event we provide inhibiting regions by specifying the signature of the non-obeying events.
More exactly, these regions are presented in tables with columns 
\begin{itemize}
	\item \textit{States}: states are listed at which the investigated event is inhibited by the current region, 
	\item \textit{Exit}: all exiting events of the current region,
	\item \textit{Enter}: all entering events of the current region,
	\item \textit{Affected TS}: transition systems in which at least one of the events listed in \textit{Exit} ot \textit{Enter} occur.
\end{itemize}
It is easy to check, that, on the one hand, the non-obeying events define a region of the union $U^\varphi$ recursively build from the affected TSs and, on the other hand, each of the listed states are excluded by the current region.
By definition, this inhibits the current event at the appropriate states in $U^\varphi$.

Moreover, it is easy to see that if $A_1,\ldots,A_n$ are TSs with event sets $E_1, \ldots, E_n$ and $B_1, \ldots, B_m$ are TSs with event sets $E_1',\ldots,E_m'$ and 
$E_i \cap E_j' = \emptyset$ holds for all $i,j$ and $R=R^{1}\cup \ldots \cup R^{n}$ is a region of $U=U(A_{1},\ldots,A_{n})$  then $R$ is a region of $U(U,B_1,\ldots,B_m)$.
We obtain the permitted signature by simply defining \textit{additionally} $sig_R(e)=0$ for all events of $e\in E_1'\cup \ldots \cup E_m'$.

With respect to a selected row of a table, this aspect allows us to consider an event $e$ which is listed in the \textit{Exit}-cell as inhibited at each state of each transition system which is not mentioned in the \textit{Affected TS}-cell. 
The following lemma states another situation when we can easily lift a region $R$ of a union $U$ to a region $R'$ of a union $U(U,A_{m+1},\dots,A_{n})$:
\begin{lemma}
\label{lemma:easyLiftedRegion}
Let $U=U(A_1,\dots,A_m)$ be a union with region $R$ and let $A_{m+1},\dots,A_{n}$ be linear TSs such that for all $i \in \{m+1,\dots, n\}$ there is at most one edge $s_i \edge{e_i} s'_i$ in $A_i$ where $sig_R(e_i)$ is defined and $sig_R(e_i) \not= 0$.
Moreover, for all $i \in \{m+1, \dots, n\}$, let $P_i$ be the set of $s_i$ and all its predecessor states and let $P_i = \emptyset$ in case $s_i$ does not exist.
Then there is a region $R'= R \cup R^{m+1} \cup \dots \cup R^{n}$ of union $U(A_1,\dots,A_m,A_{m+1},\dots,A_{n})$ with $sig_{R'} = sig_{R}$ such that, for all $i \in \{m+1, \dots, n\}$, the region is extended by $R^i=S_i \setminus P_i$ if $sig(e_i) = 1$, where $S_i$ is the state set of $A_i$, and otherwise, the region is extended by $R^i = P_i$.
\end{lemma}

\subsection{Proof for Lemma~\ref{lem:Secondary_Proofs_for_The_Hardness_of_Linear_3Feasibility_and Linear_3ESSP}} 
\label{sec:Secondary_Proofs_for_The_Hardness_of_Linear_3Feasibility_and Linear_3ESSP}

Before we start, we explain how we use indices of refreshers and duplicators in this section.
In particular, for a pair $(D_j, F_j)$ with $j\in \{0,\dots,6m-1\}$ we use the equivalent representation of the index in terms of $j = 6i+l$ where $l\in \{0,\dots,5\}$.
That relates $(D_j, F_j)$ better with the translators of clause $K_i$ that make use of events provided by $(D_j, F_j)$, namely a key event copy $k_{3(6i+l)+2}$.

\begin{lemma}
\label{lemma:KeyInhibitable}
The key event $k$ is inhibitable in $U^\varphi$.
\end{lemma}
\begin{proof}
The key event is already inhibited in the master by the key region. 
It remains to show, that $k$ is inhibitable in all the other TSs.

\begin{longtable}{ p{2cm}    p{2cm}   p{2cm}     l  }
\emph{States} & \text{Exit} & \text{Enter} & \emph{Affected TS}  \\ \hline
 $d_{0,1},d_{0,2}$	& $k,k_0$	& $o_0,h,h_0$	& $M,F_0,D_0$\\ \hline
$d_{0,0},d_{0,3}$	& $k,h_0$	& $z_0,z_1$	& $M,D_0$\\ \hline
$ d_{0,4},d_{0,14}$	& $k,k_0$	& $z_0$	& $M,D_0$
\end{longtable}
\end{proof}

\begin{lemma}
\label{lemma:Helpers}
The helper events are inhibitable in $U^\varphi$.
\end{lemma}

\begin{proof}
We start with regarding the helper event $h$ once occurring in master $M$. \newline

\begin{longtable}{ p{2cm}  p{2cm}p{2cm}p{2cm}  }

\emph{States} & \emph{\text{Exit}} & \text{Enter} & \emph{Affected TS}  \\ \hline
$m_5,\dots,m_8$ 		& $h$		&			&	$M$		\\ \hline
$m_0,m_3$			& $h,o_0$		& $k,o_1$		&	$M, F_0$	\\ \hline
$m_1,m_2$			& $h,k_0$		& $o_0,h_0$	&	$M,F_0,D_0$
\end{longtable}

We proceed with inhibiting $h_0$ which indirectly affects the master segment $M$.

\begin{longtable}{p{4cm}  p{2cm}p{2cm}p{2cm}   }

\emph{States} & \emph{\text{Exit}} & \text{Enter} & \emph{Affected TS}  \\ \hline
 $d_{0,3}, d_{0,4}, d_{0,6},\dots, d_{0,8}$ &  $h_0$ 		& $z_1$ 		& $D_0 $ \\ \hline
  $d_{0,0}, d_{0,1}$					& $h_0, h$ 	&  $ z_0,z_1$ 	& $D_0, M$  \\ \hline
$d_{0,9},\dots, d_{0,14}$				& $h_0$		& $k_0$		& $D_0, F_0$
\end{longtable}

Next, for $i\in \{0,\dots,m-1\}$ and $l\in \{0,\dots,5\}$, assuming that one of $i$ or $l$ is different from zero, we regard the helper event $h_{6i+l}$. 
For readability, we set $x=6i+l$. 

\begin{longtable}{ p{4cm}  p{2cm}p{2cm}p{2cm}   }

\emph{States} & \text{Exit} & \text{Enter} & \emph{Affected TS}  \\ \hline
 $d_{x,0},d_{x,1}, d_{x,3},d_{x,4}, \newline d_{x,6},\dots, d_{x,8}$ 	& $h_x, z_{2x-1}$ 	& $z_{2x}, z_{2x+1}$	& $D_{x-1}, D_x$ \\ \hline
 $d_{x,9},\dots, d_{x,14}$ 									& $h_x$			& $k_{3x}$		& $F_x,D_x$	
\end{longtable}

\end{proof}

\begin{lemma}
\label{lemma:KeyCopies}
The key-copy events are inhibitable in $U^\varphi$.
\end{lemma}
\begin{proof}
For $i\in \{0,\dots,m-1\}$ and $l\in \{0,\dots,5\} $ we distinguish three cases indicated by (the different types of indices of): $k_{3(6i+l)},k_{3(6i+l)+1},k_{3(6i+l)+2}$.
For abbreviation we set $x=6i+l$.
Firstly, we present regions of $U^\varphi$ which altogether inhibit $k_{3x}$ in $U^\varphi$. We observe, that $k_{3x}$ occurs in $F_x$ at state $f_{x,1}$ and in $D_x$ at the states $d_{x,0}$ and $d_{x,3}$.

\begin{longtable}{  p{3cm}  p{2cm}p{2cm}p{2cm}   }
\emph{States} & \text{Exit} & \text{Enter} & \emph{Affected TS}  \\ \hline
$d_{x,1},d_{x,2},d_{x,4},d_{x,5}$ $f_{x,2},\dots,f_{x,7}$	& $k_{3x}$		& $h_x$		& $D_x,F_x$ \\ \hline
$d_{x,6},\dots,d_{x,14}$							& $k_{3x},z_{2x-1}$	& $z_{2x}$	& $D_{x-1},D_x,F_x$ \\ \hline	
$f_{x,0}$ 										& $k_{3x}$		& $o_{2x},h_x$	& $D_{x-1}, F_x,D_x$
\end{longtable}

Secondly, we present regions of $U^\varphi$ which altogether inhibit $k_{3x+1}$ in $U^\varphi$ which occurs in  in $F_x$ at state $f_{x,3}$ and in $D_x$ at the states $d_{x,7}$ and $d_{x,10}$.  
For $i=0=l$, we have no $D_{x-1}$ (second row). 
In this case the opposite event $o_{2x+1}=o_1$ occurs in $M$ exactly once. 

\begin{longtable}{p{3cm}  p{2cm}p{2cm}p{3cm} }
\emph{States} & \text{Exit} & \text{Enter} & \emph{Affected TS}  \\ \hline
$f_{x,4},\dots,f_{x,7}$, $d_{x,0},\dots,d_{x,4},d_{x,8}$, $d_{x,11},\dots,d_{x,14}$	& $k_{3x+1}$ 			& $z_{2x+1}$				& $F_x,D_x$ \\ \hline
$f_{x,0},\dots,f_{x,2}$												& $k_{3x+1}$			& $z_{2x+1}, o_{2x+1}$		& $F_x,D_x, D_{x-1}$, $(M)$ \\ \hline
$d_{x,5},d_{x,6},d_{x,9}$												& $k_{3x+1},k_{3(x+1)}$	& $z_{2x+2},o_{2x+2}$		& $F_x,D_x,F_{x+1},D_{x+1}$
\end{longtable}
Thirdly, we present regions of $U^\varphi$ which for $l \in \{0,1,2\} $ altogether inhibit $k_{3(6i+l)+2}$ in $U^\varphi$ which occurs in $F_x$ at state $f_{x,5}$, in $D_x$ at the state $d_{x,13}$ and in $T_{i,0},T_{i,1},T_{i,2}$ at the states $t_{i,0,0},t_{i,1,0}$ or $t_{i,2,0}$, respectively.

For abbreviation, we set $x=6i+l$.

\begin{longtable}{ p{3cm}  p{2cm}p{2cm}p{4cm}  }
\emph{States} & \text{Exit} & \text{Enter} & \emph{Affected TS}  \\ \hline
$d_{x,0},\dots,d_{0,12},d_{0,14}$ $t_{i,l,1}\dots,t_{i, l, 4},(t_{i, l, 5})$	& $k_{3x+2},k_{3x+4}$	& $o_{2x+3},z_{2x+3}$	& $F_x,D_x, T_{i,l}$ $F_{x+1},D_{x+1}$\\ \hline
$f_{x,0},f_{x,2},\dots,f_{x,4}$, $f_{x,6},f_{x,7}$			& $k_{3x},k_{3x+2}$		& $o_{2x}, h_{x}$		& $F_x,D_x, T_{i,l}$, $(D_{x-1},M)$	 \\ \hline
$f_{x,1}$											& $k_{3x+2}  $			& $o_{2x+1}$			& $F_x,D_x,T_{i,l}$, $(D_{x-1},M)$
\end{longtable}

Now, we present regions of $U^\varphi$ inhibiting $k_{3(6i+3)+2}$ in $U^\varphi$ which occurs in $F_x$ at state $f_{x,5}$, in $D_x$ at the state $d_{x,13}$ and in $T_{i,0}$ at the state $t_{i,0,4}$ where $x=6i+3$.

The translator $T_i$ represents the clause $K_i=\{X_{a},X_{b},X_{c}\}$, where $a,b,c\in \{0,\dots,m-1\}$.
As $\varphi$ is cubic and monotone, there are exact two further clauses $K_j$ and $K_l$ such that $X_{c}\in K_j\cap K_l$.
Consequently, $X_{c}$ occurs exactly once in $T_{j}$ as well as in $T_{l}$.
By the uniqueness of $X_c$ in $T_j$ respectively $T_l$ and Lemma \ref{lemma:easyLiftedRegion} we need no further investigations at which state of $t_{j,0,1},t_{j,0,3},t_{j,1,1}$ or $t_{l,0,1},t_{l,0,3},t_{l,1,1}$ the event $X_c$ occurs a second respectively a third time in $U^\varphi$, for proving that the following table provides reliable regions of $U^\varphi$.
Inhibiting regions for $k_{3x+2}$ are given by

\begin{longtable}{ p{3cm}  p{2cm}p{2cm}p{4cm}  }
\emph{States} & \text{Exit} & \text{Enter} & \emph{Affected TS}  \\ \hline
$d_{x,0},\dots,d_{0,12},d_{0,14}$  $t_{i,0,0}\dots,t_{i, 0, 5}$, $t_{i, 0, 7}$	& $k_{3x+2},k_{3x+5}$	& $o_{2x+3}$, $C_i$	& $F_x,D_x, T_{i,0}, T_{i,1}$ $F_{x+1},D_{x+1}$, $(T_{j},T_{l})$\\ \hline
$f_{x,0},f_{x,2},\dots,f_{x,4}$, $f_{x,6},f_{x,7}$							& $k_{3x},k_{3x+2}$		& $o_{2x}, h_{x}$		& $F_x,D_x, T_{i,0}$, $D_{x-1}$	 \\ \hline
$f_{x,1}$														& $k_{3x+2}  $			& $o_{2x+1}$			& $F_x,D_x,T_{i,0}$, $D_{x-1}$
\end{longtable}

Finally, we represent regions of $U^\varphi$ inhibiting $k_{3(6i+l)+2}, \ l\in \{4,5\}$, in $U^\varphi$ which occurs in $F_x$ at state $f_{x,5}$, in $D_x$ at the state $d_{x,13}$ and in $T_{i,1},T_{i,2}$ at the states $t_{i,1,3}$ or $t_{i,2,3}$, respectively.
We abbreviate $x=6i+l$ and $y=l-3,l\in\{4,5\}$.
Inhibiting regions for $k_{3x+2}$ are given by

\begin{longtable}{ p{3cm}  p{2cm}p{2cm}p{4cm}  }
\emph{States} & \text{Exit} & \text{Enter} & \emph{Affected TS}  \\ \hline
$d_{x,0},\dots,d_{0,12},d_{0,14}$  $t_{i,y,0}\dots,t_{i, y, 2}$, $t_{i, y, 4}$	& $k_{3x+2},k_{3x+5}$	& $o_{2x+3}$, $p_i$	& $F_x,D_x, T_{i,1},T_{i,2},(T_{i+1,0})$,  $F_{x+1},D_{x+1}$\\ \hline
$f_{x,0},f_{x,2},\dots,f_{x,4}$, $f_{x,6},f_{x,7}$							& $k_{3x},k_{3x+2}$		& $o_{2x}, h_{x}$		& $F_x,D_x, T_{i,y}$, $D_{x-1}$	 \\ \hline
$f_{x,1}$														& $k_{3x+2}  $			& $o_{2x+1}$			& $F_x,D_x,T_{i,y}$, $D_{x-1}$
\end{longtable}
\end{proof}

\begin{lemma}
\label{lemma:Zeros}
The zero events are inhibitable in $U^\varphi$.
\end{lemma}

\begin{proof}
First, we consider the special case of $z_0$ which occurs in the master segment $M$ at $m_1$ and $m_5$ and in the duplicator $D_0$ at $d_{0,1}$.
Appropriate regions are defined by:

\begin{longtable}{ p{3cm}  p{2cm}p{2cm}p{3cm}  }
\emph{States} & \text{Exit} & \text{Enter} & \emph{Affected TS}  \\ \hline
$m_0,m_2,m_3,q,m_7$ $d_{0,2},\dots,d_{0,14}$	& $z_0$	& $k$	&	$M,D_0$ \\ \hline
$m_4,m_8,d_{0,0}$								& $z_0$	& $h,k_0$	&	$M,D_0,F_0$
\end{longtable}

Secondly, we consider for $i\in \{1,\dots,m\}$ and $l\in \{0,\dots,5\}$ (not both of them zero) the event $z_{2(6i+l)}$ occurring in the duplicators $D_{x-1}$ and $D_x$ at the states $d_{x-1,6}$, $d_{x-1,11}$ and $d_{x,1}$, respectively, where $x=6i+l$. 
The following table represent regions which altogether inhibit $z_{2x}$ in $U^\varphi$.

\begin{longtable}{ p{3cm}  p{2.7cm}p{2.7cm}p{3cm}   }
\emph{States} & \text{Exit} & \text{Enter} & \emph{Affected TS}  \\ \hline
$d_{x-1,0},\dots, d_{x-1,4}$, $d_{x-1,7},d_{x-1,8}$, $d_{x-1,12},\dots, d_{x-1,14}$, $d_{x,2},\dots,d_{x,14}$	& $z_{2x}$					& $z_{2x-1}$				& $D_{x-1},D_x$ \\ \hline
$d_{x-1,5},d_{x-1,9}$																& $z_{2x}, k_{3(x-1)}, k_{3x+1}$	& $h_{x-1},o_{2x},z_{2x+1}$	& $D_{x-1},F_{x-1}, D_x, F_x$ \\ \hline
$d_{x-1,10},d_{x,0}$																	& $z_{2x},o_{2x}$				& $k_{3x-2},k_{3x}$			& $D_{x-1},F_{x-1}, D_{x}, F_{x}$
\end{longtable}

Thirdly, we present regions which in sum inhibit $z_{2(6i+l)+1}$ in $U^\varphi$ where  $i\in \{0,\dots,m-1\}$ and $l\in \{0,\dots,5\}$.
The event $z_{2(6i+l)+1}$ occurs in the duplicator $D_{6i+l}$ at the states $d_{6i+l,4}$ and $d_{6i+l,8}$.
For abbreviation, we set $x=6i+l$.
In case of $i=l=0$, no duplicator $D_{x-1}$ and no zero $z_{2x-1}$ is defined. 
In this case, the event $z_{2_x}=z_0$ occurs in the master $M$, helper $h$ plays the same role as $z_{2x-1}$ for the other cases and we refer to this special case by setting $M$ and $h$ in brackets.

\begin{longtable}{ p{3cm}  p{2cm}p{3.2cm}p{3cm}  }
\emph{States} & \text{Exit} & \text{Enter} & \emph{Affected TS}  \\ \hline
$d_{x,0}, d_{x,2},d_{x,3}$, $d_{x,5},d_{x,6}$, $d_{x,9},\dots,d_{x,11}$									& $z_{2x+1},z_{2x}$				& $z_{2x-1}, k_{3x}, z_{2x+2}, (h)$	& $D_{x-1}, D_{x}, D_{x+1}, F_x$, $ (M)$\\ \hline
$d_{x,1},d_{x,12},\dots,d_{x,14}$, $t_{i,0,0},t_{i,0,1},t_{i,0,3}$	 	& $z_{2x+1}$					& $h_x$				& $D_x$\\ \hline
 $d_{x,7}$																				& $z_{2x+1}$					& $k_{3x+1}$					& $D_x,F_x$
\end{longtable}
\end{proof}

\begin{lemma}
\label{lemma:OpponentsInhibitable}
The opposite events are inhibitable in $U^\varphi$.
\end{lemma}
\begin{proof}
Firstly, we investigate the opposite event $o_0$ occurring in the master $M$ and refresher $F_0$ at the states $m_2$ and $f_{0,0},f_{0,4}$.

\begin{longtable}{p{3cm}  p{2cm}p{3.2cm}p{3cm}}
\emph{States} & \text{Exit} & \text{Enter} & \emph{Affected TS}  \\ \hline
$m_3,\dots,m_8$, $f_{0,1},\dots, f_{0,3}$, $f_{0,5},\dots,f_{0,7}$, $d_{0,0},\dots, d_{0,7}$	& $o_0,o_2,h_1$	& $k_1,k_3$	& $M,D_0,F_0,D_1,F_1$	\\ \hline
$m_0,m_1$, 															& $o_0,z_1$		& $z_0,k_1$	& $M, F_0, D_0, T_{0,0}$ \\ \hline
$d_{0,8},\dots,d_{0,14}$													& $o_0, h_0$		& $k_0$		& $M, F_0, D_0$
\end{longtable}

Secondly, we investigate the opposite event $o_1$ occurring in the master $M$ and refresher $F_0$ at the states $q$ and $f_{0,2},f_{0,6}$.

\begin{longtable}{ p{3.5cm}  p{2cm}p{2cm}p{3cm} }
\emph{States} & \text{Exit} & \text{Enter} & \emph{Affected TS}  \\ \hline
$m_0,m_1,m_5,m_7,m_8$, $f_{0,3},\dots,f_{0,5},f_{0,7}$	&	$o_1,h,o_2,h_1$	& $z_0,k_2,k_3$	& $M,D_0,F_0,D_1,F_1$, $T_{i,0}$ \\ \hline
$m_2,\dots,m_4$, $f_{0,0},f_{0,1}$					&	$o_1,h_0$		& $h,k_0,k_2$		& $M,F_0,D_0,T_{i,0}$ \\ \hline
$d_{0,0},\dots,d_{0,14}$							&	$o_1$			& $o_0$			& $M, F_0$
\end{longtable}

Thirdly, we present inhibiting regions for the event $o_{2x+2}$ occurring in $D_x$ and $F_{x+1}$ at the states $d_{x,9}$, $f_{x+1,0}$ and $f_{x+1,4}$, where $x=6i+l$ for $i\in \{0,\dots,m\}$ and $l\in \{0,\dots,5\}$.

\begin{longtable}{ p{3.3cm}  p{2.7cm}p{2.7cm}p{2.8cm}   }
\emph{States} & \text{Exit} & \text{Enter} & \emph{Affected TS}  \\ \hline
$d_{x,10},\dots,d_{x,12}$, $f_{x+1,1},f_{x+1,2}, f_{x+1,5}$, $f_{x+1,6}$		&	$o_{2x+2}$					& $o_{2x+3}$			& $D_x,F_{x+1}$ \\ \hline
$d_{x,3}, d_{x,4},d_{x,6},\dots,d_{x,8}$, $d_{x,13}, d_{x,14},f_{x+1,7}$		& 	$h_x,o_{2x+2},h_{x+1}$			& $z_{2x+1},k_{3(x+1)}$	& $D_x,F_{x+1},D_{x+1}$\\ \hline
$d_{x,0},\dots, d_{x,2},d_{x,5}$, $f_{x+1,3}$							&	$o_{2x+2}, o_{2x+4}, h_{x+2}$	& $k_{3x+1}, k_{3x+4},k_{3x+6}$			& $D_x,F_x,D_{x+1},F_{x+1}$, $D_{x+2}, F_{x+2}$
\end{longtable}

Finally, we investigate the state $o_{2x+3}$ occurring in $D_x$ and $F_{x+1}$ at the states $d_{x,12}$ and $f_{x+1,2}$ and $f_{x,6}$, where $x=6i+l$ for $i\in \{0,\dots,m\}$ and $l\in \{0,\dots,5\}$.
We observe, that $z_{2x+1}$ and $z_{2x+3}$ never occur simultaneously in an $T_{i,0}$.

\begin{longtable}{ p{3.8cm}  p{2.7cm}p{2.4cm}p{2.5cm}  }
\emph{States} & \text{Exit} & \text{Enter} & \emph{Affected TS}  \\ \hline
$d_{x,0},\dots,d_{x,9},d_{x,13}, d_{x,14}$, $f_{x+1,0},f_{x+1,3},f_{x+1,4}$, $f_{x+1,7}$		& $o_{2x+3}$			& $o_{2x+2}$					& $D_x,F_{x+1}$\\ \hline
$d_{x,10},d_{x,11}$														& $z_{2x+1},z_{2x+3},o_{2x+3}$	& $z_{2x+2},k_{3x+4}$	& $D_x, F_{x+1},D_{x+1}$\\ \hline
$f_{x+1,1},f_{x+1,5}$														& $o_{2x+3},h_{x+1}$	& $k_{3x+3},k_{3x+5}$			& $D_x, D_{x+1}, F_{x+1}$
\end{longtable}
\end{proof}

\begin{lemma}
\label{lemma:B-Copies}
If $i\in \{0,\dots,m-1\}$ and $K_i=\{X_a,X_b,X_c\}$ is the clause that is represented by the translator $T_i$ then the event $\tilde{X}_b$ occurring at $t_{i,0,2}$ and $t_{i,2,1}$ is inhibitable in $U^\varphi$.
\end{lemma}
\begin{proof}

The translator $T_i$ represents the clause $K_i=\{X_{a},X_{b},X_{c}\}$, where $a,b,c\in \{0,\dots,m-1\}$.
As $\varphi$ is cubic and monotone, there are exact two further clauses $K_j$ and $K_l$ such that $X_{a}\in K_j\cap K_l$.
Consequently, $X_{a}$ occurs exactly once in $T_{j}$ as well as in $T_{l}$.
By the uniqueness of $X_a$ in $T_j$ ($T_l$) and Lemma \ref{lemma:easyLiftedRegion} we need no further investigations at which state of $t_{j,0,1},t_{j,0,3},t_{j,1,1}$ ($t_{l,0,1},t_{l,0,3},t_{l,1,1}$) the event $X_a$ occurs a second respectively a third time in $U^\varphi$, for proving that the following table provides reliable regions of $U^\varphi$.

\begin{longtable}{ p{4cm}  p{2cm}p{2cm}p{3cm} }
\emph{States} & \text{Exit} & \text{Enter} & \emph{Affected TS}  \\ \hline
$t_{i,0,0},t_{i,0,1}, t_{i,0,3}\dots t_{i,0,5}$, $t_{i,2,0},t_{i,2,2},\dots,t_{i,2,4}$ 	& $\tilde{X}_b$ & $X_a,k_{18i+8}$ & $F_{6i+2},D_{6i+2}$, $T_{i,0},T_{i,2}$, $(T_{j},  T_{l})$ \\ \hline
																& $\tilde{X}_b$&					& $T_{i,0},T_{i,2}$
\end{longtable}

\end{proof}

\begin{lemma}
\label{lemma:ProxiesInhibitable}
The proxy events are inhibitable in $U^\varphi$.
\end{lemma}
\begin{proof}
Let $i\in \{0,\dots,m-1\}$ arbitrary but fixed. 
The event $p_i$ occurs once in the translator $T_{i,0}$ at the state $t_{i,1,2}$ and in $T_{i,2}$ at $t_{i,2,2}$.
The translator $T_i$ represents the clause $K_i=\{X_{a},X_{b},X_{c}\}$, where $a,b,c\in \{0,\dots,m-1\}$.
As $\varphi$ is cubic and monotone, there are exact two further clauses $K_j$ and $K_l$ such that $X_{b}\in K_j\cap K_l$.
Consequently, $X_{b}$ occurs exactly once in $T_{j}$ as well as in $T_{l}$.
By the uniqueness of $X_b$ in $T_j$ ($T_l$) and Lemma \ref{lemma:easyLiftedRegion} we need no further investigations at which state of $t_{j,0,1},t_{j,0,3},t_{j,1,1}$ ($t_{l,0,1},t_{l,0,3},t_{l,1,1}$) the event $X_b$ occurs a second respectively a third time in $U^\varphi$, for proving that the following table provides reliable regions of $U^\varphi$.

\begin{longtable}{ p{4cm}  p{2cm}p{2cm}p{3cm}  }
\emph{States} & \text{Exit} & \text{Enter} & \emph{Affected TS}  \\ \hline
$t_{i,0,0}\dots t_{i,0,2},t_{i,1,0},t_{i,1,1}$, $t_{i,1,3},t_{i,1,4},t_{i,2,0},t_{i,2,1}$, $t_{i,2,3},t_{i,2,4}$	& $p_i$	& $X_b,\tilde{X}_b$ 	& $T_{i,0},\dots,T_{i,2}$, $(T_{j},T_{l})$ \\ \hline
		$t_{i,0,3},\dots,t_{i,0,5}$												& $p_i$	& 			& $T_{i,1},T_{i,2}$
\end{longtable}
\end{proof}

\begin{lemma}
\label{lemma:VariablesInhibitable}
The variable events are inhibitable in $U^\varphi$.
\end{lemma}
\begin{proof}
Let $i\in \{0,\dots,m-1\}$ arbitrary but fixed. 
The translator $T_i$ represents the clause $K_i=\{X_{a},X_{b},X_{c}\}$, where $a,b,c\in \{0,\dots,m-1\}$.
As $\varphi$ is cubic and monotone, if $Z\in \{X_{a},X_{b},X_{c}\}$ there are exact two further clauses $K_j$ and $K_l$ such that $Z\in K_j\cap K_l$ and $j,l\in \{0,\dots,m-1\}$.
Consequently, $Z$ occurs exactly once in $T_{j}$ as well as in $T_{l}$.
We proceed as follows.
If $Z\in \{X_{a},X_{b},X_{c}\}$ is arbitrary but fixed and $T_{j}$ and $T_{l}$ are the translators where the second and third occurrence of $Z$ appear, then we show that $Z$ is inhibitable in $U^\varphi\setminus \{T_{j},T_{l}\}$ with regions of $U^\varphi$.
Then it follows directly, that $Z\in \{X_{a},X_{b},X_{c}\}$ is inhibitable in $U^\varphi$, regardless at which state of $t_{j,0,1},t_{j,0,3},t_{j,1,1}$ respectively $t_{l,0,1},t_{l,0,3},t_{l,1,1}$ the event $Z$ occurs it second or third time, respectively. 
Observe, that for all presented regions at most the event $Z$ is non-obeying in $T_j$ and $T_l$ and, considered $T_j$ ($T_l$) in isolation, it is unique.
Therefor, by Lemma~\ref{lemma:easyLiftedRegion} the following tables present reliable regions of $U^\varphi$.

Firstly, we show $X_a$, occurring in $T_{i,0}$ at state $t_{i,0,1}$, is inhibitable in $U^\varphi\setminus \{T_{j},T_{l}\}$ on the basis of regions of $U^\varphi$.

\begin{longtable}{ p{3cm}  p{1.5cm}p{2cm}p{3.5cm}  }
\emph{States} & \text{Exit} & \text{Enter} & \emph{Affected TS}  \\ \hline
$t_{i,0,0},t_{i,0,2},\dots,t_{i,0,5}$ 	& $X_a$	& $k_{18i+2}$ 	& $T_{i,0}, D_{6i},F_{6i}, T_{j},T_{l}$ \\ \hline
														& $X_a$	& 			& $T_{i,0}, T_{j},T_{l}$
\end{longtable}

Secondly, we show $X_b$, occurring in $T_{i,1}$ at state $t_{i,1,1}$, is inhibitable in $U^\varphi\setminus \{T_{j},T_{l}\}$ with regions of $U^\varphi$.

\begin{longtable}{ p{3cm}  p{1.5cm}p{2cm}p{2.5cm}  }
\emph{States} & \text{Exit} & \text{Enter} & \emph{Affected TS}  \\ \hline
$t_{i,1,0},t_{i,1,2},\dots,t_{i,1,4}$ 	& $X_b$	& $k_{18i+8}$ & $D_{6i+3}, F_{6i+3}$, $T_{i,1}, T_{j},T_{l}$\\ \hline
																& $X_b$	&					& $T_{i,1}, T_{j},T_{l}$
\end{longtable}

Finally, we show $X_c$, occurring in $T_{i,0}$ at state $t_{i,0,3}$, is inhibitable in $U^\varphi\setminus \{T_{j},T_{l}\}$ with regions of $U^\varphi$.

\begin{longtable}{ p{3cm}  p{1.5cm}p{2cm}p{2.5cm}  }
\emph{States} & \text{Exit} & \text{Enter} & \emph{Affected TS}  \\ \hline
$t_{i,0,0},\dots, t_{i,0,2}$, $t_{i,0,4},t_{i,0,5}$ 	& $X_c$	& $\tilde{X}_b$ & $T_{i,0}, T_{i,2},T_{a,j},T_{b,l}$\\ \hline
																& $X_c$	&					& $T_{i,0}, T_{j},T_{l}$
\end{longtable}
\end{proof}


\subsection{Proof of Lemma \ref{lem:Secondary_Proofs_for_the_Hardness_of_2grade_2-ESSP}}
\label{sec:Secondary_Proofs_for_the_Hardness_of_2grade_2-ESSP}

As in case of Lemma~\ref{lem:Secondary_Proofs_for_the_Hardness_of_2grade_2-ESSP} we have to refer to some regions, the tables of this subsection are additionally equipped with a column \textit{Name} where the name of the current region is listed.
Moreover, in the sequel, if the context is clear we refer with a little abuse of notation with $A$ to both the TS $A$ itself and the states of $A$ only.
In such a way we can refer easily to subsets of the states of $A$.
For example, by $H\setminus \{H_0\}$ we mean all states of the headmaster excluding all states belonging to the part $H_0$.

\begin{lemma}
\label{lemma:key}
The key event is inhibitable in $U^\varphi$.
\end{lemma}
\begin{proof}
First, we show that $k$ is inhibitable in $H$ by regions of $U^\varphi$.
Second, we represent a region of $U^\varphi$ where $k$ exits such that all non-obeying events occur only in $H$ inhibiting $k$ at all states of $U^\varphi\setminus H$.
Event $k$ is already inhibited at certain states by the key region.
By Figure \ref{fig:gadgets} the states left in $H$ for investigation are $h_{j,0},h_{j,4}$ for $j\in \{1,\dots,14m-1\}$.

\begin{figure}[H]
\centering
\begin{tabular}{ p{3cm}   p{3cm}    p{3cm}      p{3cm}  }
\textit{States} & \textit{Exit} & \textit{Enter} & \textit{Affected TS}  \\ \hline
Remaining states of $H$ except $h_{1,0}$. & $k, k_0,a_0,r_1$ & $z_0,v_1,w_1$ & $H,D_0$\\
$h_{1,0}$ & $k,r_0,k_1,a_1, k_3,r_2$, $w_4,w_5,k_2,k_8$, & $z_0,v_1,w_1,z_2,z_3$, $k_4,a_2,k_6,k_5$	& $H,D_0,D_1,D_2$, $T_0$\\	
$U^\varphi\setminus H$	& $k$ & $z_0,z_1$ & $H$							
\end{tabular}
\end{figure}
\end{proof}

\begin{lemma}
\label{lemma:keycopies}
The key event copies are inhibitable in $U^\varphi$.
\end{lemma}
\begin{proof}
Let $j\in \{0,\dots,14m-1\}$.
The following two regions show that $k_{3j}$ to is inhibitable in $U^\varphi$.
In both regions, the event $k_{3j+2}$ has a non-obeying signature and except for $D_j$ it might occur in a translator or a barter where, in both cases, it is unique.
Because of its uniqueness, in both cases, the appropriate region of the translator respectively the barter can be chosen such that $k_{3j+2}$ is the only non obeying event. 
That is, at all states of the translator respectively barter where $k_{3j}$ is not inhibited when $k_{3j+2}$ is exiting (first region) it is inhibited when $k_{3j+2}$ is entering (second region) simply by complementation. 
This justifies the third row of the table.
\begin{figure}[H]
\centering
\begin{tabular}{p{1.2cm} p{3.2cm}   p{2.5cm}    p{2.5cm}      p{2.7cm}  }
\textit{Nom} &\textit{States} & \textit{Exit} & \textit{Enter} & \textit{Affected TS}  \\ \hline
$R^{k_{3j}}_1$  & $H\setminus \{h_{j,4},h_{j,5},h_{j,0}\}$, $D_j\setminus \{d_{j,4}\}$ & $w_{2j},w_{2j+1},k_{3j}$, $r_{j+1}, k_{3j+2}$ & $v_{2j},v_{2j+1},r_{j}$ & $H,D_j$, a~single translator/barter \\
 $R^{k_{3j}}_2$ & $h_{j,4},h_{j,5},d_{j,4}$&  $r_{j-1},k_{3j},a_j$& $w_{2j},w_{2j+1},k_{3j+2}$, $r_j$ & $H,D_j$, a~single translator/barter \\	
 $R^{k_{3j}}_1, R^{k_{3j}}_2$ & states of the single translator/barter where $k_{3j+2}$ occur &  & &\\							
\end{tabular}
\end{figure}
The next table shows the inhibiting regions for $k_{3j+1}$.
Observe, that the key region inhibits $k_{3j+1}$ at all states of $H$ except for $h_{i,0},h_{i,4}, i\in \{0,\dots,14m-1\},i\not=j$.
\begin{figure}[H]
\centering
\begin{tabular}{p{1cm} p{3.5cm}   p{2.5cm}    p{2.5cm}      p{2.7cm}  }
\textit{Nom} &\textit{States} & \textit{Exit} & \textit{Enter} & \textit{Affected TS}  \\ \hline
$R_{3j+1}^1$&  remaining states in $H~\setminus\{h_{j-1,0},h_{j,0}\}$, $D_j$ &$z_{2j-2},z_{2j-1}$, $r_j,k_{3j+1}$ &  $r_{j-2},w_{2j-2},w_{2j-1}$  & $H,D_{j}$ \\
$R_{3j+1}^2$ & $h_{j-1,0},h_{j,0}$& $k_{3j+1},k_{3j+4}, k_{3j+2}$& $k_{3j},w_{2j},w_{3j},a_{j-1}$, $z_{2j+2},z_{2j+2},r_{j+1}$ & $H,D_j,D_{j+1}$, a~single translator/barter  \end{tabular}
\end{figure}
Proving that $k_{3j+2}$ is inhibitable in $U^\varphi$ let us face the special challenge to examine all possible cases where, instead of $D_j$, the event $k_{3j+2}$ can occur a second time.
Roughly said, $k_{3j+2}$ occurs its second time either in a translator or in a barter. 
In a translator it can occur either at the first or the last position.
These two cases require different investigations concerning the variable representers or the proxy events, respectively.
If a variable manifolder is affected, we have to distinguish different cases concerning the consistency events.
In a barter $B_n$ event $k_{3j+2}$ can occur at $b_{n,0}$ or $b_{n,2}$.
In both cases we have to observe different effects on the variable manifolders.
We firstly list all possible cases and give them names implied by the listing.
Then we show a table where for each possible case an appropriate region is presented.

There is an $i\in \{0,\dots,m-1\}$ such that
\begin{itemize}
\item[1.]  $k_{3j+2}$ occurs in $T_i$ at $t_{i,0,0}$ or $t_{i,1,0}$ or $t_{i,2,0}$.
\item[2.] $k_{3j+2}$ occurs in $T_i$ at $t_{i,0,4}$, and the variable representer at $t_{i,0,3}$ is $X^{i}_c$.
	\begin{itemize}
		\item[2.1] The second occurrence of the variable representer $X^{i}_c$ is at $x_{c,0}$.
		\item[2.2]The second occurrence of the variable representer $X^{i}_c$ is at $x_{c,1}$.
		\item[2.3]The second occurrence of the variable representer $X^{i}_c$ is at $x_{c,2}$.
	\end{itemize}
\item[3.] $k_{3j+2}$ occurs in $T_i$ 
	\begin{itemize}
		\item[3.1] at $t_{i,1,3}$. 
		\item[3.2] at $t_{i,2,3}$.
	\end{itemize}
\item[4.] There is an $n\in \{0,\dots,4m-1\}$  and there is an $i\in \{0,\dots,m-1\}$ such that $k_{3j+2}$ occurs in $B_{n}$
	\begin{itemize}
		\item[4.1]  at $b_{n,0}$ where $n=4i$ or $n=4i+2$.
		\item[4.2] at $b_{n,0}$ where $n=4i+1$ or $n=4i+3$. 
		\item[4.3]  at $b_{n,2}$  and $n=4i$ or $n=4i+2$.
		\item[4.4]  at $b_{n,2}$ and $n=4i+1$ or $n=4i+3$.
	\end{itemize}
\end{itemize}

Observe, that $k_{3j+2}$ is inhibited by the key region at all states of $H$ except for $h_{i,0},h_{i,4}$ for $i\in \{0,\dots,14m-1\}$.
It also inhibits $k_{3j+2}$ at the states $d_{j,0},d_{j,1},d_{j,3}$.
%
\begin{longtable}{p{0.9cm} p{3cm}   p{2.5cm}    p{2.5cm}      p{2.7cm}  }
\textit{Nom} &\textit{States} & \textit{Exit} & \textit{Enter} & \textit{Affected TS}  \\ \hline
$R_{1.}$ & Remaining~states~of $H\setminus~\{h_{j,4},h_{j,5},h_{j+1,0}\}$, $d_{j,0}$, $T_i$& $w_{2j},w_{2j+1},k_{3j}$, $r_{j},k_{3j+2}$ & $v_{2j},v_{2j+1},r_{j-1}$ & $H,D_j, T_i$\\
 $R_{2.1}$ & Remaining~states~of $H\setminus~\{h_{j,4},h_{j,5},h_{j+1,0}\}$, $d_{j,0},d_{j,1},d_{j,3}$, $T_i$& $w_{2j},w_{2j+1},k_{3j}$, $r_{j},k_{3j+2},c_{4c+2}$ & $v_{2j},v_{2j+1},r_{j-1}$, $X^{i}_c$ & $H,D_j, T_i, X_c$, $B_{4c+2}$\\
 $R_{2.2}$ & Remaining~states~of $H\setminus~\{h_{j,4},h_{j,5},h_{j+1,0}\}$, $d_{j,0},d_{j,1},d_{j,3}$, $T_i$& $w_{2j},w_{2j+1},k_{3j}$, $r_{j},k_{3j+2},c_{4c+3}$ & $v_{2j},v_{2j+1},r_{j-1}$, $x^{i}_c,c_{4c+2}$ & $H,D_j, T_i, X_c$, $B_{4c+2},B_{4c+3}$\\
 $R_{2.3}$ & Remaining~states~of $H\setminus~\{h_{j,4},h_{j,5},h_{j+1,0}\}$, $d_{j,0},d_{j,1},d_{j,3}$, $T_i$& $w_{2j},w_{2j+1},k_{3j}$, $r_{j},k_{3j+2}$ & $v_{2j},v_{2j+1},r_{j-1}$, $X^{i}_c,c_{4c+3}$ & $H,D_j, T_i, X_c$, $B_{4c+3}$\\
 $R_{3.1}$ & Remaining~states~of $H\setminus~\{h_{j,4},h_{j,5},h_{j+1,0}\}$, $d_{j,0},d_{j,1},d_{j,3}$,  $T_i\setminus \{t_{i,2,3},t_{i,2,4}\}$& $w_{2j},w_{2j+1},k_{3j}$, $r_{j},k_{3j+2}$ & $v_{2j},v_{2j+1},r_{j-1}$, $p_i$ & $H,D_j, T_i$\\
 $R_{3.1}$ & Remaining~states~of $H\setminus~\{h_{j,4},h_{j,5},h_{j+1,0}\}$, $d_{j,0},d_{j,1},d_{j,3}$,  $T_i\setminus \{t_{i,1,3},t_{i,1,4}\}$& $w_{2j},w_{2j+1},k_{3j}$, $r_{j},k_{3j+2}$ & $v_{2j},v_{2j+1},r_{j-1}$, $p_i$ & $H,D_j, T_i$\\
  $R_{4.1}$ & Remaining~states~of $H\setminus~\{h_{j,4},h_{j,5},h_{j+1,0}\}$, $d_{j,0},d_{j,1},d_{j,3}$, $B_{n}$ & $w_{2j},w_{2j+1},k_{3j},r_j$, $k_{3j+2},c_{4i},c_{4i+2}$ & $v_{2j},v_{2j+1},r_{j-1}$ & $H,D_j,X_i,B_{4i}$, $B_{4i+2}$\\
  $R_{4.2}$ & Remaining~states~of $H\setminus~\{h_{j,4},h_{j,5},h_{j+1,0}\}$, $d_{j,0},d_{j,1},d_{j,3}$, $B_{n}$ & $w_{2j},w_{2j+1},k_{3j},r_j$, $k_{3j+2},c_{4i+1},c_{4i+3}$ & $v_{2j},v_{2j+1},r_{j-1}$ & $H,D_j,X_i,B_{4i+1}$, $B_{4i+3}$\\
 $R_{4.3}$ & Remaining~states~of $H\setminus~\{h_{j,4},h_{j,5},h_{j+1,0}\}$, $d_{j,0},d_{j,1},d_{j,3}$, $B_{n}$ & $w_{2j},w_{2j+1},k_{3j},r_j$, $k_{3j+2}$ & $v_{2j},v_{2j+1},r_{j-1}$, $c_{4i},c_{4i+2}$ & $H,D_j,X_i,B_{4i}$, $B_{4i+2}$\\
  $R_{4.4}$ & Remaining~states~of $H\setminus~\{h_{j,4},h_{j,5},h_{j+1,0}\}$, $d_{j,0},d_{j,1},d_{j,3}$, $B_{n}$ & $w_{2j},w_{2j+1},k_{3j},r_j$, $k_{3j+2}$ & $v_{2j},v_{2j+1},r_{j-1}$, $c_{4i+1},c_{4i+3}$ & $H,D_j,X_i,B_{4i+1}$, $B_{4i+3}$\\
   $R_{1./2./3.}$ & $h_{j,4},h_{j,5},h_{j+1,0},d_{j,2}$, $t_{i,2,3},t_{i,2,4}$~respectiv. $t_{i,1,3},t_{i,1,4}$ & $k_{3j+1},k_{3j+2},k_{3j+4}$ & $a_j,k_{3j},w_{2j+2}$, $w_{2j+3},z_{2j+4}$, $z_{2j+5}, r_{j+2}$ & $H,D_j,D_{j+1},T_i$\\	
$R_{4.}$ & $h_{j,4},h_{j,5},h_{j+1,0},d_{j,2}$  & $k_{3j+1},k_{3j+2},k_{3j+4}$ & $a_j,k_{3j},w_{2j+2}$, $w_{2j+3},z_{2j+4}$, $z_{2j+5}, r_{j+2}$ & $H,D_j,D_{j+1},B_n$\\	\end{longtable}
\end{proof}

\begin{lemma}
\label{lemma:zero}
The zero events are inhibitable in $U^\varphi$.
\end{lemma}
\begin{proof}
Because of the two occurrences of $k$ in $H_0$ we have to investigate the events $z_0,z_1$ in isolation.
The region $R_{z_0\&z_1}$ of the following table shows that $z_0,z_1$ are inhibitable in $U^\varphi$ except for the states $h_{0,6},h_{0,7}$.
The region $R_{z_0}$ ($R_{z_1}$) proves $z_0$ ($z_1$) to be inhibitable at $h_{0,7}$ ($h_{0,6}$) .
\begin{longtable}{p{0.7cm} p{3cm}   p{2.5cm}    p{2.5cm}      p{2.7cm}  }
\textit{Nom} &\textit{States} & \textit{Exit} & \textit{Enter} & \textit{Affected TS}  \\ \hline
$R_{z_0\&z_1}$ & $H\setminus \{h_{0,6},h_{0,7}\}$& $z_0,z_1$ & $k$ & $H$ \\
 $R_{z_0}$& $h_{0,7}$& $z_0,v_1,r_0$ & $w_0,k_2$ & $H,D_0,T_0$ \\	
$R_{z_1}$ & $h_{0,6}$& $r_0,v_0,z_1,z_2,v_3$  & $w_1,k_0,w_2,k_5$ & $H,D_0,D_1,T_0$\\							
\end{longtable}

For $j\in \{1,\ldots,14m-1\}$ we decide the two cases $z_{2j}$ and $z_{2j+1}$.
The region $R_z$ proves that both of $z_{2j},z_{2j+1}$ are inhibitable at the stated states.
The region $R_{z_{2j}}$ ($R_{z_{2j+1}}$) proves $z_{2j}$ ($z_{2j+1}$) to be inhibitable at the remaining states.
We assume $ i\in \{0,\dots,j-1\}$.
\centering
\begin{longtable}{p{0.9cm} p{3cm}   p{2.5cm}    p{2.5cm}      p{2.7cm}  }
\textit{Nom} &\textit{States} & \textit{Exit} & \textit{Enter} & \textit{Affected TS}  \\ \hline
$R_z$ & $H_j \setminus \{h_{j,6},h_{j,7}\}$, $H_{j+1},\dots, H_{14m-1}$ & $z_{2j},z_{2j+1}$, $v_{2i},z_{2i+1},w_{2i}$, $w_{2i+1}$ & $k,k_{2i},k_{2i+1}$, $k_{2j},k_{2j+1}$  & $H,D_i$,  \\
$R_{2j}$ & $h_{j,7}$, $H_0,\dots,H_{j-1}$& $z_{2j},r_{j+1},w_{2j+1}$, $k_{3j+3}$ & $r_{j-1},v_{2j}$ & $H,D_{j+1}$ \\
$R_{2j+1}$ & $h_{j,7}$, $H_0,\dots,H_{j-1}$  & $r_j,z_{2j+1},w_{2j+1}$, $k_{3j+2}$ & $r_{j-1},v_{2j+1}$ & $H,D_j$, a~single transl./barter \\						
\end{longtable}
\end{proof}

\begin{lemma}
\label{lemma:wildcards}
The wildcard events are inhibitable in $U^\varphi$.
\end{lemma}
\begin{proof}
Because of $k$ in $H_0$ we have to treat $w_0,w_1$ separately.
Region $R_{w_0,w_1}$ inhibits $w_0,w_1$ simultaneously at the listed states.
Region $R_{w_0}$ ($R_{w_0}$) concerns $w_0$ ($w_1$) only.
\begin{longtable}{p{0.9cm} p{3cm}   p{2.5cm}    p{2.5cm}      p{2.7cm}  }
\textit{Nom} &\textit{States} & \textit{Exit} & \textit{Enter} & \textit{Affected TS}  \\ \hline
$R_{w_0,w_1}$ & $H\setminus \{h_{0,4},h_{1,0}\}$, $d_{0,0},d_{0,1},d_{0,3}$ & $w_0,w_1,k_0,r_1$ &  $v_0,v_1,r_0$& $H,D_0$ \\
 $R_{w_0}$ & $h_{0,4},h_{1,0},d_{0,2}$ & $v_0,z_1,w_0$, $z_2,z_3$ & $ k,k_0,k_1 $& $H,D_0$ \\	
 $R_{w_1}$& $h_{0,4},h_{1,0}$& $z_0,v_1,w_1$ & $k$&  $H$\\							
\end{longtable}
%
For $j\in \{1,\ldots,14m-1\}$ we decide the two cases $w_{2j}$ and $w_{2j+1}$.
The region $R_w$ proves that both of $w_{2j},w_{2j+1}$ are inhibitable at the listed states.
The region $R_{w_{2j}}$ ($R_{w_{2j+1}}$) proves $w_{2j}$ ($w_{2j+1}$) to be inhibitable at the remaining states.
We assume $ i\in \{0,\dots,j-1\}$.
%
\begin{longtable}{p{0.7cm} p{3cm}   p{2.5cm}    p{2.5cm}      p{2.7cm}  }
\textit{Nom} &\textit{States} & \textit{Exit} & \textit{Enter} & \textit{Affected Segments}  \\ \hline
$R_w$& $H\setminus\{h_{j,4},h_{j+1,0}\}$, $d_{j,0}$, $d_{j,1},d_{j,3}$ & $w_{2j},w_{2j+1},k_{3j}$ & $v_{2j},v_{2j+1},r_j$ & $H,D_j$, a~single transl./barter\\
 $R_{w_{2j}}$& $h_{j,4},h_{j+1,0}$, $d_{j,2}$ & $r_0,w_{2i},w_{2i+1}$ $r_j,v_{2j},z_{2j+1}$, $w_{2j},z_{2j+2},z_{2j+3}$  & $k_{3i+1},r_{j-1}$, $ k_{3j-2},k_{3j},k_{3j+1}$ & $H,D_0,\dots,D_j$ \\	
$R_{w_{2j+1}}$ &$h_{j,4},h_{j+1,0}$, $d_{j,2}$& $r_0,w_{2i},w_{2i+1}$, $r_j,z_{2j},v_{2j+1},w_{2j+1}$ & $k_{3i+1},r_{j-1}$, $k_{3j-2},a_j,k_{3j}$& $H,D_0,\dots,D_j$ \\	\end{longtable}
\end{proof}

\begin{lemma}
\label{lemma:vice}
The vice events are inhibitable in $U^\varphi$.
\end{lemma}
\begin{proof}
For $j\in \{0,\dots, 14m-1\}$ we investigate the two cases $v_{2j}$ and $v_{2j+1}$ separately.
The following table proves $v_{2j}$ to be inhibitable in $U^\varphi$.
%
\begin{longtable}{p{0.5cm} p{3.9cm}   p{2.5cm}    p{2.7cm}      p{2cm}  }
\textit{Nom} &\textit{States} & \textit{Exit} & \textit{Enter} & \textit{Affected TS}  \\ \hline
$R_{2j}^1$ & $H_0,\dots,H_{j-1}, D_0,\dots,D_{j-1}$, $H_j\setminus\{h_{j,0},h_{j,1},h_{j,7}\}$, \newline$H_{j+1}\setminus\{h_{j+1,1},h_{j+1,6},h_{j+1,7}\}$, $d_{j,2}, D_{j+1}\setminus\{d_{j+1,0}\}$, $H_{j+2}\setminus\{h_{j+2,0},\ldots, h_{j+2,4}\}$, $H_{j+3},\dots,H_{14m-1}$, $D_{j+2},\dots,D_{14m-1}$ & $v_{2j},z_{2j+1},r_j,z_{2j+2}$, $z_{2j+3},k_{3j+4},r_{j+2}$ & $r_{j-1},w_{2j+1},k_{3j}$, $w_{2j+2},w_{2j+3},r_{j+1}$ & $H,D_{j},D_{j+1}$ \\
 $R_{2j}^2$ & $h_{j,0},h_{j,1}, D_j,d_{j+1,0}$, $h_{j+1,1},h_{j+1,6},h_{j+1,7}$, $h_{j+2,0},\ldots$, $h_{j+2,4}$& $r_{j-1},v_{2j},a_j$, $w_{2j+2},\ldots,w_{28m-1}$  & $z_{2j},w_{2j+1},k_{3j+2}$, $k_{3i+1}$, \newline $i\in\{j+1,\dots,14m-1\}$ & $H$, $D_j,\dots,D_{14m-1}$, single trans./ barter \\	
 $R_{2j}^3$ & $h_{j,7}$ & $v_{2j},v_{2j+1},r_j$, $z_{2j}, z_{2j+1}$  & $r_{j-1},k_{3j},k_{3j+1}$, $k_{3j+2}$ &  $H,D_j$, single trans./barter \\						\end{longtable}
%
The following table proves $v_{2j+1}$ to be inhibitable in $U^\varphi$.
%
\begin{longtable}{p{0.7cm} p{3cm}   p{2.5cm}    p{2.5cm}      p{2.7cm}  }
\textit{Nom} &\textit{States} & \textit{Exit} & \textit{Enter} & \textit{Affected TS}  \\ \hline
$R_{2j+1}^1$ & $H_0,\dots,H_{j-1}$, $H_{j+1},\dots,H_{14m-1}$, $H\setminus\{h_{j,0},h_{j,1},h_{j,6}\}$& $r_j,z_{2j},v_{2j+1}$ & $r_{j-1},w_{2j}, k_{3j+2}$& $H, D_j$, single trans./barter\\
 $R_{2j+1}^2$ & $h_{j,6}$& $r_j,v_{2j},v_{2j+1}$, $z_{z_j+2},z_{2j+3}$  & $r_{j-1},k_{3j},k_{3j+1}$, $k_{3j+2}$ &  $H,D_j$,  single trans./barter\\	
 $R_{2j+1}^3$& $h_{j,0}, D_j\setminus \{d_{j,3}\}$ &  $z_{2j},v_{2j+1}$, $v_{2i},z_{2i+1},w_{2i}$, $i\in \{0,\dots,j-1\}$ & $k,a_j,k_{3j}$, $k_{3i},k_{3i+1}$, $i\in \{0,\dots,j-1\}$  & $H,D_0,\dots,D_j$\\
 $R_{2j+1}^4$ & $h_{j,1}$ & $r_{j-1},v_{2j+1}$ & $ r_j,z_{2j+1},w_{2j},k_{3j+2}$ & $H,D_j$, single transl./barter\\							
\end{longtable}
\end{proof}

\begin{lemma}
\label{lemma:reachability}
The reachability events are inhibitable in $U^\varphi$.
\end{lemma}
\begin{proof}
Because of the special role of $k$ we have to treat the cases $r_0$ and $r_j,j\in \{1,\dots,14m-1\}$ separately.
The next table proves $r_0$ to be inhibitable in $U^\varphi$.
%
\begin{longtable}{p{0.7cm} p{3cm}   p{2.5cm}    p{2.5cm}      p{2.7cm}  }
\textit{Nom} &\textit{States} & \textit{Exit} & \textit{Enter} & \textit{Affected TS}  \\ \hline
$R_{r_0}^1$ &  $h_{0,2},h_{0,4},h_{0,5},h_{0,7}$, $h_{0,8},H_1,\ldots,H_{14m-1}$, $t_{0,0,0}$ & $r_0,z_0,v_1$ & $w_0,k_2$ & $H,D_0,T_{0,0}$ \\
 $R_{r_0}^2$& $h_{0,1}, h_{0,3}, h_{0,6}$, $D_0\setminus\{d_{0,0},d_{0,1}\}$, $t_{0,0,1},\ldots,t_{0,0,5}$ & $r_0,k,k_1,k_2$ & $z_0,v_1,w_1,z_2,z_3,r_1$ & $H,D_0,T_{0,0}$ \\	
 $R_{r_0}^3$ & $d_{0,0},d_{0,1}$&  $r_0,v_0,v_1,z_2,z_3$ & $k_0,k_1,k_2$ & $H,D_0,T_{0,0}$\\							
\end{longtable}
%
The next table proves $r_j$ to be inhibitable in $U^\varphi$ where $j\in \{0,\dots,14m-1\}$.
The event $k_{3j-1}$ has in both cases a non obeying signature.
But in the first case it exits and in the second case it enters. 
As it is unique in the TS where it occurs the second time the construction admits to chose $k_{3j+2}$ to be the only non obeying event in the TS.
Therefor, by complementation we can assure that $r_j$ is inhibitable at the appropriate TS.
%
\begin{longtable}{p{0.7cm} p{3cm}   p{2.5cm}    p{2.5cm}      p{2.7cm}  }
\textit{Nom} &\textit{States} & \textit{Exit} & \textit{Enter} & \textit{Affected TS}  \\ \hline
$R_{r_j}^1$ & $H\setminus H_{j-1}$, $H_{j-1}\setminus\{h_{j-1,4},h_{j-1,5}\}$, $d_{j-1,0},d_{j-1,1},d_{j-1,3}$ & $w_{2j-2},w_{2j-1},r_j$, $k_{3j-3},k_{3j-1}$ & $v_{2j-2},v_{2j-1},r_{j-1} $ & $H,D_{j-1}$, single transl./barter \\
 $R_{r_j}^2$& $h_{j-1,4},h_{j-1,5}$, $d_{j-1,0},d_{j-1,2}$ & $r_{j-2},v_{2j-2},a_{j-1}$, $r_j,v_{2j}, w_{2j}, z_{2j+1}$, $z_{2j+2},z_{2j+3}$ & $ z_{2j-2},w_{2j-1},r_{j-1}$, $k_{3j-2},k_{3j-1}$, $k_{3j},k_{3j+1}$ & $H,D_{j-1},D_j$						
\end{longtable}
\end{proof}

\begin{lemma}
\label{lemma:accordance}
The accordance events are inhibitable in $U^\varphi$.
\end{lemma}
\begin{proof}
The following table proves $a_j$ to be inhibitable in $U^\varphi$.
 In $R_{a_j}^4$ the key copy $k_{3j-1}$ exits and $k_{3j+2}$ enters. 
 Both them can occur a second time in a translator or a barter.
 If they occur in different TSs the set of states of these TSs can be chosen in a way that $k_{3j-1}$ and $k_{3j+2}$ are the only non obeying events.
 A special case occur when both of them are in the same TS.
For avoiding a tedious detailed case analyses we give a brief justification for $R_{a_j}^3$  being a region in that case.
If  $k_{3j-1},k_{3j+2}$ occurin the same TS, then, by construction, they occur in a barter $B_n$ since no consecutive \glqq free to use\grqq\ copies occur in the same translator.
Again by construction, $k_{3j-1}$ occur at $b_{n,0}$ and $k_{3j+2}$ occur at $b_{n,2}$.
Consequently, by the signature of these events, the event $c_n$ must have a negative signature.
Now there is an $i\in \{0,\dots,m-1\}$ such that $n\in \{4i ,4i+1,4i+2,4i+3\}$.
If $n\in \{4i,4i+2\}$ we chose both $c_{4i},c_{4i+2}$ as the only exiting events of $X_i$ which assures that no variable representer is concerned.
It is easy to see that the choice of the signature of the second consistency event, which is not $c_n$, can be done such that it is the only non obeying event in the barter where it occur.
Similarly, if $n\in \{4i+1,4i+3\}$ we chose both of them as the only exiting events of $X_i$ which surely is possible by construction.
Again, the second consistency event is the only non obeying event in the barter where it occur.
Therefor, the regions of all the concerned barters and the variable manifolder can be chosen compatible with the signature given in $R_{a_j}^4$.
%
\begin{longtable}{p{0.7cm} p{3cm}   p{2.5cm}    p{2.5cm}      p{2.7cm}  }
\textit{Nom} &\textit{States} & \textit{Exit} & \textit{Enter} & \textit{Affected TS}  \\ \hline
$R_{a_j}^1$ & $H_j\setminus\{h_{j,6},h_{j,7}\}$, $H_{j+1}\setminus\{h_{j+1,0}\}$, $H_{j+2},\dots H_{14m-1}$ & $r_{j-1},r_{j+2},a_j$, $k_{3j}$ & $r_j,w_{2j},w_{2j+1}$, $k_{3j+2}$ & $H,D_j$, single transl./barter  \\
 $R_{a_j}^2$ & $h_{j,6},h_{j,7}$ & $r_{j-2},k_{3j-4},a_j$, $r_{j+1}, k_{3j}$, $w_{2i},w_{2i+1}$, \scalebox{0.8}{$i\in \{j+1,\dots,14m-1\}$}  & $w_{2j-2},w_{2j-1},z_{2j}$, $z_{2j+1}, r_{j}, k_{3i+1}$, \scalebox{0.8}{$i\in \{j+1,\dots,14m-1\}$}& $H,D_{j-1},\dots,D_{14m-1}$ \\	
 $R_{a_j}^3$  & $H_0,\dots,H_{j-1}$& $k_{3j-2},a_j,r_{j+1}$, $k_{3j},k_{3j-1}$, possibly $c_{4i},c_{4i+2}$ or $c_{4i+1},c_{4i+3}$ &  $k_{3j-3},w_{2j},w_{2j+1}$,  $a_{j-1},r_j,k_{3j+2}$& $H,D_{j-1},D_j$, some transl./barters/ var. manifld.\\	
 $R_{a_j}^4$ & $h_{j+1,0},D_j$& $r_{j-1},v_{2j},a_j$, $w_{2i},w_{2i+1}$, \scalebox{0.8}{$i\in \{j+1,\dots,14m-1\}$} & $z_{2j},w_{2j+1},k_{3j+2}$, $k_{3i+1}$, \scalebox{0.8}{$i\in \{j+1,\dots,14m-1\}$} & $H,D_j,\dots,D_{14m-1}$, single transl./barter \\ 		
 $R_{a_j}^5$ & sink of $k_{3j+2}$ & $r_{j-1},k_{3j-2},a_j$, $r_{j+1},k_{3j}$& $z_{2j},v_{2j+1},w_{2j+1},r_j$&	$H,D_j$			
\end{longtable}
\end{proof}

\begin{lemma}
\label{lemma:consistency}
The consistency events are inhibitable in $U^\varphi$.
\end{lemma}
\begin{proof}
Before we show the consistency events to be inhibitable, we remind in which way the indices of the parts of the headmaster, the duplicators, the translators and the barters are connected.
The duplicator $D_j$ forces with the aid of the events $v_{2j},v_{2j+1},w_{2j}$, occurring a second time in $H_j$, the key copies $k_{3j},\ldots,k_{3j+2}$ to exit where $k_{3j+2}$ is for free use elsewhere.
For each clause $K_i\in \{K_0,\ldots,K_{m-1}\}$ the concerning union $T_i$ needs six key copies.
That is, for $i\in \{0,\dots,m-1\}$ the duplicators $D_{6i},\dots,D_{6i+5}$  provides the six key copies $k_{3(6i)+2},k_{3(6i+1)+2},\dots,k_{3(6i+5)+2}$ for $T_{i}$ with the aid of the vice events from $H_{6i},\dots,H_{6i+5}$.
Each of the $4m$ barters $B_0,\dots,B_{4m-1}$ provides exactly one of the consistency events $c_0,\dots,c_{4m-1}$ with the aid of $2$ key copies.
That is, the duplicator $D_{6(m-1)+6}=D_{6m}$ provides the first key copy for free use in a barter.
More exactly, with $y=6m$, by our aforementioned considerations we have that the duplicators $D_{y+2j}$ and $D_{y+2j+1}$ provides the $2$ key copies $k_{3(y+2j)+2}$ and $k_{3(y+2j+1)+2}$ for the barter $B_j$ with the aid of the vice events provided by $H_{y+2j}$ and $H_{y+2j+1}$.
Eventually, for each $c_j$, where $j\in \{0,\ldots,4m-1\}$, there is a unique $i\in \{0,\dots,m-1\}$ and a unique $l\in \{0,\dots,3\}$such that $j=4i+l$.
It depends on $l$ at which state of $X_i$ the event $c_j$ occurs.
With $R_{4i},\dots, R_{4i+3}$ we refer to these cases.
%
\begin{longtable}{p{0.7cm} p{3cm}   p{2.5cm}    p{2.5cm}      p{2.7cm}  }
\textit{Nom} &\textit{States} & \textit{Exit} & \textit{Enter} & \textit{Affected TSs}  \\ \hline
$R_{4i}^1$ & $X_i,B_{4i}$, $d_{y+8i,0},\dots, d_{y+8i,3}$ & $c_{4i},X_i^\alpha$, $r_{y+8i},z_{2(y+8i)+1}$, $w_{2(y+8i)}$, $k_{3(y+8i)+2}$ & $r_{y+8i-1}$, $v_{2(y+8i)+1}$& $H, D_{y+8i},B_{4i},X_i$, $T_\alpha$\\
$R_{4i}^2$ & $d_{y+8i,4}$, $T_\alpha$ & $c_{4i},c_{4i+2}$& & $B_{4i},B_{4i+2}$ \\
$R_{4i+1}^1$ & $X_i,B_{4i+1}$, $d_{y+8i+2,0},\dots, d_{y+8i+2,3}$ & $c_{4i+1},X_i^\beta$, $r_{y+8i+2}$, $z_{2(y+8i+2)+1}$, $w_{2(y+8i+2)}$, $k_{3(y+8i+2)+2}$ & $c_{4i},r_{y+8i+1}$, $v_{2(y+8i+2)+1}$& $H, D_{y+8i+2},B_{4i+1}$, $X_i$, $T_\beta$ \\
$R_{4i+1}^2$ &  $d_{y+8i+2,4}$, $T_\beta$ & $c_{4i+1},c_{4i+3}$ & & $X_i,B_{4i+1},B_{4i+3}$\\
$R_{4i+2}^1$ & $X_i,B_{4i+2}$, $d_{y+8i+4,0},\dots, d_{y+8i+4,3}$ & $c_{4i+2}$, $r_{y+8i+4}$, $z_{2(y+8i+4)+1}$, $w_{2(y+8i+4)}$, $k_{3(y+8i+4)+2}$ & $r_{y+8i+3}$, $v_{2(y+8i+4)+1},X_i^\alpha$ & $H, D_{y+8i+4},B_{4i+2}$, $X_i$, $T_\alpha$ \\
$R_{4i+2}^2$ & $d_{y+8i+4,4}$, $T_\alpha$ & $c_{4i},c_{4i+2}$& & $B_{4i},B_{4i+2}$ \\
$R_{4i+3}^1$ & $X_i,B_{4i+3}$, $d_{y+8i+6,0},\dots, d_{y+8i+6,3}$ & $c_{4i+3}$, $r_{y+8i+6}$, $z_{2(y+8i+6)+1}$, $w_{2(y+8i+6)}$, $k_{3(y+8i+6)+2}$ & $c_{4i+2}$,$r_{y+8i+5}$, $v_{2(y+8i+6)+1}$, $X_{i}^\beta$ & $H, D_{y+8i+6},X_i$, $B_{4i+2},B_{4i+3}$, $T_\beta$ \\						
$R_{4i+3}^2$ &  $d_{y+8i+6,4}$, $T_\beta$ & $c_{4i+1},c_{4i+3}$ & & $X_i,B_{4i+1},B_{4i+3}$
\end{longtable}
\end{proof}

\begin{lemma}
\label{lemma:variable}
The variable events are inhibitable in $U^\varphi$.
\end{lemma}
\begin{proof}
On the one hand we have to distinguish between $X_i^\alpha,X_i^\beta,X_i^\gamma$, where $i\in \{0,\dots,m-1\}$, and on the other hand, we have to investigate different cases arising by the different possible positions at which a variable representer can occur.
By definition, $T_\alpha$ is the translator where the variable representer $X_a^\alpha$ occur.
We show the inhibition of $X_i^\alpha$ in $T_\alpha$ for every position where it can occur, that is, either at $t_{\alpha,0,1}$ ($R_{t_{\alpha,0,1}}$), at $t_{\alpha,0,3}$ ($R_{t_{\alpha,0,3}}$) or at $t_{\alpha,1,1}$ ($R_{t_{\alpha,1,1}}$).
Observe that, as each variable occur in exactly three clauses, $\alpha,\beta,\gamma$ are pairwise different.
Concerning the inhibition in their translators, the proof can be done in an absolutely similar way for $X_i^\beta,X_i^\gamma$.
Hence, we omit these proofs.
That each of $X_i^\alpha,X_i^\beta,X_i^\gamma$ is inhibitable at all other states, at which inhibition isn't done by one of the regions of $R_{t_{\alpha,0,1}},R_{t_{\alpha,1,1}},R_{t_{\alpha,0,3}}$,  the inhibition of each of $X_i^\alpha,X_i^\beta,X_i^\gamma$ is shown explicitly by the regions $R_{X_i^\alpha},R_{X_i^\beta},R_{X_i^\gamma} $.

\begin{longtable}{p{0.7cm} p{3cm}   p{2.5cm}    p{2.5cm}      p{2.7cm}  }
\textit{Nom} &\textit{States} & \textit{Exit} & \textit{Enter} & \textit{Affected TSs}  \\ \hline
$R_{t_{\alpha,0,1}}$ & $T_{\alpha,0}$ & $X_i^\alpha,X_i^\beta,X_i^\gamma$, $r_{6\alpha-1},v_{12\alpha+1}$ & $z_{12\alpha+1}, k_{18\alpha+2}$, $r_{6\alpha},w_{12\alpha}$ & $X_i, T_{\alpha,0},T_\beta,T_\gamma$, $D_{6\alpha},H$  \\
$R_{t_{\alpha,1,1}}$ & $T_{\alpha,1}$ & $X_i^\alpha,X_i^\beta,X_i^\gamma$, $r_{6\alpha},v_{12\alpha+3}$ & $z_{12\alpha+3}, k_{18\alpha+5}$, $r_{6\alpha+1}$, $w_{12\alpha+2}$ & $X_i, T_{\alpha,0},T_\beta,T_\gamma$, $D_{6\alpha+1},H$  \\
 $R_{t_{\alpha,0,3}}$& $T_{\alpha,0}$   & $X_i^\alpha,X_i^\beta,X_i^\gamma$ &  $\tilde{X}_b^\alpha$  & $X_i,T_{\alpha,0},T_{\alpha,2}$, $T_\beta,T_\gamma$   \\	
 $R_{X_i^\alpha} $ & remaining &  $X_i^\alpha,c_{4i}$ &  & $X_i,B_{4i}, T_\alpha$\\	
 $R_{X_i^\beta} $ & remaining &  $X_i^\beta,c_{4i+1}$ & $c_{4i}$ & $X_i,B_{4i},B_{4i+1}$, $T_\beta$\\	
 $R_{X_i^\gamma} $ & remaining &  $X_i^\gamma$ &  $c_{4i+1}$ & $X_i,B_{4i+1},T_\gamma$
\end{longtable}

\end{proof}

\begin{lemma}
\label{lemma:proxylocum}
The proxy events are inhibitable and the event $\tilde{X}^i_b$ is inhibitable in $U^\varphi$ for each $i \in \{0,\dots,m-1\}$.
\end{lemma}
\begin{proof}
We show that $\tilde{X}^i_b$ is inhibitable in $U^\varphi$ where $i \in \{0,\dots,m-1\}$.
Let $X_a^i$ be the variable representer occurring at $t_{i,0,1}$ and $X_a^j,X_a^l$ be the second and third representation of the variable $X_a$.
\begin{figure}[H]
\centering
\begin{longtable}{p{0.7cm} p{3cm}   p{2.5cm}    p{2.5cm}      p{2.7cm}  }
\textit{Nom} &\textit{States} & \textit{Exit} & \textit{Enter} & \textit{Affected TSs}  \\ \hline
$R_{\tilde{X}^i_b}^1$ & $T_i$ & $\tilde{X}^i_b,r_{6i+1}$, $v_{12i+5}$ & $z_{12i+5},w_{12i+4}$, $X_a^i,X_a^j,X_a^l$, $k_{18i+8}$, $r_{6i+2}$ & $H,D_{6i+2},X_a$, $T_i,T_j,T_l$ \\
 $R_{\tilde{X}^i_b}^2$ & remaining & $\tilde{X}^i_b$ & & $T_i$ \\						
\end{longtable}
\end{figure}
We show that $p_i$ is inhibitable in $U^\varphi$ where $i\in \{0,\dots,m-1\}$.
Let $X_b^{i}$ be the variable representer occurring at $t_{i,1,1}$  and $X_b^j,X_b^l$ be the second and third representation of the variable $X_b$.
\begin{figure}[H]
\centering
\begin{tabular}{p{0.7cm} p{3cm}   p{2.5cm}    p{2.5cm}      p{2.7cm}  }
\textit{Nom} &\textit{States} & \textit{Exit} & \textit{Enter} & \textit{Affected TSs} \\ \hline
$R_{p_i}^1$ & $T_i\setminus \{t_{0,0,3},\dots,t_{0,0,5}\}$ & $p_i,r_{6i}$, $v_{12i+3}$ & $\tilde{X}^i_b, X^i_b, X^j_b, X^l_b$, $z_{12i+3},w_{12i+2}$, $k_{18i+5}$, $r_{6i+2}$ & $H,D_{6i+1}, X_b$, $T_i, T_j,T_l$ \\
 $R_{p_i}^2$ & remaining & $p_i$ & & $T_i$ \\						
\end{tabular}
\end{figure}
\end{proof}
%


\subsection{Proof of Lemma \ref{lem:alwaysSSP}}

\begin{proof}
The observation of the following features for each transition system $A$ in $U^\varphi\setminus H$ will help us to see easily, that each SSP instance $(s,s')$ of $U^\varphi\setminus H$ is solvable in $U^\varphi$ if $U^\varphi$ has the ESSP:
\begin{itemize}
	\item Each event occurring in $A$ is unique in $A$.
	\item There is at most one state in $A$ having no successor.
\end{itemize}
Now let $s,s'$ be two states of $A$, a TS of $U^\varphi\setminus H$.
The second feature assures that at least one of $s,s'$ is the source of an event $e$.
Without lost of generality we assume $s\edge{e}$.
By the first feature we have that $\neg(s'\edge{e})$.
As $U^\varphi$ has the ESSP, we have that $e$ is inhibitable at $s'$ by a region of $U^\varphi$.
That is, there is a region $R$ of $U^\varphi$ such that $sig_R(e)=-1$, implying $R(s)=1$, and $R(s')=0$, witnessing the separability of $s$ and $s'$.
Hence, it only remains to proof with regions of $U^\varphi$ that $H$ has the SSP, too. 
We observe that for each $j\in \{0,\dots,14m-1\}$ the events $r_j,v_{2j},v_{2j+1},w_{2j},w_{2j+1},a_j$ are unique in $H$ and that $z_{2j}$ and $z_{2j+1}$ occur only in the part $H_j$ of $H$.
Moreover, the state $h_{14m-1,8}$ is the only state of $H$ without successor.
With the same arguments stated before we have that each state of $s$ of $H\setminus \{h_{j,1}, h_{j,6},h_{j,7}\}$ is separable from each state of $H\setminus\{s\}$ and, furthermore, each of $\{h_{j,1}, h_{j,6},h_{j,7}\}$ is additionally separably from each state of $H\setminus H_j$.

Consequently, it remains to show, that $\{h_{j,1}, h_{j,6},h_{j,7}\}$ are pairwise separably.
Section~\ref{lem:Secondary_Proofs_for_the_Hardness_of_2grade_2-ESSP} yields us appropriate regions:
For all $j\in \{0,\dots,14m-1\}$ the region $R_{z_{2j}}$ separates the states $h_{j,1},h_{0,7}$ and $h_{j,6}, h_{j,7}$ and the region $R_{z_{2j+1}}$ separates $h_{j,1}$ and $h_{j,6}$.
\end{proof}


\end{document}